\documentclass[letterpaper, 11pt]{article}
\usepackage[margin=0.86in]{geometry}

\usepackage{amsmath,amssymb,amsthm}
\usepackage{authblk}
\usepackage{dsfont}
\usepackage{cancel}
\usepackage{enumerate}
\usepackage{graphicx}
\usepackage{subfig}
\usepackage{braket}
\usepackage{mathtools}
\usepackage{etoolbox}
\usepackage{cite}
\usepackage[titletoc, title]{appendix}
\usepackage{caption}
\usepackage{tikz}
\usepackage{tikz-qtree}
\usetikzlibrary{patterns}

\usepackage[titletoc]{appendix}
\usepackage{sectsty}
\sectionfont{\fontsize{14}{15}\selectfont}
\subsectionfont{\fontsize{11}{15}\selectfont}
\usepackage[colorlinks=true, citecolor=blue, linkcolor=blue, urlcolor=blue, hypertexnames=false]{hyperref}

\graphicspath{{./figures/}}
\numberwithin{equation}{section}
\numberwithin{figure}{section}

\newcommand{\ketbra}[1]{|#1\rangle\langle#1|}
\newcommand{\ketbrat}[2]{|#1\rangle\langle#2|}
\newcommand{\vect}[1]{\boldsymbol{#1}}
\renewcommand{\Re}{\text{Re}}

\DeclareMathOperator{\Tr}{Tr}
\DeclareMathOperator{\poly}{poly}
\DeclareMathOperator{\dist}{dist}
\DeclareMathOperator{\diam}{diam}
\DeclareMathOperator{\spn}{span}
\DeclareMathOperator{\sgn}{sign}
\DeclareMathOperator{\rank}{rank}
\DeclareMathOperator{\erf}{erf}
\newcommand{\anc}{\textnormal{anc}}
\newcommand{\clock}{\textnormal{clock}}
\newcommand{\I}{\mathcal{I}}
\newcommand{\Id}{\mathds{1}}
\renewcommand{\H}{\mathcal{H}}
\newcommand{\V}{\mathcal{V}}
\newcommand{\E}{\mathcal{E}}
\renewcommand{\S}{\mathcal{S}}
\newcommand{\A}{\mathcal{A}}
\renewcommand{\O}{\mathcal{O}}
\newcommand{\C}{\mathds{C}}
\newcommand{\J}{\mathcal{J}}

\newcommand{\eff}{\textnormal{eff}}

\newcommand{\U}{\mathcal{U}}
\newcommand{\Hcirc}{\tilde{H}_\oneone^\textnormal{circuit}}

\newcommand{\circuit}{\textnormal{circuit}}
\newcommand{\gadget}{\textnormal{gadget}}
\newcommand{\els}{\textnormal{else}}
\newcommand{\targ}{\textnormal{target}}
\newcommand{\eps}{\mathcal{E}}
\renewcommand{\L}{\mathcal{L}}
\newcommand{\rest}{\textnormal{rest}}
\newcommand{\ideal}{\textnormal{ideal}}
\newcommand{\local}{\textnormal{local}}

\newcommand{\qedextra}{\hfill\ensuremath{\blacklozenge}}

\newcommand{\VC}{\textsc{VertexCover}}
\newcommand{\G}{\mathcal{G}}

\newcommand{\coNP}{\textnormal{\texttt{coNP}}}
\newcommand{\NPpoly}{\textnormal{\texttt{NP/poly}}}

\newtheorem{lemma}{Lemma}
\newtheorem{defn}{Definition}
\newtheorem{thm}{Theorem}

\newtheorem{prop}{Proposition}
\newtheorem{coro}{Corollary}
\newtheorem{claim}{Claim}

\newcommand{\oneone}{A}
\newcommand{\dicke}{B}

\let\oldproofname=\proofname
\renewcommand{\proofname}{\textbf{\oldproofname}}

\usepackage{soul}

\newcommand{\delete}[1]{}

\title{Hamiltonian Sparsification and Gap-Simulations}

\author[1]{Dorit Aharonov\thanks{dorit.aharonov@gmail.com}}
\author[2]{Leo Zhou\thanks{leozhou@g.harvard.edu}}
\affil[1]{\small School of Computer Science and Engineering, The Hebrew University, Jerusalem 91904, Israel}
\affil[2]{\small Department of Physics, Harvard University, Cambridge, MA 02138, USA}

\date{\today}

\begin{document}

\maketitle
\thispagestyle{empty}

\begin{abstract}

Analog quantum simulations---simulations of one Hamiltonian by another---is one of the major goals in the noisy intermediate-scale quantum computation (NISQ) era, and has many applications in quantum complexity.
We initiate the rigorous study of the physical resources required for such simulations, where we focus on the task of {\it Hamiltonian sparsification}.
The goal is to find a simulating Hamiltonian $\tilde{H}$ whose underlying  interaction graph has bounded degree (this is called {\it degree-reduction}) or much fewer edges than that of the original Hamiltonian $H$ (this is called {\it dilution}).
We set this study in a relaxed framework for analog simulations that we call {\it gap-simulation}, where $\tilde{H}$ is only required to simulate the groundstate(s) and spectral gap of $H$ instead of its full spectrum, and we believe it is of independent interest.

Our main result is a proof that in stark contrast to the classical setting, general degree-reduction is {\it impossible} in the quantum world, even under our relaxed notion of gap-simulation.
The impossibility proof relies on devising counterexample Hamiltonians and applying a strengthened variant of Hastings-Koma decay of correlations theorem.
We also show a complementary result where degree-reduction is possible when the strength of interactions is allowed to grow polynomially.
Furthermore, we prove the impossibility of the related sparsification task of generic Hamiltonian dilution, under a computational hardness assumption.
We also clarify the (currently weak) implications of our results to the question of quantum PCP.
Our work provides basic answers to many of the ``first questions'' one would ask about Hamiltonian sparsification and gap-simulation;
we hope this serves a good starting point for future research of these topics.

\end{abstract}

\clearpage

\newpage

\thispagestyle{empty}

\tableofcontents

\thispagestyle{empty}

\newpage

\setcounter{page}{1}

\section{Introduction}

A major theme in quantum computation is the idea of {\it analog quantum simulation}.
This is the task of simulating one Hamiltonian $H$ by another Hamiltonian $\tilde{H}$,
which might be more readily or easily implemented.
In fact, this goal was identified as a main motivation for realizing quantum computers as early as 1981 by Feynman\cite{Feynman1982},
with the idea that such analog quantum simulations can shed light on
properties of physical quantum systems that are hard
to simulate efficiently on classical computers.
Cirac and Zoller~\cite{CiracZollerNatPhys2012}
further developed this idea, and explained that
such simulators are likely to be achievable well
before fully fault-tolerant quantum computation
\cite{AharonovFaultTolerance, KnillThreshold, KitaevFaultTolerantAnyon}
becomes practical, which might take a long time.
While fault-tolerant quantum computers when realized can be used to apply {\it digital} quantum simulations~\cite{LloydDigitalSimulation1996}
(where a quantum {\it circuit}
simulates the time-evolution $e^{-iHt}$ under a
local Hamiltonian $H$),
{\it analog} quantum simulations are more accessible for
near-term experiments because they do not require full-fledged quantum computer.
Many groups are designing implementations
in a variety of experimental platforms\cite{SimonOpticalLatticeSim2011,Bloch2012,Blatt2012,Aspuru-Guzik2012,Houck2012,GeorgescuQuantumSimulationReview2014}, and we have recently seen
some experiments in intermediate-sized quantum systems in regimes where classical simulations are difficult~\cite{Bernien2017,Zhang2017}.
It has been argued that analog quantum simulation constitutes
one of the more interesting challenges in the
{\it noisy intermediate-scale quantum computing}
(NISQ) era \cite{Preskill2018}.

Beyond their natural physical applications,
analog simulations of Hamiltonians
are also very important
for quantum complexity theory.
For example, in the theory of quantum NP,
one is often interested in {\it reducing} problems defined by one class of
Hamiltonians to another
(e.g.~\cite{KSV02, quantumNPsurvey, BravyiHastingsSim,UniversalHamiltonian,KKR06}).
These reductions are often derived using {\it perturbative
gadgets} (e.g.~\cite{KKR06, OliveiraTerhal, BDLT08, JordanGadgets,
CaoImprovedOTGadget, CaoNagajGadget}).
Moreover, analog Hamiltonian simulators might also be useful
for the design of Hamiltonian-based quantum algorithms,
such as adiabatic algorithm~\cite{FarhiAdiabatic2000} and QAOA~\cite{QAOA}.
In those settings, it is often desirable to
tailor the Hamiltonians being used,
while maintaining the properties essential for the algorithm.

In this paper, we initiate the rigorous study of the
minimal resources required to simulate a given target Hamiltonian,
and ask:
When can we simulate a Hamiltonian $H$ by another $\tilde{H}$ that is {\it simpler}, easier, or more economic to implement?
Of course, this vaguely stated question can take several forms if made
more rigorous;
here we  focus on a natural goal which we loosely call {\it Hamiltonian sparsification},
which aims to simplify the {\it interaction graph} of the Hamiltonian.
For a $2$-local $n$-qubit Hamiltonian $H$, the interaction graph has $n$ vertices,
with edges connecting any pairs of qubits
that participate in a local term in $H$.
For a $k$-local Hamiltonian, we consider an interaction {\it hyper}graph, where
each term acting on $k$ qubits is represented by a hyperedge.
A generic $k$-local Hamiltonian has $\Theta(n^k)$ edges, and $\Theta(n^{k-1})$ degree per
vertex.
Roughly speaking, Hamiltonian sparsification
aims to simulate a Hamiltonian using another whose
interaction graph is more ``economic'', e.g., it has less edges (we refer to this as {\it dilution}) or
its degree is bounded (we refer to this as {\it degree-reduction}).

Hamiltonian sparsification has several important motivations.
First, it can help physicists tackle the immense hurdles they face when trying to realize Hamiltonians in the lab.
In addition, in many settings in
quantum complexity, such as in the study of quantum PCP \cite{qPCPsurvey}
and recent approaches to the area law question \cite{LocalTestOfEntanglement}, simulating a Hamiltonian by one with constant degree or fewer edges is a potentially important primitive.
Indeed, sparsification is used ubiquitously in classical computer science,
in a variety of applications;
we mention two important ones.
The first, graph sparsification (and more generally, matrix sparsification) is a central tool in matrix algorithms \cite{STNearlyLinearTimeAlg, SpielmanSpectralSparsify, SpielmanEffectiveResistance, BSST13}.
Famously, Ref.~\cite{BSS12} proved that any graph can be replaced by another which is sparse (namely, has small degree on average), such that their Laplacian matrices are spectrally similar.
Another common use of sparsification in classical computer science is {\it degree-reduction} (DR), used in the study of
local Constraint Satisfaction Problems (CSPs) and PCPs~\cite{dinur}.
We believe that this natural and basic notion deserves to be studied in the
quantum context as well, and might have interesting applications beyond
those we can foresee today.

\subsection{Gap-Simulations: Simulating only the low-lying part of the spectrum}

Before embarking on the study of Hamiltonian sparsification,
we first need an appropriate definition of analog simulation.
The study of analog Hamiltonian simulation was set on rigorous footing
in a recent work by Cubitt, Montanaro, and Piddock~\cite{UniversalHamiltonian};
their definition refines that of Bravyi and Hastings~\cite{BravyiHastingsSim},
and it roughly goes as follows:
A given Hamiltonian $H$ is simulated by ``encoding'' its full spectrum into the low-lying part of the
spectrum of $\tilde{H}$ acting on a larger Hilbert space.
When $\tilde{H}$ is implemented, then the low-lying part of its spectrum
can be used to derive properties and information about the original Hamiltonian $H$.
For obvious reasons, we will refer to this definition as
{\it full-spectrum simulation}.
In Ref.~\cite{UniversalHamiltonian}, the notion of {\it universal} Hamiltonians
was defined and studied: these are families of Hamiltonians
which are capable of performing full-spectrum simulations of
{\it any} given Hamiltonian, albeit generally with exponential overhead in energy.

While this strong notion of {\it full-spectrum simulation} is necessary for simulating all dynamical properties of a system,
it is common in physics that one is only interested in the properties of the low-energy states and, particularly,
the groundstates.
In addition, the spectral gap separating the groundstates from the rest of the spectrum is an intimately related quantity that is usually physically important.
For example, the groundstates encode exotic quantum behaviors such as topological order, and the spectral gap protects them~\cite{Wen1990,LevinWenTopo}.
Also, they are used together to define quantum phases of matter and characterize phase transitions \cite{SachdevQPT, UndecidabilityOfGap}.
Moreover, both are the main objects of interest in quantum computational complexity:
In quantum adiabatic algorithms~\cite{FarhiAdiabatic2000}, the goal is to prepare
a groundstate of a problem Hamiltonian, and the spectral gap governs the efficiency of the process.
In quantum NP theory \cite{quantumNPsurvey}, only the groundstate(s) of the Hamiltonian matters as it is the witness for the problem.
The spectral gap also determines the temperature of a thermal equilibrium (Gibbs) state that can be used to approximate the groundstate.
Hence, we believe that a natural and minimal notion of analog Hamiltonian simulation, which is still meaningful for many physical contexts, should require that both the space of groundstates and the spectral gap above it be preserved.

Therefore,
we suggest to consider sparsification,
or more generally Hamiltonian simulation,
using this minimal notion,
which we formally define as {\it gap-simulation}.
To the best of our knowledge, despite its naturalness, this
relaxed notion of Hamiltonian simulation was not formally
defined and rigorously studied previously
in the quantum complexity literature.

A Hamiltonian $\tilde{H}$ is said to \emph{gap-simulate} $H$ if it mimics
the groundstate(s) and the spectral gap of $H$; no constraints are imposed
on the excited part of the spectrum.
To provide a sensible definition requires some care,
since in the quantum world we can allow inaccuracies and
entanglement to an ancilla.
We provide two versions of the definition:
In the weaker one (Def.~\ref{defn:hamsimul-incoherent}),
the groundspace is mimicked {\it faithfully},
i.e. the {\it support} of any groundstate
of $\tilde{H}$, when reduced to the original Hilbert space,
is close to the groundspace of $H$.
However, this definition does
not require quantum {\it coherence} within the groundspace be maintained.
Such coherence is guaranteed by our stronger definition (Def.~\ref{defn:hamsimul}),
in which all superpositions within the groundspace are simulated.
The extent to which the gap-simulation is incoherent (or unfaithful) is quantified via a small constant $\epsilon$ (or $\delta$).
It seems that the coherent notion
is the ``correct'' one for most quantum applications,
though the weaker one might also be useful in certain contexts (see Sec.~\ref{sec:discussion}).
We mention that here, like in Ref.~\cite{BravyiHastingsSim, UniversalHamiltonian},
we allow encoding of the qubits.
Typically, we consider ``localized'' encodings, though this is not explicitly required in the definition.

To set the stage, some basic results about the framework are provided: We show in Lemma \ref{lem:equiv} that
for Hamiltonians with unique groundstates, our two definitions of gap-simulations coincide.
Moreover, both coherent and incoherent gap-simulation definitions are
shown to be stable under compositions.

How does the gap-simulation framework compare with the
stricter definitions of full-spectrum simulations developed in
Ref.~\cite{BravyiHastingsSim,UniversalHamiltonian}?
In Appendix~\ref{sec:comp-defns}, this connection is discussed formally;
roughly,
our definition is indeed a relaxed version of full-spectrum simulations whose spectral error is smaller than the spectral gap, up to varying restrictions about encoding.
We choose to work here with the more relaxed definition of
gap-simulation, since impossibility results for a weaker definition are
of course stronger. More generally,
it seems that this framework is an important and relevant one to consider
in physics and quantum complexity contexts.
Being less demanding, gap-simulation is likely achievable in certain cases where
full-spectrum simulation is difficult or even impossible.

\vspace{-10pt}
\subsection{Main Results}
Equipped with this framework of Hamiltonian sparsification via
gap-simulations, we ask: When are sparsifications
possible in the quantum world?
It is conceivable that, like in the analogous classical settings mentioned above \cite{dinur,BSST13}, they ought to be always possible.
The main result of in this paper (Theorem \ref{thm:main}) shows that in stark contrast to the classical setting,
both coherent and incoherent degree-reductions are not generally possible in the quantum world, even if one uses the relaxed notion of gap-simulation.
This impossibility phenomenon is due to the existence of many-body entanglement in some quantum groundstates; we show, using a strengthened version of Hastings-Koma decay of correlation theorem \cite{HastingsKoma}, that there exist local Hamiltonians whose groundstates cannot be coherently mapped into the groundspace of a gapped Hamiltonian with constant degree.
Though one might suspect this is a consequence of degeneracy in the groundspace, we show that it holds even in the case of a unique groundstate.
We believe this is a surprising and curious phenomenon, which demonstrates
the richness in the subject, and highlights the difference in the resources required for classical versus quantum Hamiltonian simulation.

This impossibility result on degree-reduction is essentially tight, as we provide a complementary result (Theorem~\ref{thm:degree-reduction-poly}) based on a somewhat sophisticated application of the circuit-to-Hamiltonian construction, stating that  degree-reduction becomes possible for any local Hamiltonian with non-negligible spectral gap, when polynomially large overhead in interaction strength is allowed.

We also study a related important sparsification task: dilution.
While our main result Theorem~\ref{thm:main} is an information-theoretic result that rules out {\it existence} of degree-reducers regardless of computational power, we are unable to provide such a strong result in the case of dilution.
Information-theoretically, we can only rule out dilution with perfect (or inverse-polynomially close to perfect) coherence (Theorem~\ref{thm:imposs1-dilute}).
Nevertheless, we are able to prove impossibility of any efficient classical algorithm to find diluters with constant unfaithfulness, for generic (even classical) Hamiltonians (Theorem~\ref{thm:imposs-dilute}).
The proof of this theorem (relying on Ref.~\cite{DellvanMelkebeek}) works under the assumption that $\coNP \not\subseteq \NPpoly$ (alternatively, the polynomial hierarchy does not collapse to its third level).
Although generic constructive dilution is ruled out by our Theorem~\ref{thm:imposs-dilute}, the question of existence of diluters for general Hamiltonian, with bounded or large interaction strengths, remains open.

The paper provides quite a few further results
complementing the above-mentioned main contributions.
These build on ideas in classical PCP reductions and perturbative gadgets.
In addition, the ideas studied here are strongly reminiscent of questions arising in the context of the major open problem of quantum PCP~\cite{qPCPsurvey}.
We clarify this connection and provide some preliminary results along these lines.

We believe that
the study of the resources required for
Hamiltonian simulations in various contexts,
as well as the framework of gap-simulation, are
of potential deep interest for physics as well as
quantum complexity.
The questions raised touch upon
a variety of important challenges, from
quantum simulations, to algorithm design, to quantum PCP and NP reductions,
to physical implementations on near-term quantum processors,
and more.
Their study might also shed light on questions in many-body physics, by developing tools
to construct ``equivalent'' Hamiltonians, from the point of view
of the study of groundstate physics.
The discussion in Sec.~\ref{sec:discussion} includes a more detailed list
of some of the open questions and implications.

\subsection{Overview}
In Sec.~\ref{sec:set-the-stage}, we set the stage by providing definitions of gap-simulation and sparsification, and proving basic facts about this new framework.
In Sec.~\ref{sec:results}, we state our results formally.
Subsequently, Sec.~\ref{sec:proofs-overview} provides elaborated and intuitive proof sketches, and Sec.~\ref{sec:discussion} provides further discussion.
All technical proofs are deferred to the appendices.

\section{Definition of the Framework: Setting the Stage\label{sec:set-the-stage}}
\vspace{-5pt}

\subsection{Gap-Simulations of Hamiltonians\label{sec:gap-simulation}}

We restrict our attention to $k$-local Hamiltonians $H=\sum_{i=1}^M H_i$ acting
on $n$ qu$d$its (with internal states $\{\ket{0},\ldots,\ket{d-1}\}$), where each term $H_i$ acts nontrivially on a (distinct)
subset of at most $k$ qudits.
We denote $\lambda_j(X)$ as the $j$-th lowest eigenvalue of $X$, and $\|X\|$ as the spectral norm of $X$.
In addition, for any Hermitian projector $P$, we
denote $P^\perp \equiv \Id-P$, and $\ket{\psi}\in P \Longleftrightarrow P\ket{\psi} = \ket{\psi}$.

\begin{defn}[groundspace, energy spread and gap]
\label{defn:gap}
Consider a family of $n$-qudit Hamiltonians $\{H_{(n)}\}_{n=1}^\infty$.
Let $E_n^g = \lambda_1(H_{(n)})$, and suppose $P_{(n)}$ is a Hermitian projector onto the subspace of eigenstates of $H_{(n)}$ with energy  $\le E_n^g + w_n\gamma_n$,
for some $\gamma_n > 0$, $0\le w_n < 1$, such that
\begin{gather}
[H_{(n)}, P_{(n)}] = 0, \quad \|P_{(n)} (H_{(n)}-E_n^g) P_{(n)} \|\le w_n \gamma_n, \nonumber \\
\textnormal{and} \quad
\lambda_j(P^\perp_{(n)} (H_{(n)} -E_n^g ) P^\perp_{(n)} + \gamma_n P_{(n)}) \ge \gamma_n \quad \forall j.
\end{gather}
We call the subspace onto which $P_{(n)}$ projects a
\emph{quasi-groundspace}, $w_n$ its \emph{energy spread},
and $\gamma_n$ its \emph{quasi-spectral gap}.
When we choose $w_n = 0$ and $\gamma_n = \min_j \{\lambda_j(H_{(n)})-E_n^g: \lambda_j(H_{(n)})\neq E_n^g\}$, we call the quasi-groundspace that $P_{(n)}$ projects onto simply \emph{the groundspace} of $H_{(n)}$, and $\gamma_n$ \emph{the spectral gap} of $H_{(n)}$.
Let $w_\infty = \sup_n w_n$ and $\gamma_\infty = \inf_n \gamma_n$.
If $\gamma_\infty > 0$ and $w_\infty<1$, we say  $\{H_{(n)}\}_{n=1}^\infty$ is \emph{spectrally gapped}.
\end{defn}

Below, we omit the subscript $n$ in $H_{(n)}$, referring to a single
$H$, with the usual implicit understanding that we consider
families of Hamiltonians, where $n\to\infty$.
All explicit Hamiltonians we gap-simulate here
have $w_n=0$, but Definition \ref{defn:gap} is more general and
allows $w_n>0$, so that it can capture situations with slightly
perturbed groundstates (or when simulating a larger
low-energy part of the spectrum). We now define Hamiltonian
gap-simulation, visualized in Fig.~\ref{fig:gapsimul}:

\begin{defn}[gap-simulation of Hamiltonian]
\label{defn:hamsimul}
Let $H$ and $\tilde{H}$ be two Hamiltonians, defined on Hilbert spaces $\H$ and $\tilde{\H}$ respectively, where $\tilde{\H}$.
Let $V: \H\otimes \H_\anc  \to \tilde{\H}$ be an isometry ($V^\dag V=\Id$), where $\H_\anc$ is some ancilla Hilbert space.
Denote $\tilde{E}^g \equiv \lambda_1(\tilde{H})$.
Per Definition~\ref{defn:gap}, let $P$ be a quasi-groundspace projector of $H$, $\gamma$ its quasi-spectral gap.
We say that $\tilde{H}$ \emph{gap-simulates}
$(H,P)$ with \emph{encoding} $V$, \emph{incoherence} $\epsilon\ge 0$ and \emph{energy spread} $0\le\tilde{w}<1$ if the following conditions are both satisfied:
\begin{enumerate}
\item There exists a Hermitian projector $\tilde{P}$ projecting onto a
subspace of eigenstates of $\tilde{H}$ such that
    \begin{gather}
    [\tilde{H}, \tilde{P}]= 0, \quad \|\tilde{P}(\tilde{H} -\tilde{E}^g)\tilde{P}\|\le \tilde{w}\gamma, \quad \textnormal{and} \quad
    \lambda_j(\tilde{P}^\perp (\tilde{H} - \tilde{E}^g )\tilde{P}^\perp + \gamma \tilde{P}) \ge \gamma \quad \forall j.
     \label{eq:strongsimul}
    \end{gather}
	I.e., $\tilde{P}$ projects onto a
 quasi-groundspace of $\tilde{H}$ with quasi-spectral gap not smaller
than that of $P$ in $H$, and energy spread $\tilde{w}$.
\item There exists a Hermitian projector $P_\anc$ acting on $\H_\anc$, so that
\end{enumerate}
\begin{flalign} \label{eq:incoherence}
\textnormal{[bounded incoherence]} &&
\|\tilde{P} - V(P\otimes P_\anc)V^\dag \| \le \epsilon &&
\phantom{\textnormal{(incoherence)}}
\end{flalign}
When $P$ projects onto the groundspace of $H$, rather than a
quasi-groundspace,
we usually do not mention $P$ explicitly, and simply
say that $\tilde{H}$ \emph{gap-simulates} $H$.
\end{defn}

\begin{figure}[h]
\centering
\includegraphics[height=3.5cm]{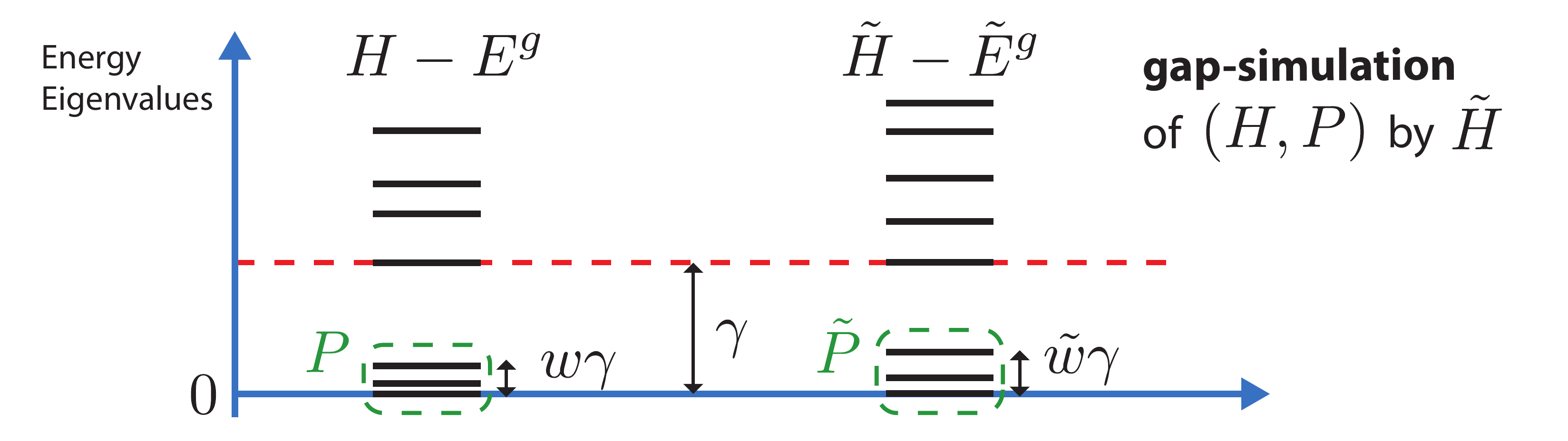}
\caption{\label{fig:gapsimul}Visualizing gap-simulation of Hamiltonian $(H, P)$ by $\tilde{H}$. If $\|\tilde{P}-V(P\otimes P_\anc)V^\dag\|\le \epsilon$, for some isometry $V$, then this is a coherent gap-simulation with $\epsilon$-incoherence.
If $\|\tilde{P} - V(P\otimes \Id_\anc)V^\dag \tilde{P}\|\le \delta$, then this is an incoherent but faithful gap-simulation with $\delta$-unfaithfulness.}
\vspace{-5pt}
\end{figure}

Requiring $\epsilon$ from Eq.~\eqref{eq:incoherence} be small
ensures that
\emph{coherence} in the groundspace is maintained by the gap-simulation.
This is illustrated by considering a Hamiltonian $H$ with two orthogonal
groundstates $\ket{g_1}$ and $\ket{g_2}$.
The condition of Eq.~\eqref{eq:incoherence} essentially says that for any coherent superposition $\ket{g}=c_1\ket{g_1}+c_2\ket{g_2}$, and a state $\ket{a}\in P_\anc$ on the ancilla,
there exists a ground state of $\tilde{H}$ that looks like $\ket{\tilde{g}}= V \ket{g}\otimes\ket{a} + \O(\epsilon)$. Moreover, any groundstate of $\tilde{H}$
could be written in this form.
This would preserve the expectation value of any observable in the groundspace, i.e.
$\braket{g|\hat{\sigma}|g}\approx \braket{\tilde{g}|V\hat{\sigma}V^\dag|\tilde{g}} + \O(\epsilon)$.
In contrast, one can consider an alternative situation where the groundspace of a simulator $\tilde{H}$ is spanned by states of the form $\ket{\tilde{g}_i'} \approx V \ket{g_i}\otimes\ket{a_i}$, where $\braket{a_i|a_j}\ll 1$.
This situation remains interesting,
as finding a ground state $\ket{\tilde{g}'_i}$ of $\tilde{H}$ reveals information about a ground state of $H$ by decoding:
$\ketbra{g_i} \approx \Tr_\anc(V^\dag\ketbra{\tilde{g}'_i} V)$.
However,
the coherence among groundstates is destroyed,
since $\ket{g}=\ket{g_1}+\ket{g_2}$ is mapped to $\ket{\tilde{g}'}\approx V(\ket{g_1}\otimes\ket{a_1}+\ket{g_2}\otimes\ket{a_2})$, and observables
such as $\hat{\sigma}=\ketbrat{g_1}{g_2}$ are not preserved: $\braket{g|\hat{\sigma}|g}\not\approx\braket{\tilde{g}'|V\hat{\sigma}V^\dag|\tilde{g}'}$.

Although coherence seems important to maintain in most quantum settings,
we also define {\it incoherent} gap-simulation, which
may be relevant in some situations (see discussion in Sec.~\ref{sec:discussion}).

\begin{defn}[incoherent gap-simulation]
\label{defn:hamsimul-incoherent}
Consider two Hamiltonians $H$ and $\tilde{H}$, $P$ a quasi-groundspace projector of $H$, and $V$ some isometry in the same setting as in Definition~\ref{defn:hamsimul}.
We say that $\tilde{H}$ \emph{incoherently gap-simulates} $(H,P)$ with \emph{encoding} $V$, \emph{unfaithfulness} $\delta \ge 0$ and energy spread $0\le \tilde{w}<1$ if it satisfies the first condition of Definition~\ref{defn:hamsimul} and, instead of the second condition of Eq.~\eqref{eq:incoherence},
\begin{flalign}\label{eq:unfaithfulness}
\textnormal{[bounded unfaithfulness]}&&
\|\tilde{P} - V(P\otimes \Id_\anc)V^\dag \tilde{P}\|\le \delta
&&
\phantom{\textnormal{(unfaithfulness)}}
\end{flalign}
Again, when $P$ projects onto the groundspace of $H$, we simply say $\tilde{H}$ \emph{incoherently gap-simulates} $H$.
\end{defn}

Small unfaithfulness essentially
means that the {\it support} of the vectors in the groundspace of
$\tilde{H}$
is roughly contained in an subspace spanned by encoding the groundspace of $H$ with some ancilla.

It is easy to see that small incoherence implies small unfaithfulness, namely $\delta\le 2\epsilon$ (see Appendix \ref{sec:uniqueGS}).
However, small unfaithfulness is a strictly weaker condition
than small incoherence; we will see an example in
Prop.~\ref{prop:incoherent-tree}.
Importantly, when $H$ has a {\it unique}
groundstate, the two notions are equivalent up to a constant
(the proof of this fact is perhaps surprisingly not entirely trivial; see Appendix \ref{sec:uniqueGS}):

\begin{lemma}[incoherent gap-simulation is coherent when groundstate is unique]
\label{lem:equiv}
Suppose $H$ has a unique groundstate, with
groundspace projector $P=\ketbra{g}$.
If $\tilde{H}$ incoherently gap-simulates $H$ with unfaithfulness
$\delta <1$, then it also gap-simulates $H$ with incoherence
$\epsilon \le \sqrt{2}\delta/\sqrt{1-\delta^2}$.
\end{lemma}

While we do not explicitly restrict the form of encoding $V$ in the above definitions, we need to specify them for the impossibility proofs, where we will consider localized encoding:

\begin{defn}[localized encoding]
\label{defn:localized-encoding}
Consider a (possibly incoherent) gap-simulation of $H$ by $\tilde{H}$ encoded by an isometry $V:\H\otimes \H_\anc \to \tilde{\H}$.
Let $\H\otimes \H_\anc =\bigotimes_{i=1}^n(\H_i\otimes \A_i)$, where $\A_i$ is the $i$-th ancilla subsystem;
also let $\tilde{\H}=\bigotimes_{i=1}^m \tilde{\H}_i$, $m\ge n$.
We say $V$ is a \emph{localized encoding} if either of the following is true:
\begin{enumerate}
\item $V=\bigotimes_{i=1}^n V_i$, where $V_i:\H_i\otimes \A_i\to \tilde{\H}_i$, and $\tilde{\H}_i$ consists of $O(1)$ qudits in $\tilde{H}$ for $i=1,\ldots,n$.
\item $V$ is a constant-depth quantum circuit: $V=\prod_{a=1}^D U_a$, where $D=O(1)$, $U_a = \bigotimes_{\mu} U_{a,\mu}$, and $U_{a,\mu}$ is a unitary operator acting on $O(1)$ number of qudits.
\end{enumerate}
We say $V$ is an \emph{$\eta$-localized encoding} if there exists a localized encoding $V_L$ such that $\|V-V_L\|\le \eta$.
\end{defn}

In addition to constant-depth quantum circuits, any quantum error-correcting code where each logical qudit is encoded as $O(1)$ qudits is also a localized encoding.
Note it is easy to see that if a gap-simulation has $\eta$-localized encoding $V$ and incoherence $\epsilon$ (or unfaithfulness $\delta$), it is also a gap-simulation with localized encoding $V_L$ and incoherence $\epsilon'\le\epsilon+2\eta$ (or unfaithfulness $\delta'\le \delta+2\eta$).
Hence, we usually restrict our attention to fully localized encoding in the remainder of the paper.

It is also fairly straightforward to show that compositions of gap-simulation behave intuitively:
\begin{lemma}[Composition]
\label{lem:composition}
Suppose $H_1$ (incoherently)  gap-simulates $(H_0,P_0)$ with encoding $V_1$, incoherence $\epsilon_1$ (or unfaithfulness $\delta_1$), energy spread $\tilde{w}_1$, and a corresponding quasi-groundspace projector $P_1$.
Also suppose $H_2$ (incoherently) gap-simulates $(H_1, P_1)$ with encoding $V_2$, incoherence $\epsilon_2$ (or unfaithfulness $\delta_2$), and energy spread $\tilde{w}_2$.
Then $H_2$ (incoherently) gap-simulates $(H_0,P_0)$ with encoding $V_2 (V_1\otimes \Id_{\anc,1})$, incoherence $\le\epsilon_2+\epsilon_1$ (or unfaithfulness $\le 2\delta_2+\delta_1$), and energy spread $\tilde{w}_2$.
\end{lemma}

\subsection{Hamiltonian Sparsification: Degree-Reduction and Dilution}
\vspace{-5pt}

We define here the set of parameters of interest when considering minimizing
resources in gap-simulations:
\begin{enumerate}
\item $k$ -- locality of individual Hamiltonian term; typically $O(1)$
in physical systems, but we parametrize it to allow minimization,
as well as to allow $O(\log^a n)$-local Hamiltonians, for some constant $a$.
\item  $r$ -- maximum degree of Hamiltonian, the main objective in degree-reduction.
\item $M$ -- number of terms in the Hamiltonian, the main objective in dilution.
\item $J$ -- the interaction strength of individual Hamiltonian terms.
This is typically restricted to $O(1)$ in physical systems, but allowing it to grow with $n$
leads to more possibilities of gap-simulators.
Equivalently, a gap-simulator with $J$ growing with $n$ can be converted to one that simulates the original Hamiltonian
but has a vanishing gap if we restrict to bounded-strength Hamiltonian terms.
\item $\epsilon$ and $\delta$ -- incoherence $\epsilon$ and unfaithfulness $\delta$ that capture how well the Hamiltonian gap-simulates the original Hamiltonian in terms of groundspace projectors.
\item  $\tilde{w}$ -- energy spread in the gap-simulator Hamiltonian; allowing it to be different from the original Hamiltonian gives more freedom in
  gap-simulations.
\end{enumerate}

We will use the notation of $[r,M,J]$-gap-simulator
to indicate that the maximum degree is $r$, the number of local terms is
$M$, and for each term $\tilde{H}_i$ we have $\|\tilde{H}_i\|\le J$.
We define:

\vspace{-5pt}

{~}

\begin{defn}[Degree-reduction (DR) and dilution]
Let $\tilde{H}$ be a $k$-local $[r,M,J]$-gap-simulator of $H$ with $\epsilon$-incoherence (or $\delta$-unfaithfulness) and energy spread $\tilde{w}$.
Additionally suppose $H=\sum_{i=1}^{M_0} H_i$ is a sum of $M_0=M_0(n)$ terms,
each of which is $O(1)$-local.
Then
\begin{itemize}
\item We call $\tilde{H}$ an $[r,M,J]$-\emph{degree-reducer} of $H$ if $r=O(1)$.
\item We call $\tilde{H}$ an $[r,M,J]$-\emph{diluter} of $H$ if $M=o(M_0(n))$.
\end{itemize}
We also call any degree-reducer or diluter of $H$ a \emph{sparsifier} of $H$.
\end{defn}

\vspace{-10pt}

\section{Results\label{sec:results}}
\vspace{-5pt}
Our impossibility results are based on
two families of $2$-local $n$ qubits Hamiltonians, which can both be
expressed in terms of the collective angular momentum
operator $\J_\alpha = \sum_{i=1}^n \sigma_\alpha^{(i)}/2$ for $\alpha\in\{x,y,z\}$,
where $\sigma_\alpha^i$ are the standard Pauli matrices.

\paragraph{Example A (degenerate groundstates)}---
\vspace{-5pt}
\begin{equation}\label{eq:Honeone}
H_\oneone  = \left(\J_z+\frac{n}{2} \right)\left(\J_z+\frac{n}{2}-1\right) =  \frac14\sum_{i< j}^n (1-\sigma_z^{(i)})\otimes(1-\sigma_z^{(j)}) = \sum_{i< j}^n \ketbra{1}^{(i)}\otimes\ketbra{1}^{(j)}.
\vspace{-5pt}
\end{equation}
There are $M_0(n)=n(n-1)/2=\Omega(n^2)$ terms in $H_\oneone$, and each qubit
has degree $n-1$. The terms in $H_\oneone$ mutually commute,
and its groundspace is spanned by the following $n+1$ zero-energy
orthonormal states that have $\J_z=-n/2$ or $\J_z=-n/2+1$:
\vspace{-5pt}
\begin{equation}
GS(H_\oneone) = \text{span}\{\ket{00\cdots00}, \ket{00\cdots01}, \ket{00\cdots10}, \ldots, \ket{10\cdots00}\}.
\vspace{-5pt}
\end{equation}
If we consider a qubit in $\ket{1}$ to be an ``excitation'',
the groundstates are states with one or zero ``excitations".
Observe that $w_n=0$ and $\gamma_n=1$, independent of $n$;
the system is thus spectrally gapped.

\vspace{-5pt}
\paragraph{Example B (unique groundstate)}---
In this example we require that $n$ is even, $n=2s$:
\vspace{-5pt}
\begin{equation}
\label{eq:Hdicke}
H_\dicke =  \J_z^2 - \frac12 \vec{\J}^2 + b_n = \frac12(\J_z^2-\J_x^2-\J_y^2) +b_n = \frac14 \sum_{i<j}^{n} (\sigma_z^{(i)}\sigma_z^{(j)} - \sigma_x^{(i)}\sigma_x^{(j)} - \sigma_y^{(i)}\sigma_y^{(j)}) -\frac{n}{8} + b_n
\end{equation}
where $b_n\equiv \frac18 n(n + 2)$ is a constant chosen so that $\lambda_1(H_\dicke)=0$. Similarly to $H_\oneone$, this
Hamiltonian has $M_0(n)=\frac12 n(n-1)$ $2$-local terms,
and each qubit has degree $n-1$. Since $[\vec{\J}^2,\J_z]=0$,
the eigenstates of $H_B$ can be written in eigenbasis of
both $\vec{\J}^2$ and $\J_z$; it is an easy exercise
(see Appendix~\ref{sec:HamProperties})
that the following well-known Dicke state from atomic physics~\cite{Dicke}
is the unique groundstate of $H_\dicke$ with eigenvalue $0$:
\begin{equation}
\ket{g_\dicke} = \ket{\J=\frac{n}{2}; \J_z = 0} =
\binom{n}{n/2}^{-1/2} \sum_{|\{i\,:\,x_i=1\}| = n/2} \ket{x_1\cdots x_n}.
\end{equation}
Other eigenstates have energy at least $1$, so
the system is spectrally gapped with $w_n=0$ and $\gamma_n = 1$.

It turns out that these deceptively simple examples form a
challenge for Hamiltonian sparsification.

\subsection{Limitations on Degree-Reduction}
For didactic reasons, we start
by ruling out generic {\it perfectly coherent}
DR. This is done by showing that such DR is impossible for
$H_\oneone$.

\begin{lemma}[Impossibility of generic $0$-incoherence DR]
\label{lem:imposs1-DR}
There does not exist any $k$-local
Hamiltonian $\tilde{H}_\oneone$ that is an
$[o(n/k),M,J]$-degree-reducer of the $n$-qubit Hamiltonian $H_\oneone$
with localized encoding, $0$-incoherence, and energy spread $\tilde{w} <1/2$,
regardless of number of terms $M$ or interaction strength $J$.
\end{lemma}

A closer inspection of the proof implies a
trade-off between $\epsilon$ and $J$,
from which it follows that if $J=O(1)$
then generic DR is impossible even if we allow $\epsilon$
which is inverse polynomially small
(see exact statement in Lemma~\ref{lem:imposs1},
Appendix \ref{sec:imposs1}.)
We note that this result in fact rules out any improvement of
the degree
for $H_\oneone$, to some sub-linear degree.

However, perfect (or even inverse-polynomially close to perfect)
coherence is a rather strong requirement.
Indeed, by improving our proof techniques,
we manage to improve our results for $H_\oneone$ to show
impossibility even for constant coherence.
Moreover, by
devising another Hamiltonian with a unique groundstate, $H_\dicke$, and proving such an impossibility result also
for this Hamiltonian, we arrive at the following theorem.
Our main result is a strong impossibility result, ruling out generic
DR with {\it constant
unfaithfulness} (and consequently, also constant incoherence).

\begin{thm}[Main: Impossibility of constant coherence (faithfulness)
DR for
$H_\oneone$ ($H_\dicke$)]
\label{thm:main}
For sufficiently small constants $\epsilon\ge0$ $(\delta\ge0)$ and $\tilde{w}\ge0$,
there exists system size $n_0$ where for any $n\ge n_0$, there is no
$O(1)$-local $[O(1),M,O(1)]$-degree-reducer of the $n$-qubit Hamiltonian $H_\oneone$ $(H_\dicke)$
with localized encoding,  $\epsilon$-incoherence ($\delta$-unfaithfulness),
and energy spread $\tilde{w}$, for any number of Hamiltonian terms $M$.
\end{thm}

We deduce that generic quantum DR, with even constant unfaithfulness,
is impossible. This stands in striking contrast to the classical setting.
It is well known that classical DR is possible for all
CSPs in the context of PCP reductions\cite{dinur}.
This construction
easily translates to a $0$-unfaithfulness degree-reducer for any
{\it classical} local Hamiltonian:

\begin{prop}[Incoherent DR of classical Hamiltonians]
\label{prop:classical-deg-reduct}
Consider an $n$-qudit $k$-local \emph{classical} Hamiltonian $H = \sum_{S\subset \{1,\ldots, n\}} C_S$, where each $C_S:\{z_i:i\in S\} \to [0, 1]$ is a function of $d$-ary strings of length $|S|\le k$ representing states of qudits in $S$.
Let the number of terms in $H$ be $M_0=|\{S\}|=O(n^k)$.
Then there is a $k$-local $[3,O(kM_0),O(1)]$-degree-reducer of $H$
with $0$-unfaithfulness, no energy spread, and trivial encoding $V=\Id$.
\end{prop}

This demonstrates a large difference between the quantum and classical settings in the context
of Hamiltonian sparsification.
Characterizing which quantum Hamiltonians can be degree-reduced (with bounded interaction strength),
either coherently or just faithfully, remains open.

The impossibility of DR by Theorem \ref{thm:main}, which heavily relies on the interaction strength $J$ being a constant, is essentially tight.
We prove this in a complementary result showing that degree-reduction is possible when $J$ is allowed to grow polynomially for any local Hamiltonian whose spectral gap closes slower than some polynomial (which is the case of interest for gap-simulation):

\begin{thm}[Coherent DR with polynomial interaction strength]
\label{thm:degree-reduction-poly}
Suppose $H$ is an $O(1)$-local Hamiltonian with a quasi-groundspace projector $P$, which has quasi-spectral gap $\gamma=\Omega(1/\poly(n))$ and energy spread $w$.
Also assume $\|H\|=O(\poly(n))$.
Then for every $\epsilon>0$, one can construct an $O(1)$-local $[O(1), O(\poly(n)/\epsilon^2), O(\poly(n,\epsilon^{-1}))]$-degree-reducer of $H$
with incoherence $\epsilon$, energy spread $w+O(1/\poly(n))$, and trivial encoding.
\end{thm}
The proof is constructive:
we map any given Hamiltonian to the quantum phase-estimation circuit,
make the circuit sparse, and transform it back to a Hamiltonian using Kitaev's circuit-to-Hamiltonian construction~\cite{KSV02}.
Some innovations are required to ensure coherence within the groundspace isn't destroyed.
For the most general local Hamiltonian whose spectral gap may close exponentially, we can show that coherent DR is possible with exponential interaction strength:
\begin{thm}[Coherent DR with exponential interaction strength]
\label{thm:degree-reduction-exp}
Let $H$ be an $n$-qubit $O(1)$-local Hamiltonian with
$M_0$ terms, each with bounded norm.
Suppose $H$ has quasi-spectral gap $\gamma$ and energy spread $w$ according to Def.~\ref{defn:gap}.
For any $\epsilon>0$, one can construct a $2$-local
$[O(1), O(M_0), O ((\gamma\epsilon)^{-\poly (n)} )]$-degree-reducer of $H$ with incoherence $\epsilon$, energy spread $w+\O(\epsilon)$, and trivial encoding.
\end{thm}
The proof uses a construction from perturbative gadgets, and is similar to other results in the Hamiltonian simulation
literature~\cite{OliveiraTerhal,UniversalHamiltonian}.
Due to significantly more resource required compared to Theorem~\ref{thm:degree-reduction-poly}, this construction is only useful in situations where we want to preserve some exponentially small spectral gap.

\vspace{-10pt}

\subsection{Limitations on Dilution}
For perfect or near-perfect dilution, we can prove a similar impossibility
result to Lemma~\ref{lem:imposs1-DR}:

\begin{thm}[Impossibility of generic $0$-incoherence dilution]
  \label{thm:imposs1-dilute}
There does not exist any $k$-local
Hamiltonian $\tilde{H}_\oneone$ that is an
$[r,o(n^2/k^2),J]$-diluter of the $n$-qubit Hamiltonian $H_\oneone$
with localized encoding, 0-incoherence, and energy spread $\tilde{w}<1/2$, 
regardless of degree $r$ or interaction strength $J$.
\end{thm}

Similar to Lemma~\ref{lem:imposs1-DR}, this in fact holds even if we allow inverse polynomial
incoherence (see Lemma~\ref{lem:imposs1});
and like above, this seems to be a rather weak impossibility result since
requiring inverse polynomial incoherence may be too strong in many situations.
Can we strengthen this to rule out dilution with constant incoherence?
The proof technique in Theorem \ref{thm:main} does not apply for dilution, since it
relies on the decay of correlation between {\it distant} nodes in the
interaction graph of $\tilde{H}$
(see Sec.~\ref{sec:proof-sketch-main}).
On the other hand, a diluter $\tilde{H}$ can have unbounded degree,
and hence constant diameter, e.g. the star graph.
Nevertheless, under a computational hardness assumption,
no efficient classical {\it algorithm} for generic constant-unfaithfulness
dilution exists, even for all $k$-local {\it classical} Hamiltonians:

\begin{thm}[Impossibility of dilution algorithm for classical Hamiltonians]
\label{thm:imposs-dilute}
If $\coNP \not\subseteq \NPpoly$, then
for any $\xi>0$, $\delta < 1/\sqrt{2}$, $\tilde{w} \le 1/2$, there is no classical algorithm that
given a $k$-local $n$-qubit classical Hamiltonian $H$, runs in $O(\poly(n))$ time to
find an $[r,O(n^{k-\xi}),J]$-diluter of $H$ with $\delta$-unfaithfulness, energy spread $\tilde{w}$, and any encoding $V$ that has an $O(n^{k-\xi})$-bit description. This holds for any $r$ and $J$.
\end{thm}

The above result rules out general (constructive)
dilution even when the Hamiltonians are classical.
For specific cases, however, dilution is possible.
Our $H_\oneone$ (which is also a classical Hamiltonian) provides
such an example,
for which we can achieve dilution
even with $0$-unfaithfulness, in the incoherent setting:

\begin{prop}[$0$-unfaithfulness incoherent dilution and DR for $H_\oneone$]
\label{prop:incoherent-tree}
There is a 3-local incoherent $[2,n-1,1]$-diluter of $H_\oneone$ with 0-unfaithfulness, energy spread $\tilde{w}=0$, and trivial encoding.
This is also an incoherent $[2,n-1,1]$-degree-reducer of $H_\oneone$.
\end{prop}

Furthermore, combining ideas from the construction in Proposition~\ref{prop:incoherent-tree} and Theorem~\ref{thm:degree-reduction-poly}, we can show that coherent dilution of $H_\oneone$ with polynomial interaction strength is also possible:

\begin{prop}[Constant-coherence dilution and DR for $H_A$ with polynomial interaction strength]
\label{prop:circuit}
There is a 6-local $[6, O(n/\epsilon^2),O(\poly(n,\epsilon^{-1}))]$-degree-reducer of $H_\oneone$ with $\epsilon$-incoherence, energy spread $\tilde{w}=0$, and trivial encoding. This is also a $[6, O(n/\epsilon^2),O(\poly(n,\epsilon^{-1}))]$-diluter of $H_\oneone$.
\end{prop}

Note since Theorem~\ref{thm:imposs-dilute} rules out constructive dilution regardless of interaction strength $J$, we cannot hope to prove an analogue of Theorem~\ref{thm:degree-reduction-poly} or \ref{thm:degree-reduction-exp} to build coherent diluters for generic Hamiltonians, even allowing arbitrarily large interaction strength.
Nevertheless, it remains an interesting open question to characterize Hamiltonians for which diluters exist,
whether coherent or incoherent, with constant or large interaction strengths.

\subsection{Connection to Quantum PCP\label{sec:other-results}}
It might appear that our results rule out quantum degree-reduction (DR) in the context of quantum PCP (which would add to existing results \cite{BravyiVyalyi, AharonovEldar2011, qPCPArad, BrandaoHarrow, qPCPHastings,  AharonovEldar2013} ruling out quantum generalizations of other parts of Dinur's PCP proof \cite{dinur}).
However, our results in this context (detailed in Appendix~\ref{sec:qPCP}) currently have rather weak implications towards such a statement.
The catch is that despite the apparent similarity, our gap-simulating DR is a very different
notion from DR transformations used in the context of
quantum and classical PCP.
Gap-simulation seeks the {\it existence} of a Hamiltonian $\tilde{H}$ that
reproduces the properties of
the groundstate(s)  and {\it spectral gap} of an input Hamiltonian $H$.
On the other hand, a qPCP reduction is an {\it algorithm} that given
$H$, it is merely required to output some $\tilde{H}$, such
that if the groundstate energy of $H$ is small (or large), then so is the
groundstate energy of $\tilde{H}$; in other words, qPCP preserves the {\it
promise gap}.
Notice that such a $\tilde{H}$ always {\it exists},
and the difficulty in qPCP reductions is to generate
$\tilde{H}$ efficiently, without knowing the groundstate energy
of $H$.
Thus, we cannot hope for an information-theoretical
impossibility result (as in Theorem \ref{thm:main} and \ref{thm:imposs1-dilute}) in the qPCP setting
without further restriction on the output.
To circumvent this,
we generalize to the quantum world
a natural requirement, which seems to hold in the classical world for all
known PCP reductions, that the reduction is {\it constructive}: roughly, it implies a mapping not only on the CSPs (Hamiltonians) but
also on individual assignments (states)~\cite{BenSassonPCP,dinurgoldreich} (see definition of qPCP-DR
in Appendix~\ref{sec:qPCP-implication}).
Under this restriction,
we prove the impossibility of qPCP-DR reductions with near-perfect coherence (see Theorem \ref{thm:imposs-qPCP} in Appendix \ref{sec:qPCP} for exact statement).
The proof of Theorem \ref{thm:imposs-qPCP} approximately follows that of impossibility results of Lemma~\ref{lem:imposs1-DR} and Theorem~\ref{thm:imposs1-dilute} for sparsification with close-to-perfect coherence.
Unfortunately, as we explain in Sec.~\ref{sec:proof-sketch-main}, strengthening these results to prove impossibility
for constant error (the regime of interest for qPCP), as is done in Theorem \ref{thm:main}, seems to require another new idea.

\section{Proofs Overview\label{sec:proofs-overview}}

\subsection{Proof Sketch for Main Theorem \ref{thm:main} (and related results: Theorem~\ref{thm:imposs1-dilute}, \ref{thm:imposs-qPCP} and Lemma~\ref{lem:imposs1-DR})\label{sec:proof-sketch-main}}

We start with the idea underlying the impossibility
of degree-reduction and dilution with (close to) perfect coherence
(Lemma~\ref{lem:imposs1-DR} and Theorem~\ref{thm:imposs1-dilute}), which we refer to as ``contradiction-by-energy''.
For simplicity, let's first examine the case of gap-simulation without encoding.
Consider all pairs of original qubits $(i,j)$.
The groundstates of $H_\oneone$ include basis states with zero or
one excitations (namely, 1's), but not 2-excitation states.
Importantly, the groundstates can be obtained from the 2-excitation state by
{\it local} operations $\sigma_x^{(i)}$ and $\sigma_x^{(j)}$.
Assuming the gap-simulator $\tilde{H}_\oneone$ of $H_\oneone$ does not interact the qubits $(i, j)$, we can express the energy of the 2-excitation state
as a linear combination of the energy of 0- and 1-excitation states,
up to an error of $\O(\tilde{w})$ and $\O(\epsilon\|\tilde{H}_\oneone\|)$,
using the fact that we can commute $\sigma_x^{(i)}$ and $\sigma_x^{(j)}$ through
independent parts of $\tilde{H}_\oneone$.
If we assume $\tilde{w}$ is small and $\epsilon=0$, the energy of the 2-excitation state cannot be distinguished from these groundstates.
Thus any gap-simulator $\tilde{H}_\oneone$ must directly interact all pairs
of qubits, which easily proves the impossibility without encoding.
We can also see that if $\epsilon>0$, then DR and dilution
remain impossible if $\|\tilde{H}_\oneone\| \le O(\epsilon^{-1})$, e.g. when $\epsilon$ is polynomially small.
This impossibility easily extends to localized encoding,
where each original qubit is encoded into $O(1)$ qudits in the gap-simulator Hamiltonian either independently or via some constant-depth circuit.
In both cases, the required $\Omega(n)$ degree and $\Omega(n^2)$ interaction terms implied for the non-encoded version translate to the same requirements for the encoded version up to a constant factor, proving Lemma~\ref{lem:imposs1-DR} and Theorem~\ref{thm:imposs1-dilute}.

We now explain the proof of Theorem \ref{thm:main} that rules out degree-reduction even with constant incoherence.
Let us first consider the statement
for $H_\oneone$ with constant $\epsilon$ incoherence.
The challenge is that
the contradiction-by-energy trick used in
the proof of Lemma~\ref{lem:imposs1-DR} and Theorem~\ref{thm:imposs1-dilute}
does not work for $\epsilon=\Theta(1)$ incoherence. The problem is that the error in
energy is of the order of $\O(\epsilon\|\tilde{H}_\oneone\|)$; this is too large for constant $\epsilon$,
and does not allow one to distinguish the energy of ground and excited states.
Instead of contradiction-by-energy, we derive a contradiction using the groundspace correlations
between qubits $(i,j)$, where $\epsilon$-incoherence only induces
an error of $\O(\epsilon)$.
Since $H_\oneone$ is gapped, then any degree-reducer Hamiltonian $\tilde{H}_\oneone$ of $H_\oneone$ must be gapped (while allowing some small energy spread $\tilde{w}$) by Def.~\ref{defn:hamsimul}. We can therefore apply a result (modified to accommodate non-vanishing energy spread, see Lemma~\ref{lem:HastingsKoma} in Appendix~\ref{sec:MHK}) of
Hastings-Koma~\cite{HastingsKoma} stating that groundspace
correlation decays exponentially with the distance on the graph where $\tilde{H}_\oneone$ is embedded.
Since we assume bounded degree, we can find a pair $(i,j)$ among
the original $n$ qubits such that their supports $(S_i, S_j)$ after a localized encoding are $\Omega(\log n)$ distance apart, with respect to the graph metric.
Hence, their correlation $\braket{V\sigma_x^{(i)}\sigma_x^{(j)}V^\dag}$ in the groundspace of $\tilde{H}_\oneone$ must decay as $e^{-\Omega(\log n)} = O(1/\poly(n))$.
Contradiction is achieved by the fact that for any pair of original
qubits $(i,j)$, the groundspace of $\tilde{H}_\oneone$ contains a state of the form $(\ket{0_i1_j}+\ket{1_i0_j})\ket{0^{n-2},\text{rest}}+\O(\epsilon)$,
which has correlation at least $\braket{V\sigma_x^{(i)}\sigma_x^{(j)}V^\dag }=1 - \O(\epsilon)$.
For sufficiently small $\epsilon$ and $\tilde{w}$, this constant correlation from the latter lower bound
contradicts the $O(1/\poly(n))$ upper bound from the Hastings-Koma result.

The second part of Theorem~\ref{thm:main} proves impossibility of incoherent
DR for $H_\dicke$ with $\delta$-unfaithfulness.
Since $H_\dicke$ has a unique groundstate that can be shown to have constant correlation between any pair of original qubits $(i,j)$, we can apply the same argument above
for $H_\oneone$ and show a contradiction with the Hastings-Koma's vanishing upper bound of $O(1/\poly(n))$ for small $\delta$ and $\tilde{w}$.

We now remark how these impossibility proofs can be extended to the context of quantum PCP.
The contradiction-by-energy idea in Lemma~\ref{lem:imposs1-DR} and Theorem~\ref{thm:imposs1-dilute} can indeed be generalized in this context.
In Appendix \ref{sec:qPCP}, we show that under a reasonable restriction on the reduction -- namely that the energy of non-satisfying assignments (frustrated or excited states) after the mapping is
lower bounded by the promise gap -- degree-reduction or dilution for quantum PCP is not generally possible with close-to-perfect (namely inverse polynomial) coherence (Theorem \ref{thm:imposs-qPCP}).
However, this impossibility proof would not work when constant incoherence is allowed.
To move to contradiction-by-correlation as in Theorem~\ref{thm:main}, we need to use some form
of Hastings-Koma, which requires a spectral gap in $\tilde{H}$.
Thus, more innovation is needed as it may be an unnecessarily strong requirement for quantum PCP to preserve the spectral gap.

\subsection{Overview of Remaining Proofs}

\paragraph{Proof sketch: Equivalence between coherent and incoherent gap-simulations for unique groundstates (Lemma~\ref{lem:equiv})}---
We want to show that incoherent gap-simulation implies coherent gap-simulation, in the case of unique groundstate of the original Hamiltonian $H$.
A naive approach using the small error per groundstate of the gap-simulator will not work due to possible degeneracy in the groundspace of the simulator $\tilde{H}$; this (possibly exponential) degeneracy could add an unwanted exponential factor.
Hence, we explicitly construct the subspace on which the ancilla qubits
should be projected by $P_\anc$.
The main observation is that since faithful gap-simulation implies that any state in the groundspace of $\tilde{H}$ must be close to the space spanned by $P_\anc$,
the dimensions of $P_\anc$ and the groundspace of $\tilde{H}$ must be the same.
A sequence of simple arguments then allows us to derive
a bound on the incoherence of any state
(i.e., its norm after the incoherence operator in Eq.~\eqref{eq:incoherence} is applied).

\vspace{-7pt}

\paragraph{Proof sketch: DR of any classical Hamiltonian (Proposition~\ref{prop:classical-deg-reduct})}---
Here we follow the standard classical DR (as in~\cite{dinur}) in which each variable
(of degree $d$) is replaced by $d$ variables, and a ring of equality constraints on these variables is added to ensure
they are the same.
The proof that this satisfies our gap-simulator definition is straightforward.

\vspace{-7pt}

\paragraph{Proof sketch: Coherent DR of any Hamiltonian with $\Omega(1/\poly(n))$ spectral gap using polynomial interaction strength (Theorem~\ref{thm:degree-reduction-poly})}---
The construction is based on mapping the quantum phase estimation (PE) circuit\cite{NielsenChuang} to a Hamiltonian, using a modified version of Kitaev's circuit-to-Hamiltonian construction\cite{KSV02}.
The PE circuit can write the energy of any eigenstate of a given $H$ in an ancilla register, up to polynomial precision using polynomial overhead.
The degree of the Hamiltonian is reduced by ``sparsifying'' the circuit before converting to the Hamiltonian.
To repair the incoherence due to different histories, we run the circuit backwards, removing entanglement
between the ancilla and the original register.
To achieve $\epsilon$-incoherence, we add $O(\poly(n)/\epsilon^2)$ identity gates to the end of the circuit.
The eigenvalue structure of the original Hamiltonian $H$ is restored by imposing energy penalties on the energy bit-string written on the ancilla by the PE circuit.
This yields a full-spectrum simulation of $H$, which also implies a gap-simulation of $H$.

\vspace{-7pt}

\paragraph{Proof sketch: Impossibility of generic dilution algorithm (Theorem~\ref{thm:imposs-dilute})}---
Ref.~\cite{DellvanMelkebeek} shows that under the
assumption $\coNP \not\subseteq \NPpoly$, there is no poly-time algorithm to ``compress'' vertex-cover problems on $n$-vertex $k$-uniform hypergraphs
and decide the problem by communicating $O(n^{k-\xi})$ bits for any $\xi>0$ to a computationally unbounded oracle.
Suppose towards a contradiction that $\A$ is a poly-time algorithm
to dilute any $k$-local classical Hamiltonian;
we use it to derive a compression algorithm for vertex cover.
To this end, $\A$ is given a classical
$k$-local Hamiltonian $H$ encoding a vertex cover problem;
$\A$ produces the diluter $\tilde{H}$ with $O(n^{k-2\xi})$ terms and some encoding $V$ described by $O(n^{k-2\xi})$ bits.
Using Green's function perturbation theory (Lemma \ref{lem:PPgroundspace}),
we show that
$\tilde{H}$ can be written using only $\log(n)$-bit precision as $\tilde{H}'$
with $O(1)$ error in the quasi-groundspace (even accounting for degeneracy).
We then communicate $(\tilde{H}',V)$ to the oracle by sending
$O(n^{k-2\xi}\log n)=O(n^{k-\xi})$ bits.
The oracle then uses any groundstate of $\tilde{H}'$, which has large overlap with groundstates of $H$ for small $\delta$ and high precision, to
decide the vertex cover problem and transmit back the answer.

\vspace{-7pt}

\paragraph{Proof sketch: Incoherent dilution and DR of $H_A$ (Proposition~\ref{prop:incoherent-tree})}---
We use here the usual translation of a classical circuit to a CSP:
$n-1$ qubits in a tree structure
(see Figure \ref{fig:tree}) are used to
simulate counting of the number of $1$s among the original qubits,
and the CSP checks the correctness of this (local) computation.
The ``history'' of the computation
is written on the ancilla qubits, and since
different strings have different such histories, the construction is
incoherent (see Figure \ref{fig:treeincoherent}).

\vspace{-7pt}

\paragraph{Proof sketch: Coherent dilution and DR of $H_\oneone$ with polynomial interaction strength (Proposition~\ref{prop:circuit})}---
We improve upon the construction in Prop.~\ref{prop:incoherent-tree} and Theorem~\ref{thm:degree-reduction-poly}
to obtain a coherent diluter of $H_\oneone$ with polynomial interaction strength.
The key is an $O(n)$-length circuit similar to that of Prop.~\ref{prop:incoherent-tree} with a circuit that counts
the number of $1$s in the same tree geometry.
Using the same tricks in Theorem~\ref{thm:degree-reduction-poly} to uncompute computational histories and idling at the end, we show that this leads to a coherent gap-simulator of $H_\oneone$ with $\epsilon$-incoherence and $O(n/\epsilon^2)$ terms.

\vspace{-7pt}

\paragraph{Proof sketch: Coherent DR for any Hamiltonian
using exponential interaction strength (Theorem~\ref{thm:degree-reduction-exp})}---
In order to provide generic coherent degree-reduction for any local
Hamiltonian, using exponential interaction strength,
we first show that perturbative gadgets\cite{KKR06,OliveiraTerhal,CaoImprovedOTGadget}
can be used for gap-simulation. The proofs make use of Green's function
machinery to bound incoherence.
This allows us to construct a degree-reducer for any $k$-local Hamiltonian by a sequence of perturbative gadget applications.
In the first part of the sequence, we reduce the locality of individual Hamiltonian terms to 3-local via $O(\log k)$ serial applications of subdivision gadgets \cite{OliveiraTerhal}, and each 3-local term is further reduced to 2-local via ``3-to-2-local'' gadgets \cite{OliveiraTerhal}.
Then, each original qubit is isolated from each other by subdivision gadgets so that they only interact with $O(n^{k-1})$ ancilla qubits that mediate interactions.
Finally, applying fork gadgets \cite{OliveiraTerhal} in $O(\log n)$ iterations allows us to reduce maximum degree of these original qubits to 6, generating our desired degree-reducer.
It is this last part that causes the exponential blow-up in the interaction strength required to maintain the gap-simulation.

\vspace{-7pt}

\paragraph{Proof sketch: Generalized Hastings-Koma (Lemma~\ref{lem:HastingsKoma})} ---
In Ref.~\cite{HastingsKoma}, Hastings and Koma
proved the exponential decay of correlations in the quasi-groundspace of a Hamiltonian $H$ consisting of finite-range (or exponentially decaying) interactions between particles embedded on a lattice (or more generally on some graph).
They assume that the system is
spectrally gapped, and has vanishing energy spread as the system size $n\to\infty$.
Their proof is based on the relationship between the correlation $\braket{\sigma^{(i)} \sigma^{(j)}}$ they want to upper bound,
and the commutator $\braket{[e^{-iHt}\sigma^{(i)} e^{iHt}, \sigma^{(j)}]}$.
By applying the Lieb-Robinson bound\cite{LiebRobinson} on the latter, and integrating out the time $t$, they show that under the above conditions, the correlations between operators acting on particles $i$ and $j$ decay exponentially with the graph-theoretic distance between the particles.
For application to the gap-simulation framework, we need to generalize their result to cases where the energy spread is not assumed to vanish with the system size.
This is done by a careful modification of their proofs where we optimize the bounds and integration parameters so that errors due to the non-zero energy spread are suppressed.

\vspace{-5pt}

\section{Discussion and outlook\label{sec:discussion}}

We have initiated the rigorous research of resources required for analog simulation of Hamiltonians, and proved unexpected impossibility results for Hamiltonian sparsification.
Instead of working with full-spectrum simulations~\cite{BravyiHastingsSim,UniversalHamiltonian},
we use a new, relaxed definition of gap-simulation that is motivated by minimal requirements in physics.
We note that impossibility results proven in a relaxed framework are of course stronger.

It would be very interesting to improve our understanding of the new framework of
gap-simulations presented here, and clarify its applicability. 
For a start, it will be illuminating
to find example applications of gap-simulations in cases where
full-spectrum simulations as in Ref.~\cite{BravyiHastingsSim,UniversalHamiltonian}
are unknown or difficult to achieve.
Such simulations can enable experimental studies of these physical systems, by reducing resources required for analog simulations. 
Moreover, in many-body quantum physics, tools to construct ``equivalent''  Hamiltonians that preserve  groundstate properties are of great utility.
In this context, the study of gap-simulations can potentially lead to better understanding of universal behaviors in quantum phases of matter, which are characterized only by groundstate physics~\cite{SachdevQPT}. 
Another possible application of gap-simulators may be in the design of Hamiltonian-based quantum algorithms.
In adiabatic algorithms~\cite{FarhiAdiabatic2000}, it is well known that the higher parts of the spectrum
of the final and initial Hamiltonians can significantly affect the adiabatic
gap~\cite{FarhiQAAFail,DicksonAmin,AharonovAtia};
gap-simulating these final and initial Hamiltonians by others will not
affect the final groundstate, and can sometimes 
dramatically improve on the gap along the adiabatic path.
Gap-simulations may also be a useful tool for tailoring the Hamiltonians used
in other Hamiltonian-based algorithms such as QAOA~\cite{QAOA}. 

We note that incoherent but faithful gap-simulations can be very interesting despite
the apparent violation of the quantum requirement for coherence.
For example, in adiabatic algorithms~\cite{FarhiAdiabatic2000}, we only want to arrive
at one of the solutions (groundstates) to a quantum constraint satisfaction problem.
In addition, in quantum NP~\cite{quantumNPsurvey}, one is interested only in whether a certain eigenvalue {\it exists}, and not in the preservation of the entire groundspace.
However, in the context of quantum simulation and many-body physics,
maintaining coherence seems to be crucial for transporting all the physical
properties of the groundspace.
One would also expect maintaining coherence to be important when gap-simulating a subsystem
(perhaps in an unknown state) of a larger system.

We remark that our framework deliberately avoids requiring that the eigenvalue structure
of the spectrum be maintained even in its low-lying part, so as
to provide a minimal but still interesting definition.
Indeed, when simulating the groundspace,
or a quasi-groundspace with small energy spread, this structure is not important.
Nevertheless, one can imagine an intermediate definition, in which full-spectrum simulation
is too strong, but the structure of a significant portion of the lower part of the spectrum matters.
It might be interesting to extend the framework of gap-simulations to allow for such intermediate
cases in which, for example, Gibbs states at low (but not extremely low)
temperatures are faithfully simulated.

A plethora of open questions arise in the context of sparsification.
First, it will be very interesting to find more examples
where degree-reduction and/or dilution
are possible, or are helpful from the perspective of physical implementations.
Assuming bounded interaction strength, which is generally a limitation of physical systems,
can we rigorously characterize which Hamiltonians can be coherently (or incoherently) degree-reduced?
Of course, similar questions can be asked about dilution.
It will also be interesting to consider saving other resources such as the dimensionality of the particles,
which would be a generalization of alphabet-reductions from the context of PCP to Hamiltonian sparsification.

Our results on the impossibility of dilution are weaker than those
for DR. Can we strengthen these to stronger information-theoretical results,
by finding a quantum Hamiltonian for whom a
diluter does not {\it exist} with constant incoherence, or even constant unfaithfulness?

We mention here that the classical
graph sparsification results of Ref.~\cite{BSS12,BSST13}
can be viewed as dilution of a graph while approximately maintaining
its spectrum.
These results have been generalized to
the matrix setting in Ref.~\cite{SilvaSparsification}; however,
this generalization does not seem to be useful in the context of
diluting the interaction graph of a local Hamiltonian.
The result of Ref.~\cite{SilvaSparsification} shows that for sums of
$d\times d$ positive Hermitian matrices, $O(d)$ matrices are sufficient to
reproduce the spectral properties to good approximation, improving over Chernoff-like bounds~\cite{AhlswedeWinter}.
While this in principle allows one to approximate a sum of terms by a sum of fewer terms,
the required number of terms grows as $d=2^{\Omega(n)}$ for quantum Hamiltonians on $n$ qubits,
and is thus irrelevant in our context.

Improving the {\it geometry} of the simulators is another important task that is relevant
for applications of Hamiltonian sparsification to physical implementations. 
Ref.~\cite{LechnerHaukeZoller} has devised a method of converting the
NP-complete Ising model
Hamiltonian ($H=\sum_{ij} J_{ij} \sigma_z^{(i)}\sigma_z^{(j)} + \sum_i h_i \sigma_z^{(i)}$) on $n$ qubits to a new Hamiltonian on $O(n^2)$ qubits with interactions embedded on a 2D lattice, and sharing the same low-energy spectrum.
Their construction encodes each edge $\sigma_z^{(i)}\sigma_z^{(j)}$ as a new qubit, and
corresponds to an incoherent degree-reducer,
where the new groundstates are
non-locally encoded version of the original states.
Our Proposition~\ref{prop:classical-deg-reduct} also provides
incoherent DR of these Hamiltonians, and without
encoding, but the geometry is not in 2D; it will be interesting to improve our   Proposition~\ref{prop:classical-deg-reduct} as well 
as our other positive Theorems \ref{thm:degree-reduction-poly} and \ref{thm:degree-reduction-exp} to hold using a spatially local $\tilde{H}$. 
We note that if we allow the overhead of polynomial interaction strength, then it should be straightforward to extend the circuit-to-Hamiltonian construction in Theorem~\ref{thm:degree-reduction-poly} for analog simulation of local Hamiltonians on a 2D lattice, by ordering the gates in a snake-like fashion on the lattice similar to Ref.~\cite{OliveiraTerhal, AharonovAQCUniversal}.
Identifying situations where DR in 2D with bounded interaction strength is possible remains an open question.

A different take on the geometry
question is to seek gap-simulators which use a single (or few) ancilla qubits that strongly interact with the rest. This may be relevant for physical systems such as neutral atoms with
Rydberg blockade~\cite{RydbergBlockade}, where an atom in a highly excited level
may have a much larger interaction radius, while no two atoms can be excited in each other's vicinity.

Can we improve our results about quantum PCP, and show impossibility of qPCP-DR with constant incoherence?
This would make our impossibility results interesting also in the qPCP context, as they would imply  impossibility of DR in the qPCP regime of constant error, under a rather natural restriction on the qPCP reduction (see discussion in Appendix~\ref{sec:qPCP}).
This would complement existing impossibility results on various avenues towards qPCP~\cite{BravyiVyalyi, AharonovEldar2011, qPCPArad, BrandaoHarrow, qPCPHastings,  AharonovEldar2013,qPCPsurvey}.
Neverthless, it seems that proving such a result might require a significantly further extension of Hastings-Koma beyond our Lemma~\ref{lem:HastingsKoma}, which may be of interest on its own.

Finally, we mention a possibly interesting variant of gap-simulation,
which we call {\it weak} gap-simulation
(see Appendix~\ref{sec:weak-sparsifier}
and Fig.~\ref{fig:weakgapsimul}).
Here, the groundspace is simulated in an {\it excited}
eigenspace of the simulating Hamiltonian, spectrally gapped from above and below,
rather than in its groundspace.
This can be useful in the context of
Floquet Hamiltonian engineering, where
eigenvalues are meaningful only up to a period, and thus a spectral gap in the middle of the spectrum is analogous to a spectral gap above the groundspace~\cite{floquet}.
Proposition~\ref{prop:star} in Appendix~\ref{sec:weak-sparsifier}
shows how to weakly gap-simulate $H_\oneone$
to provide dilution with {\it constant} incoherence and {\it bounded} interaction strength --
a task which we currently do not know how to do using ``standard'' gap-simulation.
It remains open whether one can show stronger possibility results under weak gap-simulation.
If not, can the impossibility results presented here
be extended to the weak-gap-simulation setting?
This might require
an even stronger extension of Hastings-Koma's theorem.

Overall, we hope that the framework, tools, and
results presented here will lead to progress in understanding the possibilities and limitations in simulating Hamiltonians by other
Hamiltonians -- an idea that brings the notion of {\it reduction} from
classical computer science into the quantum realm,
and constitutes one of the most important contributions of the field of
quantum computational complexity to physics.

\section{Acknowledgements}
We are grateful to Oded Kenneth for  suggesting the construction of $H_\dicke$,
and for fruitful discussions;  to Itai Arad for
insightful remarks about the connection to quantum PCP;
to Eli Ben-Sasson for discussions about PCP;
to Ashley Montanaro and Toby Cubitt for clarifications about Hamiltonian simulation.
D.A. is grateful to the generous support of ERC grant 280157 for its support
during the completion of most of this project.
L.Z. is thankful for the same
ERC grant for financing his visits to the research group of D.A.
D.A. also thanks the ISF grant No.~039-9494 for supporting this
work at its final stages.

\newpage

\begin{appendices}

\section{Properties of Gap-Simulation\label{sec:gap-simulation-properties}}

\subsection{Relationship between Coherent and Incoherent Gap-Simulation\label{sec:uniqueGS}}
Here we show the relationship between coherent and incoherent gap-simulations.
We first prove the easy direction: incoherence provides an upper bound on unfaithfulness.
Pick $\delta$ to be the exact value of
unfaithfulness (and not just an upper bound on it), then
\begin{eqnarray}
\delta &=&  \|\tilde{P}-V(P\otimes \Id_\anc )V^\dag \tilde{P}\| \nonumber\\
&=& \|\tilde{P}-V(P\otimes \Id_\anc )V^\dag V(P\otimes P_\anc )V^\dag - V(P\otimes \Id_\anc )V^\dag[\tilde{P}-V(P\otimes P_\anc )V^\dag]\| \nonumber \\
&=& \|\tilde{P}-V(P\otimes P_\anc )V^\dag - V(P\otimes \Id_\anc )V^\dag[\tilde{P}-V(P\otimes P_\anc )V^\dag]\| \nonumber \\
&\le& \|\tilde{P}-V(P\otimes P_\anc )V^\dag\| +  \|V(P\otimes \Id_\anc )V^\dag[\tilde{P}-V(P\otimes P_\anc )V^\dag]\| \nonumber \\
&\le& 2\epsilon
\end{eqnarray}
The above uses the fact that $V^\dag V=\Id$ for isometries.
When $V$ is unitary ($VV^\dag=V^\dag V=\Id$), we can obtain an even better bound of $\delta\le\epsilon$:
\begin{eqnarray}
\delta &=&  \|\tilde{P}-V(P\otimes \Id_\anc )V^\dag \tilde{P}\| = \|VV^\dag\tilde{P}-V(P\otimes \Id_\anc )V^\dag \tilde{P}\| = \|V(P^\perp \otimes \Id_\anc) V^\dag \tilde{P} \|\nonumber \\
&=&   \| V(P^\perp \otimes \Id_\anc)V^\dag [\tilde{P} - V(P\otimes P_\anc) V^\dag ]\| \le \|\tilde{P} - V(P\otimes P_\anc) V^\dag \|\nonumber \\
&\le& \epsilon.
\end{eqnarray}

We now prove Lemma~\ref{lem:equiv}, which shows that in the case of a unique groundstate, the opposite direction holds as well.
We reproduce this lemma below:

{
\renewcommand{\thelemma}{\ref{lem:equiv}}
\begin{lemma}[Equivalence of coherent and incoherent gap-simulation when groundstate is unique]
Suppose the Hamiltonian $H$ has a unique groundstate, i.e. its groundspace projector $P=\ketbra{g}$. If $\tilde{H}$ gap-simulates $H$ with unfaithfulness
$\delta <1$, then it also gap-simulates $H$ with incoherence
$\epsilon \le \sqrt{2}\delta/\sqrt{1-\delta^2}$.
\end{lemma}
\addtocounter{lemma}{-1}
}

It turns out that it'll be helpful to first prove the following technical lemma (which will also be used in Appendix~\ref{sec:degree-reduction-poly} to prove Theorem~\ref{thm:degree-reduction-poly}).

\begin{lemma}[Projector Difference Lemma]
\label{lem:proj-diff}
Consider two Hermitian projectors $\Pi_A$ and $\Pi_B$, such that $\rank(\Pi_A)\le \rank(\Pi_B)$.
Suppose that for all normalized $\ket{\phi}\in \tilde{\Pi_B}$, $\|(\Pi_B-\Pi_A)\ket{\phi}\| \le \delta$.
Then $\|\Pi_B-\Pi_A\| \le \sqrt{2}\delta/\sqrt{1-\delta^2}$.
\end{lemma}
\begin{proof}
Let us denote $\rank(\Pi_A) =k$ and $\rank(\Pi_B)=\ell$.
First, we observe that for all normalized $\ket{\phi}\in \Pi_B$, $\ket{\phi} = \Pi_A\ket{\phi} + \Pi_A^\perp\ket{\phi}$, and thus
\begin{equation}
\|\Pi_A\ket{\phi}\|^2 = 1-\|\Pi_A^\perp\ket{\phi}\|^2 = 1-\|(\Pi_B-\Pi_A)\ket{\phi}\|^2 \ge 1-\delta^2.
\label{eq:bound-BinA}
\end{equation}

Now consider any normalized $\ket{\phi^\perp} \perp \Pi_B$. We want to bound $\|(\Pi_B-\Pi_A)\ket{\phi^\perp}\|$.
To this end, consider the space $\V=\spn\{\Pi_B, \ket{\phi^\perp}\}$, which has dimension $\dim(\V)=\ell+1$.

We first argue that there exists $\ket{v}\in \V$ such that $\ket{v}\perp \Pi_A$. To see this, pick an orthonormal basis for $\Pi_A$: $\{\ket{\beta_1},\ldots,\ket{\beta_k}\}$.
We then write $\ket{\beta_i} = \ket{\beta_i^\V} + \ket{\beta_i^\perp}$,
with $\ket{\beta^\V_i}\in \V$
and $\ket{\beta^\perp_i}\perp \V$.
Let $W\subseteq \V$ be the subspace of $\V$
spanned by the $\ket{\beta^\V_i}$.
Note $\dim(W)\le k < \ell+ 1= \dim (\V)$.
Hence, there exists a unit vector $\ket{v}\in \V$,
$\ket{v}\perp W$,
which means that $\ket{v} \perp  \ket{\beta^\V_i}$ for all $1\le i\le k$.
But since $\ket{v}\in \V$, $\ket{v}$ is also orthogonal to all $\ket{\beta^\perp_i}\perp \V$.
Therefore,
$\ket{v}$ is orthogonal to all $\{\ket{\beta_i}\}_{i=1}^k$, and so
$\ket{v}\perp \Pi_A$.

Write
$\ket{v}=a\ket{\phi}+b\ket{\phi^\perp}$ for some complex numbers $a,b$ and unit vector $\ket{\phi}\in \Pi_B$.
Note
by our assumption, $\ket{\phi^\perp}\perp\Pi_B$.
Since $\ket{v}\perp \Pi_A$, we have:
$a\Pi_A\ket{\phi}=-b\Pi_A \ket{\phi^\perp}$.
Hence there exists a unit vector
$\ket{\gamma}\in \Pi_A$ and complex numbers $x,y$ such that
\begin{align}
\Pi_A\ket{\phi} &= y\ket{\gamma}, \\
\Pi_A \ket{\phi^\perp}&= x\ket{\gamma},\\
\text{and} \quad ay&=-bx.
\end{align}
We can thus write:
\begin{align}
\ket{\phi} &= y\ket{\gamma}+\ket{\gamma'}, \quad \ket{\gamma'}\perp \Pi_A \\
\ket{\phi^\perp}&= x\ket{\gamma}+\ket{\gamma''}, \quad \ket{\gamma''}\perp \Pi_A
\end{align}
From $\ket{\phi}\in \Pi_B$ we have $\ket{\phi} \perp \ket{\phi^\perp}$ which implies
\begin{equation} \label{eq:bound-xy}
xy+\langle\gamma'|\gamma''\rangle=0.
\end{equation}
By Eq.~\eqref{eq:bound-BinA}, we have
$|y|=\|\Pi_A\ket{\phi}\|\ge \sqrt{1-\delta^2}$ and therefore
\begin{equation}
\|\ket{\gamma'}\|=\sqrt{1-|y|^2}\le \delta.
\end{equation}
Due to Eq.~\eqref{eq:bound-xy}
we have $|xy|\le \delta$.
Since $|y|\ge \sqrt{1-\delta^2}$, we have for any $\ket{\phi^\perp}\perp \Pi_B$,
\begin{equation}
\|(\Pi_B-\Pi_A)\ket{\phi^\perp}\| = \|\Pi_A\ket{\phi^\perp}\| = |x|\le \delta/\sqrt{1-\delta^2}.
\label{eq:bound-Bout}
\end{equation}

Given any unit vector $\ket{\psi}$
in the Hilbert space,
write
$\ket{\psi}=\sqrt{1-z}\ket{\phi}+\sqrt{z}\ket{\phi^\perp}$
for some unit vectors $\ket{\phi}\in \Pi_B$ and $\ket{\psi^\perp} \perp \Pi_B$,
along with some number $z$ where $0\le z\le 1$.
Using the triangle inequality, it follows from Eq.~\eqref{eq:bound-Bout} and the given fact of $\|(\Pi_B-\Pi_A)\ket{\phi}\|\le \delta$ that
\begin{eqnarray}
\|(\Pi_B-\Pi_A)\ket{\psi}\| &\le& \sqrt{1-z} \delta + \sqrt{z}\frac{\delta}{\sqrt{1-\delta^2}} \le \sqrt{\delta^2+\frac{\delta^2}{1-\delta^2}} \nonumber \\
&\le& \sqrt{2}\delta/\sqrt{1-\delta^2}
\end{eqnarray}
where in last part of the first line we used the fact that $\sqrt{1-z}a+\sqrt{z}b\le\sqrt{a^2+b^2}$ for any real numbers $a,b$ and $0\le z\le 1$.

\end{proof}

\begin{proof}[\textbf{Proof of Lemma~\ref{lem:equiv}}]
We know that $\ket{g}$ is the unique groundstate of $H$.
We can always represent in a unique way any state in the Hilbert space of
$\tilde{H}$ by what we call the {\it $g$-representation}:
\begin{equation}
\ket{\alpha}=V(\ket{g}\ket{\alpha_g})+\ket{\alpha^\perp}.
\end{equation}
such that the reduced matrix of $V^\dag \ket{\alpha^\perp}$ to the left register,
has zero support on $\ket{g}$.
We call $\ket{\alpha_g}$ the $g$-vector of $\ket{\alpha}$.
Keep in mind that the two components are orthogonal since
$\big(\bra{g}\braket{\alpha_g|\big)V^\dag \,|\alpha^\perp} = 0$.

We now construct the projector $P_\anc$
and show that it satisfies the requirement of the Lemma.
We define:
\begin{defn}
$P_\anc$ is defined to be the projector onto the span of all $g$-vectors $\ket{\alpha_g}$
of all vectors $\ket{\alpha}\in \tilde{P}$.
\end{defn}

Consider the space on which $\Pi_A = V(P\otimes P_\anc)V^\dag = V(\ketbra{g}\otimes P_\anc)V^\dag $
projects on.
We first show that $\rank(\Pi_A) \le \rank(\tilde{P})$.

Let $\ell=\rank(\tilde{P})$, $k=\rank(\Pi_A)$.
We first show $k\le \ell$.
Let $\ket{\alpha_1},...,\ket{\alpha_\ell}$ be an orthonormal
basis for $\tilde{P}$,
and let $\ket{\alpha_g^i}$ be their $g$-vectors, respectively.
Note that $\{\ket{\alpha_g^i}\}_{i=1}^\ell$ span $P_\anc$.
This follows from the definition of $P_\anc$ as the span of the $g$-vectors
for all $\ket{\alpha}\in \tilde{P}$,
since $\ket{\alpha}$ can be written as
a linear combination of $\ket{\alpha_i}$, which implies that its
$g$-vector is also a linear combination of their $g$-vectors.
Hence, $k=\rank(\Pi_A)\le \ell$.

Consider any unit vector $\ket{\alpha}\in \tilde{P}$.
We have $\tilde{P}\ket{\alpha}=\ket{\alpha}$,
$V (P\otimes \Id_\anc) V^\dag \tilde{P}\ket{\alpha}= V \ket{g}\ket{\alpha_g}$.
By the unfaithfulness condition we have
$\|\ket{\alpha}- V \ket{g}\ket{\alpha_g}\| = \|\ket{\alpha^\perp}\| \le \delta$.
Since
$\Pi_A\ket{\alpha}= V \ket{g}\ket{\alpha_g}$, we also have
$\|(\tilde{P} - \Pi_A)\ket{\alpha}\| = \|\ket{\alpha^\perp}\| \le \delta$.
Hence, we can apply Lemma~\ref{lem:proj-diff} by identifying $\Pi_B \equiv \tilde{P}$, and
obtain the desired bound:
\begin{equation}
\|\tilde{P}-V(P\otimes P_\anc)V^\dag\| = \|\tilde{P} - \Pi_A\| \le \sqrt{2}\delta/\sqrt{1-\delta^2}.
\end{equation}
\end{proof}

\subsection{Composition of Gap-Simulations}

We now prove Lemma~\ref{lem:composition}, which demonstrates that composition of gap-simulations behaves as expected:

{
\renewcommand{\thelemma}{\ref{lem:composition}}
\begin{lemma}[Composition]
Suppose $H_1$ (incoherently)  gap-simulates $(H_0,P_0)$ with encoding $V_1$, incoherence $\epsilon_1$ (or unfaithfulness $\delta_1$), energy spread $\tilde{w}_1$, and a corresponding quasi-groundspace projector $P_1$.
Also suppose $H_2$ (incoherently) gap-simulates $(H_1, P_1)$ with encoding $V_2$, incoherence $\epsilon_2$ (or unfaithfulness $\delta_2$), and energy spread $\tilde{w}_2$.
Then $H_2$ (incoherently) gap-simulates $(H_0,P_0)$ with encoding $V_2 (V_1\otimes \Id_{\anc,1})$, incoherence $\le\epsilon_2+\epsilon_1$ (or unfaithfulness $\le 2\delta_2+\delta_1$), and energy spread $\tilde{w}_2$.
\end{lemma}
\addtocounter{lemma}{-1}
}

\begin{proof}
Below, we denote $P_2$ as the corresponding quasi-groundspace projector of $H_2$.

Let us first prove the case of coherent gap-simulation.
Note by definition, the quasi-spectral gap of $H_2$ is the same as $H_1$, which is the same as $H_0$.
As the energy spread of $H_2$ is already given as $\tilde{w}_2$, condition 1 of Def.~\ref{defn:hamsimul} is satisfied.
Let us denote $P_{\anc,i}$ as the ancilla projector for gap-simulating $(H_i,P_i)$, $i\in\{0,1\}$.
It remains to satisfy the condition 2 of bounded incoherence, i.e. Eq.~\eqref{eq:incoherence}.
Let us denote $\E_1 = P_1-V_1(P_0\otimes P_{\anc,0})V_1^\dag$, which satisfies $ \|\E_1\|  \le \epsilon_1$.
Then
\begin{eqnarray}
&&P_2-V_2(P_1\otimes P_{\anc,1})V_2^\dag \nonumber \\
&=& P_2 - V_2 \left[V_1(P_0\otimes P_{\anc,0})V_1^\dag \otimes P_{\anc,1}\right]V_2^\dag- V_2\left[\big(P_1-V_1(P_0\otimes P_{\anc,0})V_1^\dag\big) \otimes P_{\anc,1}\right]V_2^\dag  \nonumber \\
&=& P_2 - V_2(V_1\otimes \Id_{\anc,1})(P_0\otimes P_{\anc,0}\otimes P_{\anc,1}) (V_1\otimes \Id_{\anc,1})^\dag V_2^\dag - V_2(\E_1\otimes P_{\anc,1})V_2^\dag
\end{eqnarray}
By defining an isometry $V_{21}\equiv V_2 (V_1\otimes \Id_{\anc,1})$ as in the statement of the Lemma, we see that
\begin{align}
P_2 - V_{21}(P_0\otimes P_{\anc,0}\otimes P_{\anc,1})V_{21}^\dag &= P_2-V_2(P_1\otimes P_{\anc,1})V_2^\dag + V_2(\E_1\otimes P_{\anc,1}) V_2^\dag \nonumber \\
\|P_2 - V_{21}(P_0\otimes P_{\anc,0}\otimes P_{\anc,1})V_{21}^\dag \| & \le \|P_2-V_2(P_1\otimes P_{\anc,1})V_2^\dag\| + \|\E_1\| \nonumber \\
&\le \epsilon_2 + \epsilon_1.
\end{align}

Now let's consider the case of incoherent gap-simulation with bounded unfaithfulness.
Again, $H_2$ satisfies condition 1 of Def.~\ref{defn:hamsimul-incoherent} for incoherently gap-simulating $(H_0,P_0)$ as given, with energy spread $\tilde{w}_2$.
It remains to satisfy the condition 2 of bounded unfaithfulness Eq.~\eqref{eq:unfaithfulness}.
Let us denote $P_0' = V_1(P_0\otimes \Id_{\anc,0} )V_1^\dag$, and $P_1' = V_2(P_1\otimes\Id_{\anc,1})V_2^\dag$.By assumption, $\|P_1-P_0' P_1\|\le \delta_1$ and $\|P_2-P_1' P_2\|\le \delta_2$.
In the following, we will omit the $\Id_{\anc,i}$ for readability.
Observe
\begin{eqnarray}
&& P_2 - V_2V_1 P_0 V_1^\dag V_2^\dag P_2 = P_2 - V_2 P_0' V_2^\dag P_2 \nonumber \\
&=& P_2-P_1'P_2 + P_1'P_2 - V_2 P_0'V_2^\dag P_1' P_2 + V_2 P_0'V_2^\dag P_1' P_2 - V_2P_0'V_2^\dag P_2 \nonumber \\
&=& (P_2-P_1'P_2) + (P_1'- V_2 P_0'V_2^\dag P_1')P_2 +  V_2 P_0'V_2^\dag ( P_1' P_2 -P_2) \nonumber \\
&=& (P_2-P_1'P_2) + (V_2 P_1V_2^\dag - V_2 P_0'V_2^\dag V_2 P_1 V_2^\dag)P_2 +  V_2 P_0'V_2^\dag ( P_1' P_2 -P_2) \nonumber \\
&=& (P_2-P_1'P_2) + V_2 (P_1 - P_0' P_1) V_2^\dag P_2 +  V_2 P_0'V_2^\dag ( P_1' P_2 -P_2) .
\end{eqnarray}
Hence, by identifying $V_{21} = V_2V_1$, we can bound the RHS above as
\begin{equation}
\|P_2 - V_{21} P_0 V_{21}^\dag\| \le \|P_2-P_1'P_2\| + \|(P_1 - P_0' P_1)\| + \|P_1' P_2 -P_2\| \le 2\delta_2+\delta_1,
\end{equation}
as stated.
\end{proof}

\subsection{Comparison of Gap-Simulation to Full-Spectrum Simulation \label{sec:comp-defns}}

Generally, analog Hamiltonian simulators are designed to reproduce the spectral properties (both eigenvalues and eigenvectors) of a given Hamiltonian.
In Ref.~\cite{BravyiHastingsSim}, Bravyi and Hastings introduced a definition that quantifies how well Hamiltonian $\tilde{H}$ simulates a given Hamiltonian $H$, while allowing some encoding by a ``sufficiently simple'' isometry $V$, which can be summarized roughly as $\|H - V^\dag \tilde{H} V\|\le \xi$.
Ref.~\cite{UniversalHamiltonian} refines this definition by allowing for the more general case of simulating complex Hamiltonians by a real ones, but imposes a more explicit constraint that the isometries to be local, i.e. $V=\bigotimes_i V_i$.
We reproduce that definition below:

\begin{defn}[Full-spectrum simulation, adapted
    from Def.~1 of \cite{UniversalHamiltonian}]
\label{defn:CMPsimul}
A many-body Hamiltonian $\tilde{H}$ \emph{full-spectrum-simulates} a Hamiltonian $H$ to precision $(\eta,\xi)$ below an energy cut-off $\Delta$ if there exists a local encoding $\E(H)=V(H\otimes P + \bar{H}\otimes Q)V^\dag$, where $V=\bigotimes_i V_i$ for some isometries $V_i$ acting on 0 or 1 qubits of the original system each, and $P$ and $Q$ are locally orthogonal projectors, such that
\begin{enumerate}
\item There exists an encoding $\tilde{\E}(H) = \tilde{V}(H\otimes P + \bar{H}\otimes Q)\tilde{V}^\dag$ such that  $\|\tilde{V}-V\|\le \eta$ and $\tilde{\eps}(\Id) = P_{\le\Delta(\tilde{H})}$, where $P_{\le\Delta(\tilde{H})}$ is the projector onto eigenstates of $\tilde{H}$ with eigenvalue $\le \Delta$.
\item $\|\tilde{H}P_{\le\Delta(\tilde{H})} - \tilde{\E}(H)\| \le \xi$.
\end{enumerate}
\end{defn}

The condition of local orthogonality of $P$ and $Q$ means that there exist orthogonal projectors $P_{i}$ and $Q_i$ acting on the same qubits as $V_i$, such that $P_i P=P$ and $Q_i Q = Q$.
The appearance of $\bar{H}$, which is the complex-conjugate of $H$, is necessary to allow for encoding of complex Hamiltonians into real ones.
Note that for any real-valued Hamiltonian $H$,  we can simply write $\E(H) = V(H\otimes P_\anc) V^\dag$, where $P_\anc = P+Q$ is a projector since $P$ are $Q$ are orthogonal.

Note the definition of Ref.~\cite{BravyiHastingsSim} can be considered as a special case of the one above by setting $P=\Id$ and $Q=0$, while allowing more general isometry $V$ for encoding.
Hence, we focus our comparison to the above Definition~\ref{defn:CMPsimul} from Ref.~\cite{UniversalHamiltonian}.

We also note that our restriction to localized encodings per Definition~\ref{defn:localized-encoding} is somewhat different than the notion of ``local encoding" $V=\bigotimes_i V_i$ in Ref.~\cite{UniversalHamiltonian}.
For example, constant-depth circuit qualifies as a localized encoding but not a ``local encoding'', due to the possibility of overlaps between supports of encoded qubits (and hence cannot be written in tensor-product form).
On the other hand, Ref.~\cite{UniversalHamiltonian} does not appear to place any explicit restriction on the size of the support of each encoded qubit, other than the fact that each qubit is encoded independently.

Now, we show that full-spectrum simulation by Def.~\ref{defn:CMPsimul} with an encoding of the form $\eps(H)=V(H\otimes P_\anc)V^\dag$ and sufficiently small precision ($\xi\ll(1-w)\gamma$) implies a coherent gap-simulation by our Def.~\ref{defn:hamsimul}.
The restriction of the encoding format simplifies the comparison, and has no loss of generality when considering real-valued Hamiltonians.

\begin{lemma}[Full-spectrum simulation implies coherent gap-simulation]
\label{lem:relate-CMP}
Let $H$ be a Hamiltonian that has a quasi-groundspace projector $P$ with quasi-spectral gap $\gamma$ and energy spread $w \le 1/2$.
Suppose $\tilde{H}$ full-spectrum-simulates $H$ to precision $(\eta,\xi)$ according to Def.~\ref{defn:CMPsimul} with encoding $\eps(H)=V(H\otimes P_\anc)V^\dag$, such that $\xi \le (1-w)\gamma/8$.
Then $\tilde{H}'= \frac{4}{3} \tilde{H}$ gap-simulates $(H,P)$ with encoding $V$, incoherence $\epsilon\le 32\xi/\gamma + 2\eta$, and energy spread $\tilde{w}\le (w + 2\xi/\gamma)/(1-2\xi/\gamma)$, per our Def.~\ref{defn:hamsimul}.
\end{lemma}

To show this, we first need to state a Lemma that bounds error of groundspace due to perturbations:

\begin{lemma}[Error bound on perturbed groundspace]
\label{lem:PPgroundspace}
Let $\tilde{H}$ and $\tilde{H}'$ be two Hamiltonians.
Per Def.~\ref{defn:gap}, let $\tilde{P}$ project onto a quasi-groundspace of $\tilde{H}$ with energy spread $\tilde{w}$ and quasi-spectral gap $\gamma$.
Assume $\tilde{w}\le 1/2$ and $\|\tilde{H}' - \tilde{H}\| \le \kappa$, where $\kappa \le (1-\tilde{w})\gamma/8$.
Then there is a quasi-groundspace projector $\tilde{P'}$ of $\tilde{H}'$ with quasi-spectral gap at least $\gamma'$, comprised of eigenstates of $\tilde{H}'$ up to energy at most $\lambda_1(\tilde{H}') + \tilde{w}'\gamma'$, where
\begin{equation}
\gamma' > \gamma-2\kappa, \quad
\tilde{w}'\gamma' \le \tilde{w}\gamma + 2\kappa,
\quad \text{and} \quad
\|\tilde{P}'-\tilde{P}\| < \frac{32\kappa}{\gamma}.
\end{equation}
\end{lemma}

While this may be simple to understand in the case of unique groundstates (see e.g. Lemma 2 of Ref.~\cite{BravyiHastingsSim}), it is not obvious when there are degenerate groundstates.
The proof of the above Lemma~\ref{lem:PPgroundspace} makes use of the Green's function machinery seen in Ref.~\cite{KKR06, OliveiraTerhal},
which we describe in a self-contained manner in
Appendix~\ref{sec:PPgroundspace-proof}.

\begin{proof}[\textbf{Proof of Lemma~\ref{lem:relate-CMP}}]
Given encoding of the form $\E(H) = V(H\otimes P_\anc) V^\dag$,
we write the corresponding encoding $\tilde{\E}(H)=\tilde{V}(H\otimes P_\anc) \tilde{V}^\dag$ such that $\|\tilde{V}-V\|\le\eta$.
Let $H_P \equiv \tilde{\E}(H)$ and $\tilde{H}_P \equiv \tilde{H} P_{\le{\Delta(\tilde{H})}}$; note both are Hermitian and hence (non-local) Hamiltonians.
Note $V (P\otimes P_\anc) V^\dag$ is a quasi-groundspace projector of $H_P$ with quasi-spectral gap $\gamma$ and energy spread $w$.
Since $\|\tilde{H}_P - H_P\| \le \xi$ by Def.~\ref{defn:CMPsimul}, then due to Lemma~\ref{lem:PPgroundspace}, there is a quasi-ground space projector $\tilde{P}$ of $\tilde{H}_P$ (and thus also $\tilde{H}$) with quasi-spectral gap at least $\gamma_P\ge \gamma - 2\xi$ with energy spread $\tilde{w}_P$, where $\tilde{w}_P\gamma_P\le w\gamma+ 2\xi$, and
\begin{equation}
\|\tilde{P} - \tilde{V} (P\otimes P_\anc) \tilde{V}^\dag \| \le 32\xi/\gamma.
\end{equation}

Note for any constant $\alpha >0 $, $\tilde{P}$ is also a quasi-groundspace projector of $\tilde{H}'=\alpha \tilde{H}$ with quasi-spectral gap $\ge\alpha \gamma_P $ and energy spread
\begin{equation}
\tilde{w} = \tilde{w}_P \le \frac{w \gamma+2\xi}{\gamma-2\xi} 
\end{equation}
To satisfy condition 1 of Def.~\ref{defn:hamsimul}, i.e. $\alpha\gamma_P \ge \gamma$, it suffices to choose $\alpha=\gamma/(\gamma-2\xi) \le 4/3$. For simplicity, we choose $\alpha=4/3$, as stated in the Lemma.

We note that since $\|V-\tilde{V}\|\le \eta$, we have
\begin{equation}
\|\tilde{P} - V (P\otimes P_\anc) V^\dag \| \le \|\tilde{P} - \tilde{V} (P\otimes P_\anc) \tilde{V}^\dag \| + 2\|\tilde{V}-V\| \le 32\xi/\gamma + 2\eta,
\end{equation}
satisfying condition 2 of Def.~\ref{defn:hamsimul} with $\epsilon \le 32\xi/\gamma + 2\eta$. Hence, $\tilde{H}'=8\tilde{H}$ gap-simulates $H$.
\end{proof}

We remark the constraints on $w$ and $\xi$ in the Lemma~\ref{lem:relate-CMP}  can be relaxed, since the Lemma~\ref{lem:PPgroundspace} used is a more restricted (but simpler) version of the more general Lemma~\ref{lem:PPgsv2} that we prove in Appendix.~\ref{sec:PPgroundspace-proof}.

The above Lemma~\ref{lem:relate-CMP} implies that our Definition~\ref{defn:hamsimul} is indeed a more relaxed version of the
simulation definitions from Ref.~\cite{BravyiHastingsSim,UniversalHamiltonian}, at least for real-valued Hamiltonians and sufficiently small simulation error $\xi\ll \gamma$.
In fact, our Definition~\ref{defn:hamsimul-incoherent} provides
an even more relaxed notion of simulation, where it is not required to preserve groundspace coherence or even all the groundstates.

\section{Properties of our Example Hamiltonian $H_\dicke$\label{sec:HamProperties}}
Here we prove the properties of $H_\dicke$ required for the impossibility proofs in this paper.
We start by reintroducing this Hamiltonian (first given in Eq.~\eqref{eq:Hdicke}).
Let us denote collective angular momentum operator on $n$ qubits as
\begin{equation}\label{eq:J}
\J_\alpha = \sum_{i=1}^n \sigma_\alpha^{(i)}/2,
\end{equation}
for $\alpha\in\{x,y,z\}$.
Our example family of 2-local $n$-qubit Hamiltonian $H_\dicke$ is the following Hamiltonian, restricted to even system size $n=2s$:
\begin{equation}
H_\dicke =  \J_z^2 - \frac12 \J^2 + b_n = \frac12(\J_z^2-\J_x^2-\J_y^2) +b_n = \frac14 \sum_{i<j}^{n} (\sigma_z^{(i)}\sigma_z^{(j)} - \sigma_x^{(i)}\sigma_x^{(j)} - \sigma_y^{(i)}\sigma_y^{(j)}) -\frac{n}{8} + b_n
\end{equation}
where $b_n\equiv \frac12 s(s+1) = 	\frac18 n (n + 2)$ is a constant chosen so the ground state energy is zero.
After expansion into sum of 2-local operators, this Hamiltonian has $M_0(n)=n(n-1)/2=\Omega(n^2)$ terms, and each qubit has degree $n-1$.
Since $[\vec{\J}^2,\J_z]=0$, the eigenstates of $H_B$ can be written in eigenbasis of both $\vec{\J}^2$ and $\J_z$.
Observe that $\J_z^2$ has eigenvalues $\{0,1,2^2,\ldots,s^2\}$ and $\vec{\J}^2$ has eigenvalues $\{s(s+1), (s-1)s,\ldots,6,2,0\}$.
The ground state is thus a state that has minimal $\J_z^2=0$ and maximal total angular momentum $J=s=\frac{n}{2}$.
Such a state is well-known in atomic physics as a Dicke state~\cite{Dicke}, and it is uniquely defined as
\begin{equation}
\ket{g_\dicke} = \ket{\J=\frac{n}{2}; \J_z = 0} =
\binom{n}{n/2}^{-1/2} \sum_{|\{i\,:\,x_i=1\}| = n/2} \ket{x_1\cdots x_n}.
\end{equation}
where the state can be explicitly written as a symmetric superposition of all strings $x$ with Hamming weight $h(x)=|\{i:x_i=1\}|=n/2$.
This ground state $\ket{g_\dicke}$ has energy $0$.
Meanwhile, all other eigenstates must have energy at least 1.
In particular, any eigenstate with $\vec{\J}^2<s(s+1)$ must have energy $\ge -\frac12 (s-1)s+b_n = s = \frac{n}{2} \ge 1$.
Thus, the system is spectrally gapped with energy spread $w_n=0$ and $\gamma_n = 1$.

\section{Information-Theoretical Impossibility Results}
In what follows, we will denote $X_i\equiv \sigma_x^{(i)}$ for simplicity and clarity.

\subsection{Impossibility of DR and Dilution with Close-to-Perfect Coherence (Lemma~\ref{lem:imposs1-DR} and Theorem~\ref{thm:imposs1-dilute})\label{sec:imposs1}}

In this section, we prove Lemma~\ref{lem:imposs1-DR} and Theorem~\ref{thm:imposs1-dilute} together,
essentially showing impossibility of DR and dilution for perfect coherence.
The proof of these results contains the idea of {\it contradiction-by-energy}, which is the seed to the idea for the proof of our main Theorem~\ref{thm:main} in the next section; in that proof, contradiction-by-energy is too weak, and instead we use the related idea of {\it contradiction-by-correlation}.

Towards proving Lemma~\ref{lem:imposs1-DR} and Theorem~\ref{thm:imposs1-dilute}, we prove a more general result
in the following Lemma~\ref{lem:imposs1},
of which Lemma~\ref{lem:imposs1-DR} and Theorem~\ref{thm:imposs1-dilute} are special cases obtained by setting $\epsilon=0$.
To this end, let us recall the definition of $H_\oneone$:
\begin{equation}
H_\oneone  = \left(\J_z+\frac{n}{2} \right)\left(\J_z+\frac{n}{2}-1\right),
\end{equation}
with $\J_z$ defined in Eq.~\eqref{eq:J}. The $n+1$ groundstates of $H_\oneone$ are
\begin{equation}
\ket{00\cdots00}, \ket{00\cdots01}, \ket{00\cdots10}, \ldots, \ket{10\cdots00}.
\end{equation}

\begin{lemma}[Limitation on $\epsilon$-incoherent degree-reduction and dilution of $H_\oneone$]
\label{lem:imposs1}
Suppose we require $\epsilon$-incoherence and energy spread $\tilde{w}$, then any $k$-local $[r,M,J]$-gap-simulator $\tilde{H}_\oneone$ of the $n$-qubit Hamiltonian $H_\oneone$ with localized encoding must satisfy at least one of the following conditions:
\begin{enumerate}
\item $\epsilon = 0$ and $\tilde{w} \ge 1/2$, or
\item $\epsilon > 0$ and $\|\tilde{H}_\oneone\| \ge  [1-2\tilde{w}(1+2\epsilon)]/(4\epsilon+6\epsilon^2)$, or
\item $\tilde{H}_\oneone$ contains qubits with degree $r=\Omega({n/k})$ and has a total number of terms $M=\Omega(n^2/k^2)$.
\end{enumerate}
\end{lemma}

In other words, the above Lemma shows that if we require inverse-polynomially small incoherence and some corresponding polynomial bound on the resources of gap-simulation, then it is impossible to degree-reduce or dilute $H_\oneone$.
In particular, if $\tilde{w}<1/2$, then for any $\xi>0$ and $p,q\ge 0$, there does not exists any $[O(1),O(n^p),O(n^q)]$-degree-reducer of $H_\oneone$ with $O(1/n^{p+q+\xi})$-incoherence, nor any $[r,o(n^2),O(n^q)]$-diluter of $H_\oneone$ with $O(1/n^{2+q})$-incoherence regardless of degree $r$.
To prove the above results, we first prove the following Lemma:

\begin{lemma}\label{lem:TotalCoherentImpossible}
Suppose $\tilde{H}_\oneone$ gap-simulates $H_\oneone$ with any encoding $V$, $\epsilon$-incoherence, such that either (a) $\epsilon=0$ and $\tilde{w} < 1/2$, or (b) $\|\tilde{H}_\oneone\| <  [1-2\tilde{w}(1+2\epsilon)]/(4\epsilon+6\epsilon^2)$.
For every original qubit $i$, let $S_i$ be the support of $V \sigma_x^{(i)} V^\dag$ on the interaction graph of $\tilde{H}_\oneone$.
Then for every pair of original qubits $(i,j)$, $\tilde{H}_\oneone$ must contain a term that acts nontrivially on both a qudit in $S_i$ and a qudit in $S_j$.
\end{lemma}
\begin{proof}
For the sake of contradiction, suppose $\tilde{H}_\oneone$ contains no term that interacts $S_i$ and $S_j$.
This means we can decompose $\tilde{H}_\oneone$ into two parts: $\tilde{H}_\oneone=\tilde{H}_{\oneone,i}+\tilde{H}_{\oneone,j}$, where $\tilde{H}_{\oneone,i}$ acts trivially on $S_i$.
In other words, $[\tilde{H}_{\oneone,i},\tilde{O}_i]=0$ for any operator $\tilde{O}_i$ whose support is contained in $S_i$.
Let us denote $\tilde{P}$ as the projector onto groundspace of $\tilde{H}_\oneone$, and $P$ the projector onto groundspace of $H_\oneone$.
Since we assume that $\tilde{H}_\oneone$ gap-simulates $H_\oneone$ with $\epsilon$-incoherence according to Def.~\ref{defn:hamsimul}, then for some projector $P_\anc$, we must have $\|\tilde{P}-Q\| \le \epsilon$, where $Q=V(P\otimes P_\anc)V^\dag$.

We write $P=\sum_{i=0}^n \ketbra{g_i}$, where states $\ket{g_0}=\ket{0\cdots0}$ and $\ket{g_i}=X_i\ket{g_0}=\ket{0\cdots01_i0\cdots0}$ are ground states of $H_\oneone$.
Let $\ket{\alpha}\in P_\anc$, and denote $\ket{\bar{g}_i} \equiv V\ket{g_i}\ket{\alpha}$ for $0\le i \le n$.
Observe that $Q\ket{\bar{g}_i}=V(P\otimes P_\anc)V^\dag V\ket{g_i}\ket{\alpha} = V\ket{g_i}\ket{\alpha}  = \ket{\bar{g_i}}$, and so
\begin{equation}
\tilde{P} \ket{\bar{g}_i}= (\tilde{P} - Q + Q)\ket{g_i}\ket{\alpha} =  (\tilde{P} - Q)\ket{\bar{g}_i} + \ket{\bar{g}_i} \equiv \ket{\epsilon_i} + \ket{\bar{g}_i}
\end{equation}
where we denoted $\ket{\epsilon_i} \equiv (\tilde{P}-Q)\ket{\bar{g}_i}$ satisfying $\|\ket{\epsilon_i}\|\le \epsilon$.
Now consider the state $\ket{e_{ij}}=X_iX_j\ket{g_0}$, which is an excited state of $H_\oneone$ outside groundspace $P$, and thus satisfies $P\ket{e_{ij}}=0$.
Consider correspondingly the state $\ket{\bar{e}_{ij}} = V\ket{e_{ij}}\ket{\alpha}$.
Observe that $Q\ket{\bar{e}_{ij}} = V(P\otimes P_\anc) V^\dag V\ket{e_{ij}}\ket{\alpha} = V(P\ket{e_{ij}}\ket{\alpha}) = 0$, and so
\begin{equation}
\tilde{P}^\perp \ket{\bar{e}_{ij}} = \ket{\bar{e}_{ij}} - \tilde{P} \ket{\bar{e}_{ij}} =\ket{\bar{e}_{ij}} - (\tilde{P}-Q+Q) \ket{\bar{e}_{ij}} = \ket{\bar{e}_{ij}} - (\tilde{P}-Q ) \ket{\bar{e}_{ij}} \equiv \ket{\bar{e}_{ij}} - \ket{\epsilon_{ij}},
\end{equation}
where we denoted $\ket{\epsilon_{ij}} \equiv (\tilde{P}- Q)\ket{\bar{e}_{ij}}$ satisfying $\|\ket{\epsilon_{ij}}\|\le \epsilon$.

Now, let $\tilde{X}_i = V X_i V^\dag$ be the encoded Pauli spin flip operator, which satisfies $[\tilde{H}_{\oneone,i},\tilde{X}_i]=0$.
Observe that $\tilde{X}_i\tilde{X}_j = VX_i X_j V^\dag = VX_j X_i V^\dag = \tilde{X}_j \tilde{X}_i$.
Additionally, $\tilde{X}_i^2=VV^\dag$, which acts like identity since $\tilde{X}_j^2\ket{\bar{g}_i} = V V^\dag V \ket{g_i}\ket{\alpha} = V\ket{g_i}\ket{\alpha} = \ket{\bar{g}_i}$, and similarly $\tilde{X}_k^2\ket{\bar{e}_{ij}}=\ket{e_{ij}}$.
Note that $\ket{\bar{g}_i} = \tilde{X}_i\ket{\bar{g}_0}$ and $\ket{\bar{e}_{ij}} = \tilde{X}_i\tilde{X}_j\ket{\bar{g}_0}$, for any $1\le i < j \le n$.
Then from the assumption that no term in $\tilde{H}_\oneone$ would interact supports $S_i$ and $S_j$, we can derive the following identity:
\begin{eqnarray}
\braket{\bar{e}_{ij}|\tilde{H}_\oneone|\bar{e}_{ij}} &=& \braket{\bar{g}_0|\tilde{X}_i\tilde{X}_j (\tilde{H}_{\oneone,i}+\tilde{H}_{\oneone,j})\tilde{X}_i\tilde{X}_j|\bar{g}_0}
	= \braket{\bar{g}_0|\tilde{X}_i\tilde{H}_{\oneone,j}\tilde{X}_i\tilde{X}_j^2|\bar{g}_0} + \braket{\bar{g}_0|\tilde{X}_j\tilde{H}_{\oneone,i}\tilde{X}_j\tilde{X}_i^2|\bar{g}_0} \nonumber \\
&=& \braket{\bar{g}_0|\tilde{X}_i\tilde{H}_{\oneone,j}\tilde{X}_i|\bar{g}_0} + \braket{\bar{g}_0|\tilde{X}_j\tilde{H}_{\oneone,i}\tilde{X}_j|\bar{g}_0} \nonumber \\
&=& \braket{\bar{g}_i|\tilde{H}_{\oneone,j}|\bar{g}_i} + \braket{\bar{g}_j|\tilde{H}_{\oneone,i}|\bar{g}_j} \nonumber\\
&=& \braket{\bar{g}_i|\tilde{H}_{\oneone}|\bar{g}_i} + \braket{\bar{g}_j|\tilde{H}_{\oneone}|\bar{g}_j} - \braket{\bar{g}_i|\tilde{H}_{\oneone,i}|\bar{g}_i} - \braket{\bar{g}_j|\tilde{H}_{\oneone,j}|\bar{g}_j} \nonumber \\
&=& \braket{\bar{g}_i|\tilde{H}_{\oneone}|\bar{g}_i} + \braket{\bar{g}_j|\tilde{H}_{\oneone}|\bar{g}_j} - \braket{\bar{g}_0|\tilde{X}_i\tilde{H}_{\oneone,i}\tilde{X}_i|\bar{g}_0} - \braket{\bar{g}_0|\tilde{X}_j\tilde{H}_{\oneone,j}\tilde{X}_j|\bar{g}_0} \nonumber \\
\braket{\bar{e}_{ij}|\tilde{H}_\oneone|\bar{e}_{ij}} &=& \braket{\bar{g}_i|\tilde{H}_{\oneone}|\bar{g}_i} + \braket{\bar{g}_j|\tilde{H}_{\oneone}|\bar{g}_j} - \braket{\bar{g}_0|\tilde{H}_{\oneone}|\bar{g}_0}.
\label{eq:lemma-2-energy-trick}
\end{eqnarray}

To simplify expressions, let us denote $\tilde{H}_\oneone' \equiv \tilde{H}_\oneone - \tilde{E}^g \Id$, where $\tilde{E}^g \equiv \lambda_1(\tilde{H}_\oneone)$ is groundstate energy of $\tilde{H}_\oneone$.
We note the above identity of Eq.~\eqref{eq:lemma-2-energy-trick} remains true if we replace $\tilde{H}_\oneone$ with $\tilde{H}_\oneone'$, since the constant offsets cancel.
Now let us consider the energy of states $\ket{\bar{g}_i}$ and $\ket{\bar{e}_{ij}}$ with respect to $\tilde{H}_\oneone'$.
Since we allow energy spread $\|\tilde{P}\tilde{H}_\oneone'\tilde{P}\|\le \tilde{w}$ for the gap-simulation, and the spectral gap of $H_\oneone$ is $\gamma=1$, we must have
\begin{equation}
\label{eq:lemma-2-energy-bound}
0 \le \braket{\bar{g}_i|\tilde{P}\tilde{H}_\oneone'\tilde{P}|\bar{g}_i}\le\tilde{w}.
\end{equation}
And keeping in mind that $\tilde{H}_\oneone' \tilde{P} = \tilde{P}\tilde{H}_\oneone' \tilde{P}$, and $\braket{\psi|\tilde{H}_\oneone'|\psi}\ge 0$ for any state $\ket{\psi}$ because $\tilde{H}_\oneone'$ is positive semi-definite, we have
\begin{eqnarray}
0 \le \braket{\bar{g}_i|\tilde{H}_{\oneone}'|\bar{g}_i} &=& (\bra{\bar{g}_i}\tilde{P} - \bra{\epsilon_i}))\tilde{H}_\oneone '(\tilde{P}\ket{\bar{g}_i} - \ket{\epsilon_i}) = \braket{\bar{g}_i|\tilde{P}\tilde{H}_\oneone'\tilde{P}|\bar{g}_i} - 2\Re \braket{\epsilon_i|\tilde{H}_\oneone'\tilde{P}|\bar{g}_i} + \braket{\epsilon_i | \tilde{H}_\oneone' | \epsilon_i} \nonumber \\
&\le & \tilde{w} + 2\epsilon\tilde{w} + \epsilon^2 \|\tilde{H}_\oneone'\|= \tilde{w}(1+2\epsilon) + \epsilon^2 \|\tilde{H}_\oneone'\|.
\end{eqnarray}
Furthermore, we observe that
\begin{eqnarray}
\braket{\bar{e}_{ij} | \tilde{P}^\perp \tilde{H}_\oneone' \tilde{P}^\perp | \bar{e}_{ij}} &=& (\bra{\bar{e}_{ij}}-\bra{\epsilon_{ij}}) \tilde{H}_\oneone' (\ket{\bar{e}_{ij}}-\ket{\epsilon_{ij}}) = \braket{\bar{e}_{ij}|\tilde{H}_\oneone' |\bar{e}_{ij}} - 2\Re\braket{\epsilon_{ij}|\tilde{H}_\oneone'|\bar{e}_{ij}} + \braket{\epsilon_{ij}|\tilde{H}_\oneone'|\epsilon_{ij}}\nonumber \\
&\le& \braket{\bar{e}_{ij}|\tilde{H}_\oneone' |\bar{e}_{ij}} + 2\epsilon\|\tilde{H}_\oneone' \| + \epsilon^2 \|\tilde{H}_\oneone'\| \\
&=& \braket{\bar{g}_i|\tilde{H}_{\oneone}' |\bar{g}_i} + \braket{\bar{g}_j|\tilde{H}_\oneone' |\bar{g}_j} - \braket{\bar{g}_0|\tilde{H}_{\oneone}' |\bar{g}_0} + (2\epsilon + \epsilon^2)\|\tilde{H}_\oneone' \|  \nonumber \\
&\le& 2\tilde{w}(1+2\epsilon) + (2\epsilon+3\epsilon^2)\|\tilde{H}_\oneone' \| \le 2\tilde{w}(1+2\epsilon) + (4\epsilon+6\epsilon^2)\|\tilde{H}_\oneone\|.
\label{eq:lemma2-excited-energy-bound}
\end{eqnarray}
where we used the identity \eqref{eq:lemma-2-energy-trick} and the fact that $\|\tilde{H}_\oneone'\| = \|\tilde{H}_\oneone-\tilde{E}_g\| \le 2\|\tilde{H}_\oneone\|$.
This implies $\tilde{P}^\perp\tilde{H}_\oneone'\tilde{P}^\perp$ has an eigenvalue $\le 2\tilde{w}(1+2\epsilon) + (4\epsilon+6\epsilon^2)\|\tilde{H}_\oneone\| $.
This contradicts the gap-simulation assumption $\lambda_j(\tilde{P}^\perp \tilde{H}_\oneone'\tilde{P}^\perp+\gamma \tilde{P})\ge \gamma=1$ if
\begin{equation}
2\tilde{w}(1+2\epsilon) + (4\epsilon+6\epsilon^2)\|\tilde{H}_\oneone\|  < 1
\quad \Longleftrightarrow \quad
\begin{dcases}
\tilde{w} < 1/2, &\text{ if } \epsilon = 0\\
\|\tilde{H}_\oneone\| < \frac{1-2\tilde{w}(1+2\epsilon)}{4\epsilon+6\epsilon^2}, &\text{ if } \epsilon > 0
\end{dcases}.
\end{equation}
Hence, if either (a) $\epsilon=0$ and $\tilde{w}<1/2$, or (b) $\|\tilde{H}_\oneone\| < [1-2\tilde{w}(1+2\epsilon)]/(4\epsilon+6\epsilon^2)$, then $\tilde{H}_\oneone$ must contain a term that acts nontrivially on both qubit $i$ and $j$.
\end{proof}

\begin{proof}[\textbf{Proof of Lemma~\ref{lem:imposs1}}]
Suppose a gap-simulator of $\tilde{H}_\oneone$ of $H_\oneone$ does not satisfy any of the first two conditions enumerated in Lemma~\ref{lem:imposs1}, then
it must either (a) has $0$-incoherence and energy spread $\tilde{w}<1/2$,
or (b) $\|\tilde{H}_\oneone\| < [1-2\tilde{w}(1+2\epsilon)]/(4\epsilon+6\epsilon^2)$.
Thus, by Lemma~\ref{lem:TotalCoherentImpossible} above,
there must be at least $n-1$ terms, each interact a qudit in $S_1$ with a qudit in $S_2,S_3,\ldots S_{n}$.

Let us consider the first variant of localized encoding $V=\bigotimes_i V_i$, where the range of $V_i$ is supported by $O(1)$ qudits in $\tilde{H}$.
Here, the supports $S_i$ are mutually disjoint, with bounded maximum size $\max_i |S_i|\le a = O(1)$.
Since each $k$-local term can couple a qubit to up to $k-1$ other qubits, the average degree of qudits in $S_1$ is $r\ge (\tilde{n}-1)/[a(k-1)]=\Omega(n/k)$.
Furthermore, note that there are $\binom{n}{2}$ required pairwise interactions between supports $(S_i, S_j)$.
Since each $k$-local term can act on up to $k$ qubits, it can cover up to $\binom{k}{2}$ such pairwise interactions.
Thus, the minimum number of terms in $\tilde{H}_\oneone$ to account for all the pairwise interactions of $H_\oneone$ is
\begin{equation}
 M \ge \frac{\binom{n}{2}}{\binom{k}{2}} = \frac{n(n-1)}{k(k-1)} = \Omega(n^2/k^2).
\end{equation}

To prove the Lemma for the second variant of localized encoding where $V$ is a constant-depth quantum circuit, we modify the above argument by considering $\tilde{H}_\oneone' = V^\dag \tilde{H}_\oneone V$.
Note $\tilde{H}_\oneone$ gap-simulates $H_\oneone$ with trivial encoding.
Since $V^\dag$ is also a constant-depth quantum circuit, one can see that each term in $\tilde{H}_\oneone$ is
mapped into a term in $\tilde{H}_\oneone'$ whose locality blows up by a constant factor.
Hence, if $\tilde{H}_\oneone$ has maximum degree $r$ and $M$ Hamiltonian terms that are $k$-local,
then $\tilde{H}_\oneone'$ has maximum degree $r'=\Theta(r)$, and $M'=M$ terms that are $k'=\Theta(k)$-local.
Since the encoding is trivial, then $\tilde{H}$ must interact every pairs of qubit $(i,j)$.
This would imply $r'=\Omega(n/k')$, and $M'=\Omega(n^2/(k')^2)$.
Consequently, $r=\Omega(n/k)$, and $M=\Omega(n^2/k^2)$, proving our Lemma.
\end{proof}

\paragraph{Remark}---
Note that there is a difficulty to extend the proof of Lemma~\ref{lem:imposs1} (and thus Lemma~\ref{lem:imposs1-DR} and Theorem~\ref{thm:imposs1-dilute}) to the case where we allow $\epsilon=\Theta(1)$-incoherence, even if we require bounded interaction strength $J=O(1)$.
This difficulty is apparent in Eq.~\eqref{eq:lemma2-excited-energy-bound}, where the bound on the excited state's energy has an energy uncertainty on the order of $\O(\epsilon\|\tilde{H}_A\|)$, which would grow as system size due to the dependence on $\|\tilde{H}_\oneone\|$.
Hence, in order to extend this impossibility result to $\epsilon=\Theta(1)$-incoherence, more innovation is required -- this is done in the next section.

\subsection{Impossibility of DR with Constant Coherence or Faithfulness (Theorem~\ref{thm:main})\label{sec:imposs2}}

From Lemma~\ref{lem:imposs1-DR}, it appears that if perfect coherence is required, any meaningful reduction in the maximum degree or the total number of terms cannot in general be possible due to our first counterexample $H_\oneone$.
Here, we strengthen this to $\epsilon$-incoherence for constant $\epsilon$:
We show that reduction in the maximum degree remains impossible for $H_\oneone$, by arriving at a contradiction via a correlation-based argument,
rather than one relying on the energy.
Furthermore, impossibility of incoherent degree-reduction can also be shown by applying the same idea, now to our second counterexample $H_\dicke$ (see Appendix~\ref{sec:HamProperties} for its properties),
which has a unique groundstate (so incoherent and coherent degree-reduction are equivalent due to Lemma~\ref{lem:equiv}).
This is our main impossibility result:

{
\renewcommand{\thethm}{\ref{thm:main}}
\begin{thm}[Main: Impossibility of constant coherence (faithfulness) DR for
$H_\oneone$ ($H_\dicke$)]
For sufficiently small constants $\epsilon\ge0$ $(\delta\ge0)$ and $\tilde{w}\ge0$,
there exists system size $n_0$ where for any $n\ge n_0$, there is no
$O(1)$-local $[O(1),M,O(1)]$-degree-reducer of the $n$-qubit Hamiltonian $H_\oneone$ $(H_\dicke)$
with localized encoding,  $\epsilon$-incoherence ($\delta$-unfaithfulness),
and energy spread $\tilde{w}$, for any number of Hamiltonian terms $M$.
\end{thm}
\addtocounter{thm}{-1}
}

To prove  Theorem~\ref{thm:main},
we rely on the Hastings-Koma result\cite{HastingsKoma} demonstrating exponential decay of correlation in a spectrally gapped groundspace of Hamiltonians with {\it exponentially decaying interaction}, which we define below:
\begin{defn}[Exponentially decaying interaction, adapted from~\cite{HastingsKoma}]
\label{defn:expdecayint}
Consider a graph given by $G=(\V, \E)$, where $\V$ is a set of vertices and $\E=\{(i,j):i,j\in \V\}$ is a set of edges.
A Hamiltonian $H=\sum_{X\subset V} h_X$ defined on such a graph $G$ has \emph{exponentially decaying interaction} if $h_X$ satisfies
\begin{equation}
\sup_x \sum_{X\ni x} \|h_X\| |X| \exp[\mu \diam(X)] \le s_1 < \infty
\end{equation}
for positive constants $\mu$ and $s_1$. Here $\diam(X)=\max_{x,y\in X} \dist(x,y)$, and $\dist(x,y)$ is the graph-theoretic distance.
\end{defn}

It can be seen that any local Hamiltonian with constant degree and bounded interaction strength satisfies the criterion in the above definition:
\begin{lemma}
\label{lem:HamExpDecayInt}
Any $n$-qudit $[k=O(1)]$-local Hamiltonian with maximum degree $r=O(1)$ and bounded interaction strength $J=O(1)$ has exponentially decaying interaction per Definition~\ref{defn:expdecayint}.
\end{lemma}
\begin{proof}
Let us construct the graph on which we embed the Hamiltonian.
The set of vertices $V$ corresponds to the set of qudits.
We can then write the Hamiltonian as $H = \sum_{X\subset \V} h_X$, where $|X|\le k$ since the Hamiltonian is $k$-local.
We then choose the set of edges as $\E=\{(x,y):  x,y\in X \text{ for some } h_X\}$.
In other words, for any set of qubits that is directly interacting through a term in the Hamiltonian, we assign a clique to their vertices on the graph.
Then, $H$ has exponential decaying interaction on this graph $G=(\V,\E)$ per Definition~\ref{defn:expdecayint} since
\begin{equation}
\sup_x \sum_{X\ni x} \|h_X\| |X| \exp[\mu\diam(X)] \le \sum_{i=1}^r J k e^\mu = r J k e^{\mu} = O(1) < \infty,
\end{equation}
where we used the fact that each qudit is contained in at most $r$ terms by definition of Hamiltonian degree, and that each term $h_X$ has norm $\|h_X\|\le J$, acts on at most $|X|\le k$ qudits with diameter $\diam(X)=1$.
\end{proof}

We now give a strengthened version of the Hastings-Koma result, that we will use in the proof of Theorem~\ref{thm:main}.

\begin{lemma}[Hastings-Koma theorem for non-zero energy spread, generalized from Ref.~\cite{HastingsKoma}]
\label{lem:HastingsKoma}
Suppose we have a $n$-qudit Hamiltonian defined on a graph $G=(\V,\E)$ with exponential decaying interactions (Def.~\ref{defn:expdecayint}).
Also suppose for some constants $0\le w_\infty < 1$ and $\gamma_\infty >0$ independent of system size $n$, the Hamiltonian is quasi-spectrally gapped (Def.~\ref{defn:gap}) with energy spread $w_n \le w_\infty$ and quasi-spectral gap $\gamma_n \ge \gamma_\infty$.
Let $P_0$ be the projector onto the corresponding quasi-groundspace.
Let $A_X$ and $B_Y$ be observables with bounded norm $\|A_X\|,\|B_Y\|=O(1)$ and compact support $X,Y\subset V$, where $[A_X,B_Y]=0$ and $X\cap Y=\emptyset$. Then there exists some constants $C,\tilde{\mu}>0$, independent of $n$, such that for any normalized quasi-groundstate $\ket{\psi}\in P_0$, we have
\begin{equation}
Ce^{-\tilde{\mu}\dist(X,Y)} \ge \left| \braket{\psi|A_X B_Y|\psi} - \frac12 (\braket{\psi|A_XP_0B_Y|\psi} + \braket{\psi|B_YP_0A_X|\psi}) + \O(w_\infty\log^{\frac12}\frac{1}{w_\infty})\right|
\label{eq:MHK-inequality}
\end{equation}
In the case when $w_\infty=0$, we can ignore the $\O( w_\infty \log^{\frac12} (1/w_\infty))$ term.
\end{lemma}
Our proof of this Theorem, which is modified from the proof of Theorem 2.8 in Ref.~\cite{HastingsKoma}, can be found in Appendix~\ref{sec:MHK}.
Note the apparent singularity of $1/w_\infty$ in the last term of Eq.~\eqref{eq:MHK-inequality} is somewhat artificial, since
$w\log^{1/2}(1/w) = \O(w^{1-\epsilon})\to 0$ as $w\to 0$.
Its appearance is
due to our decision to consider the case where the energy spread is non-zero, even when the system size $n\to\infty$.

Lastly, we prove the following property for constant-degree gap-simulation with localized encoding:

\begin{lemma}\label{lem:encoded-support}
Let $\tilde{H}$ gap-simulates $H$ with a some encoding $V$.
Let $S_i$ be the support of $V \sigma_i V^\dag$ on the interaction graph of $\tilde{H}$, where $\sigma_i$ is any operator acting on the $i$-th original qudit.
Suppose $\tilde{H}$ has maximum degree $r$, and $\max_i |S_i| = a$.
Then there exist two qudits $L$ and $R$ where the distance between the sets $S_L$ and $S_R$ (in the graph metric)
satisfies $\dist(S_L, S_R)\ge \log(n/a)/\log(1+r+ra)$.
Specifically, for constant degree $r=O(1)$ and localized encoding $a=O(1)$, $\dist(S_L,S_R)=\Omega(n)$.
\end{lemma}
\begin{proof}
We define a sequence of subset of qudits, $R_\ell$, and let $R_0=S_1$.
For $\ell=1,2,\ldots$, we form $R_\ell$ by joining to $R_{\ell-1}$ both (1) all qudits $\le 1$ distance to any qudit in $R_{\ell-1}$, and (2) any $S_{i}$ containing qudit(s) with $\le 1$ distance to any qudit in $R_{\ell-1}$.
In other words
\begin{equation}
R_\ell = R_{\ell-1} \cup \{v: \exists w\in R_{\ell-1}\text{ s.t. } \dist(v,w)\le 1\} \cup \{ S_i: \exists v \in S_i, w \in R_{\ell-1} \text{ s.t. } \dist(v,w)\le 1\}
\end{equation}
By construction, if there exists $S_i\not\subset R_{\ell-1}$, then the set difference $R_{\ell} \backslash R_{\ell-1}$ must contain a support $S_j$ where $\dist(S_1,S_j)\ge \ell$.

Note that $|R_\ell| \le |R_{\ell-1}|(1+r+ra)$. Since $|R_0|=|S_1|\le a$, we have
\begin{equation}
|R_\ell| \le a (1+r+ra)^\ell.
\end{equation}
Since we have $\left|\bigcup_{i=1}^n S_i\right| \ge n $ qudits, then in order to cover all the supports, we must have
\begin{equation}
|R_\ell|\ge n  \quad \Longrightarrow  \quad \ell \ge \frac{\log(n/a)}{\log(1+r+ra)}
\end{equation}
For $r,a=O(1)$, this shows that there exists a support $S_j$ such that $\dist(S_1, S_j) \ge \Omega(\log(n))$.
We note that more generally,
for $r=O(\log^c n)$ and $a=O(\log^c n)$, the above also shows that there exists a support $S_j$ such that $\dist(S_1, S_j) \ge \Omega(\log(n)/\log(\log n))$.
\end{proof}

We are now ready to prove our main theorem.

\begin{proof}[\textbf{Proof of Theorem~\ref{thm:main}}]

\textbf{Part I}---
We first show impossibility of coherent degree-reduction for $H_\oneone$.
For the sake of contradiction, suppose there exists an $O(1)$-local $[O(1),M,O(1)]$-degree-reducer $\tilde{H}_\oneone$ of $H_\oneone$ with localized encoding $V$, $\epsilon$-incoherence and energy spread $\tilde{w}$, but without restriction on the number of terms $M$.
Then it has exponential decaying interaction due to Lemma~\ref{lem:HamExpDecayInt}.
Additionally, since the original Hamiltonian $H_\oneone$ is spectrally gapped with gap $\gamma_n=\gamma_\infty =1$, the gap-simulator should also be quasi-spectrally gapped with gap $\gamma_n=\gamma_\infty =1$ in order to gap-simulate its groundspace.
Nevertheless, we may allow some small and possibly non-zero energy spread $\tilde{w} = \tilde{w}_n\le \tilde{w}_\infty$ for the gap-simulator.
Since we assumed in the premise of the Theorem that $\tilde{w}$ is sufficiently small, it follows that $\tilde{w}_\infty = \sup_n \tilde{w}_n$ should also be a small constant $<1$.
Hence, $\tilde{H}_\oneone$ satisfy the requirements for applying Lemma~\ref{lem:HastingsKoma}.

Let us denote $Q = V(P\otimes P_\anc)V^\dag$ as the encoded groundspace projector.
Since we also require $\epsilon$-incoherence, the groundspace projector $\tilde{P}$ of $\tilde{H}_\oneone$ satisfies $\|\tilde{P}-Q\| = \|\tilde{P} - V(P\otimes P_\anc)V^\dag \|\le \epsilon$, where $P$ is the groundspace projector of $H_\oneone$, and $P_\anc$ is some projector on the ancilla.
Consider the unencoded $X_i$ operator on the original qubit $i$, which corresponds to $\tilde{X}_i = V X_i V^\dag$ in the encoded Hamiltonian. Let the support of the observable $\tilde{X}_i$ be $S_i$.
Because of the assumption of constant degree and the fact that the encoding is localized, there exists two qubit $L$ and $R$ where $\dist(S_L,S_R)\ge K \log n$ for some constant $K>0$ by Lemma~\ref{lem:encoded-support}.
Consider the following \emph{approximate} groundstates of $\tilde{H}$
\begin{equation}
\ket{g_{00}} = V \ket{0_L0_R}\ket{0\cdots}\ket{a}_\anc, \quad
\ket{g_{01}} = V \ket{0_L1_R}\ket{0\cdots}\ket{a}_\anc, \quad
\ket{g_{10}} = V \ket{1_L0_R}\ket{0\cdots}\ket{a}_\anc
\end{equation}
where $P_\anc\ket{a}=\ket{a}$, so $Q\ket{g_{ij}} = V(P\otimes P_\anc)V^\dag \ket{g_{ij}}=\ket{g_{ij}}$. Also let us denote
\begin{equation}
\ket{e_{11}} = V \ket{1_L1_R}\ket{0\cdots}\ket{a}
\end{equation}
which satisfies $Q\ket{e_{11}}=0$.
Now consider an approximate groundstate of $\tilde{H}_\oneone$
\begin{equation}
\ket{\phi} = \frac{1}{\sqrt{2}}(\ket{g_{01}}+\ket{g_{10}}).
\end{equation}
Let $\tilde{X}_L = V X_L V^\dag$ and $\tilde{X}_R = V X_R V^\dag$.
It's easy to see that $\tilde{X}_R\ket{g_{01}} = \ket{g_{00}}$, etc.
Observe
\begin{eqnarray}
\braket{\phi|\tilde{X}_L \tilde{X}_R|\phi} &=& 1 \\
\braket{\phi|\tilde{X}_R\tilde{P}\tilde{X}_L|\phi} = \braket{\phi|\tilde{X}_L\tilde{P}\tilde{X}_R|\phi} &=& \frac{1}{2} (\bra{e_{11}}+\bra{g_{00}})\tilde{P}(\ket{g_{00}}+\ket{e_{11}})
\end{eqnarray}
Due to the assumption of $\epsilon$-incoherence, we require
\begin{eqnarray*}
\braket{g_{00}|\tilde{P}|g_{00}} &=& \|\tilde{P}\ket{g_{00}}\|^2
\le \left(\|Q\ket{g_{00}}\| + \|(\tilde{P}-Q\ket{g_{00}}\| \right)^2 \le (1+\epsilon)^2, \\
\braket{e_{11}|\tilde{P}|e_{11}} &=& \|\tilde{P}\ket{e_{11}}\|^2 = \|(\tilde{P}- Q)\ket{e_{11}}\|^2 \le \epsilon^2, \\
\left|\braket{e_{11}|\tilde{P}|g_{00}}\right| &=&
\left|\braket{e_{11}|\tilde{P} - Q|g_{00}}\right| \le \|\tilde{P}-Q\| \le \epsilon.
\end{eqnarray*}
Hence
\begin{equation}
\braket{\phi|\tilde{X}_L\tilde{P}\tilde{X}_R|\phi} = \frac{1}{2} \left(\braket{g_{00}|\tilde{P}|g_{00}} + \braket{e_{11}|\tilde{P}|e_{11}} + 2\Re[\braket{e_{11}|\tilde{P}|g_{00}}]\right) \le \frac{1}{2}+\O(\epsilon)
\end{equation}
Note in order to apply the Hastings-Koma theorem, we need to convert $\ket{\phi}$ into an actual groundstate of $\tilde{H}$, i.e. a state fixed by $\tilde{P}$.
Again using the $\epsilon$-incoherence condition, we have
\begin{eqnarray*}
\epsilon \ge \|(\tilde{P}-Q)\ket{\phi}\| = \|\tilde{P}\ket{\phi}-\ket{\phi}\| \quad \Longrightarrow \quad \tilde{P}\ket{\phi} = \ket{\phi} + \ket{\epsilon}
\end{eqnarray*}
where $\|\ket{\epsilon}\| \le \epsilon$. Now let
\begin{equation}
\ket{\psi} = \frac{\tilde{P}\ket{\phi} }{\|\tilde{P}\ket{\phi} \|} = \mathcal{N}(\ket{\phi}+\ket{\epsilon})
\end{equation}
be a normalized state in the groundspace, where $\mathcal{N} \equiv \|\ket{\phi}+\ket{\epsilon}\|^{-1}$ is a normalization constant satisfying $(1+\epsilon)^{-1} \le \mathcal{N} \le  (1-\epsilon)^{-1}$.
Thus
\begin{eqnarray}
\braket{\psi|\tilde{X}_L \tilde{X}_R|\psi} &=& \mathcal{N}^2\braket{\phi|\tilde{X}_L\tilde{X}_R|\phi} + \O(\epsilon) \ge 1 + \O(\epsilon)\\
\braket{\psi|\tilde{X}_L \tilde{P} \tilde{X}_R|\psi} &=& \mathcal{N}^2\braket{\phi|\tilde{X}_L\tilde{P}\tilde{X}_R|\phi} + \O(\epsilon) \le \frac{1}{2} + \O(\epsilon)
\end{eqnarray}
For any given small constant $\epsilon\ge0$ and $\tilde{w}_\infty\ge 0$,
the existence of $\tilde{H}_\oneone$ contradicts Lemma~\ref{lem:HastingsKoma} since
\begin{eqnarray}
&& \left|\braket{\psi|\tilde{X}_L\tilde{X}_R|\psi} - \frac{1}{2}(\braket{\psi|\tilde{X}_L\tilde{P}\tilde{X}_R|\psi} + c.c.) + \O(\tilde{w}_\infty\log^{\frac12}\frac{1}{\tilde{w}_\infty})  \right| \ge \left|1-\frac{1}{2} + \O(\epsilon) +\O(\tilde{w}_\infty\log^{\frac12}\frac{1}{\tilde{w}_\infty})  \right| \nonumber \\
&\ge& \left| \frac{1}{2} + \O(\epsilon) + \O(\tilde{w}_\infty\log^{\frac12}\frac{1}{\tilde{w}_\infty}) \right| \not\le C\exp(-\tilde\mu K \log n) = Cn^{-\tilde{\mu} K}.
\label{eq:contradict-correlation-A}
\end{eqnarray}
This contradiction arises because for sufficiently small $\epsilon$ and $\tilde{w}_\infty$, the LHS of the final inequality is a constant of roughly $1/2$, which is not less than the RHS $Cn^{-\tilde{\mu} K}$ when $n> n_0$, for some cutoff system size $n_0 \approx (C/2)^{1/\tilde{\mu} K}$.

\textbf{Part II}---
It remains to show the impossibility of incoherent degree-reduction for $H_\dicke$.
Suppose FSOC there exists an $O(1)$-local $[O(1),M,O(1)]$-degree-reducer $\tilde{H}_\dicke$ of $H_\dicke$ with localized encoding $V$, $\delta$-unfaithfulness and energy spread $\tilde{w}$.
Since there $H_\dicke$ has a unique groundstate $\ket{g_\dicke}$, $\tilde{H}_\dicke$ must also have $\epsilon$-incoherence for some $\epsilon\le\sqrt{2}\delta/\sqrt{1-\delta^2} = \O(\delta)$ due to Lemma~\ref{lem:equiv}.
In other words, there must exists a projector $P_\anc^\dicke$ onto the ancilla Hilbert space such that $\|\tilde{P}_\dicke - Q_\dicke\|\le\epsilon$, where $Q_\dicke = V(P_\dicke \otimes P_\anc^\dicke)V^\dag$, $P_\dicke = \ketbra{g_\dicke}$, and $\tilde{P}_\dicke$ is groundspace projector of $\tilde{H}_\dicke$.
Then consider the approximate groundstate $\ket{\phi_\dicke}= V \ket{g_\dicke}\ket{b}$ for some ancilla state $\ket{b}=P_\anc^\dicke \ket{b}$.
By Lemma~\ref{lem:encoded-support}, there exists two original qubits $L,R$ such that the support $S_i$ of the encoded observable $\tilde{X}_i = V X_i V^\dag$ satisfies $\dist(S_L,S_R)\ge K\log n$.
Let us denote $h(x)$ as the Hamming weight (number of 1s) of the binary string $x$.
Observe that
\begin{eqnarray}
&& \braket{\phi_\dicke|\tilde{X}_L \tilde{X}_R|\phi_\dicke} = \braket{g_\dicke|X_L X_R|g_\dicke} = \binom{n}{n/2}^{-1} \sum_{x, y: h(x)=h(y) = \frac{n}{2}} \bra{x} X_L X_R\ket{y}  \nonumber \\
&=&
\text{\scalebox{0.9}{
$\binom{n}{n/2}^{-1}  \sum_{x: h(x) = n/2} \bra{x} X_L X_R \left[\sum_{h(\bar{y})=\frac{n}{2}-1} (\ket{0_L1_R} + \ket{1_L 0_R})  \ket{\bar{y}} + \sum_{h(\bar{y}')=\frac{n}{2}}\ket{0_L 0_R}\ket{\bar{y}'} + \sum_{h(\bar{y})=\frac{n}{2}-2}\ket{1_L 1_R} \ket{\bar{y}''} \right]$
}}
\nonumber \\
&=& \binom{n}{n/2}^{-1}  \sum_{x: h(x) = n/2} \bra{x} \sum_{\bar{y} : h(\bar{y})=\frac{n}{2}-1} (\ket{1_L0_R} + \ket{0_L 1_R})  \ket{\bar{y}} = \binom{n}{n/2}^{-1} \times 2\binom{n-1}{n/2-1} \nonumber \\
&=& \frac{n}{2(n-1)} \ge \frac12
\end{eqnarray}
Additionally, note $P_\dicke X_i\ket{g_\dicke}=0$ since $X_i$ changes the Hamming weight of all strings in $\ket{g_\dicke}$.
Consequently, $Q_\dicke \tilde{X}_i\ket{\phi_\dicke} = V [(P_\dicke X_i\ket{g_\dicke})\otimes\ket{b}]=0$, and
\begin{equation}
\left|\braket{\phi_\dicke |\tilde{X}_L \tilde{P}_\dicke \tilde{X}_R | \phi_\dicke}\right| = \left|\braket{\phi_\dicke |\tilde{X}_L (\tilde{P}_\dicke - Q_\dicke ) \tilde{X}_R | \phi_\dicke}\right| \le \|\tilde{P}_\dicke - Q_\dicke \| \le \epsilon = \O(\delta)
\end{equation}
The rest of the proof for impossibility of $\tilde{H}_\dicke$ follows the same argument as in the previous part for $\tilde{H}_\oneone$, where we convert $\ket{\phi_\dicke}$ to a true groundstate $\ket{\psi_\dicke} \propto \tilde{P}_\dicke \ket{\phi_\dicke}$ with up to $\O(\epsilon)=\O(\delta)$ error.
Then
the existence of $\tilde{H}_\dicke$ contradicts Lemma~\ref{lem:HastingsKoma} since
\begin{eqnarray}
&&\left| \braket{\psi_\dicke|\tilde{X}_L \tilde{X}_R|\psi_\dicke} - \frac{1}{2}(\braket{\psi_\dicke|\tilde{X}_L\tilde{P}_\dicke \tilde{X}_R|\psi_\dicke} + c.c.) + \O(\tilde{w}_\infty \log^{\frac12}\frac{1}{\tilde{w}_\infty}) \right| \nonumber \\
&\ge& \frac12+ \O(\delta) + \O(\tilde{w}_\infty \log^{\frac12}\frac{1}{\tilde{w}_\infty}) \not\le C\exp(-\tilde\mu K \log n) = Cn^{-\tilde{\mu} K}
\label{eq:contradict-correlation-B}
\end{eqnarray}
for sufficiently large $n>n_0$, where $n_0\approx (C/2)^{1/\tilde{\mu} K}$.
\end{proof}

We remark that Theorem~\ref{thm:main} in fact holds for encodings which is \emph{quasi}-localized, i.e. when the support of encoded local variable $|S_i|=O(\poly\log n)$.
This is because by Lemma~\ref{lem:encoded-support}, there exists two original qubits $L,R$ whose encoded support are distance $\Omega(\log n/\log (\log n)))$ apart for quasi-localized encodings, which would lead to a contradiction by correlation decay since $e^{-K\log n/\log (\log n))}\to0$ as $n\to\infty$.

\subsection{Impossibility of Full-Spectrum Degree-Reduction\label{sec:imposs-full-spec-DR}}
The above result of Theorem~\ref{thm:main} can be used to also rule out degree-reduction of $H_\oneone$ or $H_\dicke$ with up to constant error $(\eta,\xi)$ in the framework of full-spectrum simulation of Hamiltonian~\cite{BravyiHastingsSim, UniversalHamiltonian}, as per Definition~\ref{defn:CMPsimul} (assuming the encoding is localized).
This is true since $H_\oneone$ and $H_\dicke$ are real-valued and spectrally gapped, and thus by Lemma \ref{lem:relate-CMP} such a simulation would imply a gap-simulation by our Definition~\ref{defn:hamsimul}, which
was already found to be impossible by Theorem~\ref{thm:main} with localized encoding.
Hence, it's a simple deduction that:

\begin{coro}[Impossibility of full-spectrum degree-reduction, corollary to Theorem~\ref{thm:main}]
\label{coro:imposs-CMP}
It is impossible for an $O(1)$ degree $\tilde{H}_\oneone$ (or $\tilde{H}_\dicke$) with bounded interaction strength to full-spectrum-simulate $H_\oneone$ (or $H_\dicke$) to precision $(\eta,\xi)$ per Def.~\ref{defn:CMPsimul}, where the isometry in the encoding $V$ is an localized encoding per Def.~\ref{defn:localized-encoding}, for sufficiently small constants $\eta$ and $\xi$ and large enough system size.
\end{coro}

\section{Generalized Hastings-Koma Theorem for Decay of Correlation with Non-Zero Energy Spread (Lemma~\ref{lem:HastingsKoma})\label{sec:MHK}}

In this Appendix, we prove a stronger version of Hastings-Koma's result\cite{HastingsKoma}, which we need in order to prove
our main result of Theorem~\ref{thm:main}.
To this end, we relax the assumption in Hastings and Koma's theorem so that the energy spread of the groundspace is non-zero even in the limit of system size $n\to\infty$; this enables us to handle cases where $w_n>0$ for all $n$,
namely small but non-vanishing perturbation to the groundspace.
We show that this still allows us to obtain decay of correlation.
The proof requires a modification of Hastings and Koma's proof, where we calculate better bounds and choose optimal integration parameters, allowing us to mitigate the errors due to the non-zero energy spread.

{
\renewcommand{\thelemma}{\ref{lem:HastingsKoma}}
\begin{lemma}[Hastings-Koma theorem for non-zero energy spread, generalized from Ref.~\cite{HastingsKoma}]
Suppose we have a $n$-qudit Hamiltonian defined on a graph $G=(V,E)$ with exponential decaying interactions (Def.~\ref{defn:expdecayint}).
Also suppose for some constants $0\le w_\infty < 1$ and $\gamma_\infty >0$ independent of system size $n$, the Hamiltonian is quasi-spectrally gapped (Def.~\ref{defn:gap}) with energy spread $w_n \le w_\infty$ and quasi-spectral gap $\gamma_n \ge \gamma_\infty$.
Let $P_0$ be the projector onto the corresponding quasi-groundspace.
Let $A_X$ and $B_Y$ be observables with bounded norm $\|A_X\|,\|B_Y\|=O(1)$ and compact support $X,Y\subset V$, where $[A_X,B_Y]=0$ and $X\cap Y=\emptyset$. Then there exists some constants $C,\tilde{\mu}>0$, independent of $n$, such that for any normalized quasi-groundstate $\ket{\psi}\in P_0$, we have
\begin{equation}
Ce^{-\tilde{\mu}\dist(X,Y)} \ge \left| \braket{\psi|A_X B_Y|\psi} - \frac12 (\braket{\psi|A_XP_0B_Y|\psi} + \braket{\psi|B_YP_0A_X|\psi}) + \O(w_\infty\log^{\frac12}\frac{1}{w_\infty})\right|
\label{eq:MHKineq}
\end{equation}
In the case when $w_\infty=0$, we can ignore the $\O( w_\infty \log^{\frac12} (1/w_\infty))$ term.
\end{lemma}
\addtocounter{lemma}{-1}
}

\begin{proof}
Let us call the given $n$-qudit Hamiltonian $H_{(n)}$.
To avoid dealing with variations in the quasi-spectral gap $\gamma_n$, we rescale each Hamiltonian to  $H_{(n)}' = [H_{(n)}-\lambda_1(H_{(n)})]/\gamma_n$, where $H_{(n)}'$ has the same quasi-groundspace projector $P_0$ with energy spread $w_n'=w_n$ and $\gamma_n' = 1$.
Since $\gamma_n \ge \gamma_\infty > 0$, $H_{(n)}'$ also has exponential decaying interactions since the interaction strength increase by a factor of at most $1/\gamma_\infty$, a constant.
Note that the new Hamiltonian $H_{(n)}'$ has identical eigenstates as $H_{(n)}$, whose eigenvalues are simply rescaled.

For $H_{(n)}'$, let us denote
\begin{itemize}
\item $\{\ket{\phi_{0,\nu}}\}$  as its quasi-groundstates with energy $E_{0,\nu} \le w_n'\gamma_n'$, corresponding to Hermitian projector $P_0$;
\item $\{\ket{\phi_j}\}_{j>0}$  as the rest of its eigenstates with energy $\gamma_n'$ or higher.
\end{itemize}

Due to the condition that $H_{(n)}'$ be quasi-spectrally gapped with energy spread $w_n'=w_n$ and quasi-spectral gap $\gamma_n'=1$, we have that the eigenvalues $\{E_{0,\nu}\} \cup \{E_j\}_{j>0}$ of $H_{(n)}'$ must satisfy:
\begin{equation} \label{eq:QDcond}
\max_{\nu,\nu'} |E_{0,\nu}-E_{0,\nu'}| \le w_n' \gamma_n' = w_n \le w_\infty,
\end{equation}
\begin{equation} \label{eq:gapcond}
\min_{E_j\not\in\{E_{0,\nu}\}} E_j-E_{0,\nu} \ge (1-w_n')\gamma_n' = (1-w_n) \ge (1-w_\infty) \equiv \Delta.
\end{equation}

Now consider any groundstate $\ket{\psi}\in P_0$, which can be written as
\begin{equation}
\ket{\psi} = \sum_\nu c_\nu \ket{\phi_{0,\nu}}.
\end{equation}
Let us denote $\braket{\cdots}=\braket{\psi|\cdots|\psi}$.
To obtain the inequality Eq.~\eqref{eq:MHKineq}, we begin with the equation
\begin{eqnarray}
\braket{[A_X(t),B_Y]} &=& \braket{A_X(t)(\Id-P_0)B_Y} - \braket{B_Y(\Id-P_0)A_X(t)} \nonumber \\
&&+ \braket{A_X(t)P_0B_Y} - \braket{B_Y P_0 A_X(t)}, \label{eq:commexpand}
\end{eqnarray}
where $A_X(t)=e^{iH_{(n)}'t}A_X e^{-iH_{(n)}'t}$ is the Heisenberg picture of observable $A_X$.
The proof of the inequality proceeds by (1) applying a filter function to the RHS of Eq.~\eqref{eq:commexpand} to get rid of time-dependence, and show that it is equivalent to the RHS of Eq.~\eqref{eq:MHKineq},
and then (2) bound the quantity on the LHS through Lieb-Robinson bound.

Note each term on the RHS of Eq.~\eqref{eq:commexpand} can be expanded as
\begin{eqnarray}
\braket{A_X(t)(\Id-P_0)B_Y} &=& \sum_{\nu,\nu'} \sum_{j\neq 0} c_\nu^* c_{\nu'} \braket{\phi_{0,\nu}|A_X|\phi_j}\braket{\phi_j|B_Y|\phi_{0,\nu'}} e^{-it(E_j-E_{0,\nu})}, \\
\braket{B_Y(\Id-P_0)A_X(t)} &=& \sum_{\nu,\nu'} \sum_{j\neq 0} c_\nu^* c_{\nu'} \braket{\phi_{0,\nu}|B_Y|\phi_j}\braket{\phi_j|A_X|\phi_{0,\nu'}} e^{it(E_j-E_{0,\nu'})}, \\
\braket{A_X(t)P_0B_Y} &=& \sum_{\nu,\nu'} \sum_{\mu} c_\nu^* c_{\nu'} \braket{\phi_{0,\nu}|A_X|\phi_{0,\mu}}\braket{\phi_{0,\mu}|B_Y|\phi_{0,\nu'}} e^{-it(E_{0,\mu}-E_{0,\nu})}, \\
\braket{B_YP_0A_X(t)} &=& \sum_{\nu,\nu'} \sum_{\mu} c_\nu^* c_{\nu'} \braket{\phi_{0,\nu}|B_Y|\phi_{0,\mu}}\braket{\phi_{0,\mu}|A_X|\phi_{0,\nu'}} e^{it(E_{0,\mu}-E_{0,\nu'})} .
\end{eqnarray}
Using the gap condition Eq.~\eqref{eq:gapcond}, we can extract the positive frequency part of the first two terms.
And due to the quasi-degeneracy condition Eq.~\eqref{eq:QDcond}, we should be able to eliminate the time-dependence of the last two terms. We rely on the following Lemma, which is an extension of Lemma 3.1 in Hastings-Koma\cite{HastingsKoma}:

\begin{lemma}[Filtering]
\label{lem:filter}
For any $\alpha > 0$, consider the linear operator $\I_\alpha$ on space of function $\{f(t):\mathds{R}\to\C\}$:
\begin{equation}
\I_\alpha[f(t)] \equiv \lim_{T\to\infty} \lim_{\epsilon\to0^+} \frac{i}{2\pi} \int_{-T}^T \frac{f(t)e^{-\alpha t^2}}{t+i\epsilon} dt.
\end{equation}
Then
\begin{eqnarray}
\I_\alpha[e^{-iEt}] &=& \frac{1}{2\pi}\sqrt{\frac{\pi}{\alpha}} \int_{-\infty}^0 d\omega e^{-\frac{(\omega+E)^2}{4\alpha}} = \frac12 \left(1+\erf(\frac{E}{\sqrt{4\alpha}}) \right) \nonumber \\
&=&
\begin{dcases}
1 - \O((\sqrt{\alpha}/\Delta)e^{-\Delta^2/(4\alpha)}) &  \text{ if } E\ge \Delta \\
\O((\sqrt{\alpha}/\Delta)e^{-\Delta^2/(4\alpha)}) &  \text{ if } E \le -\Delta \\
\frac{1}{2} + \O(w_\infty/\sqrt{4\alpha}) & \text{ if } |E|\le w_\infty \\
\end{dcases}.
\label{eq:filterfunction}
\end{eqnarray}
\end{lemma}
{
\renewcommand{\qedsymbol}{$\blacklozenge$}
\begin{proof}[Proof of Lemma~\ref{lem:filter}]
The first  equality of $\I_\alpha[e^{-iEt}] =\int d\omega(\cdots)$ is found in the proof of Lemma 3.1 in Ref.~\cite{HastingsKoma}.
Then the rest follows from evaluating the integral.
\end{proof}
}

It is then clear that by applying the linear operator $\I_\alpha$ (for some $\alpha$ that we will choose later) on the first term of the RHS of Eq.~\eqref{eq:commexpand}, we obtain
\begin{eqnarray}
\I_\alpha[\braket{A_X(t)(\Id-P_0)B_Y}] &=& \sum_{\nu,\nu'} \sum_{j> 0} c_\nu^* c_{\nu'} \braket{\phi_{0,\nu}|A_X|\phi_j}\braket{\phi_j|B_Y|\phi_{0,\nu'}} + \O(e^{-\Delta^2/(4\alpha)}) \nonumber \\
&=& \braket{A_X (\Id-P_0) B_Y} + \O((\sqrt{\alpha}/\Delta) e^{-\Delta^2/(4\alpha)}).
\end{eqnarray}
And similarly on the second term, we obtain
\begin{equation}
\I_\alpha[\braket{B_Y(\Id-P_0)A_X(t)}] = \O((\sqrt{\alpha}/\Delta)e^{-\Delta^2/(4\alpha)}).
\end{equation}
Lastly, on the third and fourth term, we obtain
\begin{eqnarray}
\I_\alpha[\braket{A_X(t)P_0B_Y}] &=& \frac{1}{2} \braket{A_XP_0B_Y} + \O(w_\infty/\sqrt{4\alpha}), \\
\I_\alpha[\braket{B_YP_0A_X(t)}] &=& \frac{1}{2} \braket{B_YP_0A_X} + \O(w_\infty/\sqrt{4\alpha}).
\end{eqnarray}

Now, we make use of the Lieb-Robinson bound for Hamiltonians with exponentially decaying interaction
(see Theorem A.2 in Ref.~\cite{HastingsKoma})
which states that:
\begin{equation}
\left\|[A_X(t),B_Y]\right\| \le K\|A_X\|\|B_Y\||X||Y| (e^{v|t|}-1) e^{-\mu D}, \quad \text{where} \quad D \equiv \dist(X,Y),
\end{equation}
and $K$, $v$ are positive constants that depend only on the structure of $H$ and the graph.
This allows us to bound the LHS of Eq.~\eqref{eq:commexpand}, by applying the linear operator $\I_\alpha$.
Since $(e^{v|t|}-1)/|t| \le v e^{v|t|}$, we have
\begin{eqnarray}
\left| \int_{-\infty}^{\infty} \frac{e^{v|t|}-1}{|t|} e^{-\alpha t^2} \right| \le
\int_{-\infty}^\infty  v e^{v|t|-\alpha t^2}
=
\frac{v\sqrt{\pi}}{\sqrt{\alpha}} e^{v^2/4\alpha}(1+\erf(\frac{v}{\sqrt{4\alpha}}))
\le 2\sqrt{\pi} \frac{v}{\sqrt{\alpha}} e^{v^2/4\alpha}
\end{eqnarray}
where we used the fact that $|\erf(x)|\le 1$. Thus, for some $C_1(\alpha) = K'\frac{v}{\sqrt{\alpha}} e^{v^2/4\alpha}$, where $K'$ is a constant, we have
\begin{equation}
\left| \int_{-\infty}^\infty \frac{\braket{[A_X(t),B_Y]} e^{-\alpha t^2}}{t+i\epsilon} dt\right|
\le C_1(\alpha) e^{-\mu D}
\end{equation}
Let us choose $\alpha = \Delta/ (4\tau)$.
After applying $\I_\alpha$ to both sides of Eq.~\eqref{eq:commexpand}, we obtain
\begin{eqnarray}
C e^{- \mu D}
&\ge& \left|\braket{A_X (\Id-P_0) B_Y} + \frac12\braket{A_XP_0B_Y} - \frac12\braket{B_YP_0A_X} + \O\left(\frac{w_\infty}{\sqrt{\alpha}}\right) + \O\left(\frac{\sqrt{\alpha}}{\Delta}e^{-\Delta^2/(4\alpha)}\right)\right| \nonumber
\\ &=&\left|\braket{A_X B_Y} - \frac12 (\braket{A_XP_0B_Y} + \braket{B_YP_0A_X} ) + \O\left(\frac{w_\infty\sqrt{\tau}}{\sqrt{\Delta}}\right) + \O\left(\frac{1}{\sqrt{\Delta\tau}}e^{-\Delta \tau}\right) \right|.
\label{eq:MHK-non-zero-spread}
\end{eqnarray}
Note $w_\infty/\sqrt{\Delta} = w_\infty/\sqrt{1-w_\infty}=\O(w_\infty)$.

In the case where $w_\infty > 0$, let us choose constants $\tau=(1/\Delta) \log(1/w_\infty)$, $\tilde{\mu}=\mu$, and $C=C_1(\alpha=\Delta/(4\tau))$, we have
\begin{equation}
C e^{-\tilde{\mu} D} \ge \left|\braket{A_X B_Y} - \frac12(\braket{A_XP_0B_Y} + \braket{B_YP_0A_X}) + \O(w_\infty\log^{\frac12}\frac{1}{w_\infty}) \right|,
\end{equation}
which is the desired inequality.

In the case where $w_\infty=0$, let $\xi = \mu \Delta^2/(\Delta^2+v^2)$ be a constant, and we choose $\tau = (\xi/\Delta) D$.
Then
\begin{eqnarray}
&& \left|\braket{A_X B_Y} - \frac12(\braket{A_XP_0B_Y} + \braket{B_YP_0A_X})\right| \le C_1(\alpha) e^{-\mu D} + \O \left(\frac{1}{\sqrt{\Delta \tau}} e^{-\Delta \tau}\right) \nonumber \\
&= & K'\frac{2v\sqrt{\tau}}{\sqrt{\Delta}} e^{-\mu D -v^2\tau/\Delta } + \O\left( \frac{1}{\sqrt{\xi D}}e^{-\xi D}\right) = K'\frac{2v\sqrt{D}}{\Delta} e^{-\xi D} +  \O\left( \frac{1}{\sqrt{\xi D}}e^{-\xi D}\right)
\end{eqnarray}
Now observe that for any $\kappa>0$, if we set $\tilde{\mu}=\xi/(1+\kappa)$, then $e^{(\xi-\tilde{\mu})D} = e^{\kappa\xi D/(1+\kappa)} \ge \kappa\xi \sqrt{D}/(1+\kappa)$ for any $D\ge 0$.
Hence, there exists some positive constant $C$ and $\tilde{\mu}=\xi/(1+\kappa)=\mu \Delta^2/[(\Delta^2+v^2)(1+\kappa)]$ such that
\begin{equation}
C e^{-\tilde{\mu} D} \ge \left|\braket{A_X B_Y} - \frac12(\braket{A_XP_0B_Y} + \braket{B_YP_0A_X})\right|
\end{equation}
which is the desired inequality in this special case of $w_\infty=0$.
\end{proof}

\paragraph{Remark}---
We note that our proof diverges from Hastings and Koma's \cite{HastingsKoma} most nontrivially at Eq.~\eqref{eq:MHK-non-zero-spread},
which is the reason we have an extra term $\O(w_\infty \log(1/w_\infty))$.
Note that unlike our approach, Hastings and Koma worked under the assumption that $w_n\to 0$ as $n\to\infty$;
hence, in the $n\to\infty$ limit,
they can neglect the error term
$\O(w_\infty \sqrt{\tau/\Delta})$ in Eq.~\eqref{eq:MHK-non-zero-spread} since they can choose $w_\infty$ arbitrarily close to 0,
and make the second error term $\O((1/\sqrt{\Delta \tau})e^{-\Delta \tau})$ small by choosing $\tau \sim \dist(X,Y)$.
With this assumption, Hastings and Koma arrived at an alternative bound on the decay or correlation with $\tilde{\mu}=\mu/(1+2v/\Delta)$.
In contrast to their approach, we are interested also in situations where $w_n$ is bounded from below by a non-zero constant, even as $n\to\infty$.
We thus have to optimize the parameter $\tau$ so that the terms $\O(w_\infty\sqrt{\tau})$ and $\O((1/\sqrt{\tau})e^{-\Delta \tau})$ are balanced and do not grow unboundedly.

\section{Impossibility of Dilution Algorithm for Classical Hamiltonians (Theorem~\ref{thm:imposs-dilute}) \label{sec:imposs-dilute}}

In this Appendix, we prove Theorem~\ref{thm:imposs-dilute}, namely the impossibility to dilute classical Hamiltonians with an efficient classical algorithm.

{
\renewcommand{\thethm}{\ref{thm:imposs-dilute}}
\begin{thm}[Impossibility of dilution algorithm for classical Hamiltonians]
If $\coNP \not\subseteq \NPpoly$, then
for any $\xi>0$, $\delta < 1/\sqrt{2}$, $\tilde{w} \le 1/2$, there is no classical algorithm that
given a $k$-local $n$-qubit classical Hamiltonian $H$, runs in $O(\poly(n))$ time to
find an $[r,O(n^{k-\xi}),J]$-diluter of $H$ with $\delta$-unfaithfulness, energy spread $\tilde{w}$, and any encoding $V$ that has an $O(n^{k-\xi})$-bit description. This holds for any $r$ and $J$.
\end{thm}
\addtocounter{thm}{-1}
}

For this, we use a result from Ref.~\cite{DellvanMelkebeek}:

\begin{defn}[Vertex Cover]
Consider any $n$-vertex $k$-uniform hypergraph $\mathcal{G}=(\V,\mathcal{E})$, where $\V$ is the set of vertices and $\mathcal{E}\subseteq \V^k$ is the set of hyperedges.
A vertex cover on $\G$ is a subset of vertices $S\subseteq \V$ such that $\forall e\in \mathcal{E}$, $S\cap e \neq \emptyset$.
The language $k$-$\VC$ is a set consisting of tuples $(\G,m)$, where $\G$ is a $k$-uniform hypergraph with a vertex cover of $\le m$ vertices.
To \emph{decide} $k$-$\VC$ means to output ``yes'' if a given input tuple $(\G,m)\in k$-$\VC$, and ``no'' if $(\G,m)\not \in k$-$\VC$.
\end{defn}

\begin{defn}[oracle communication protocol, adapted from\cite{DellvanMelkebeek}]
An \emph{oracle communication protocol} for a decision problem is a communication protocol between two players.
Player 1 is given the input $x$ and has to run in time polynomial in the length of the input.
Player 2 (the oracle) is computationally
unbounded but is not given any part of $x$.
At the end of the protocol the first player should be able
to decide whether $x$ is accepted (i.e. the answer to the decision problem is yes).
The cost of the protocol is the number of bits of communication from the
Player 1 to Player 2.
\end{defn}

\begin{lemma}[No compressed communication protocol to decide vertex cover, Theorem 2 of \cite{DellvanMelkebeek}]
\label{lem:nocompression}
If $\coNP \not\subseteq \NPpoly$, for any $k\ge 2$ and $\xi>0$, there is no oracle communication protocol of cost $O(n^{k-\xi})$ to decide $k$-$\VC$.
This is true even when Player 1 is co-non-deterministic.
\end{lemma}

We also use Lemma~\ref{lem:PPgroundspace} that bounds the error of groundspace projectors due to perturbations.
This Lemma is first stated in Sec.~\ref{sec:comp-defns} and proved in Sec.~\ref{sec:PPgroundspace-proof}.
We restate it here for convenience.

{
\renewcommand{\thelemma}{\ref{lem:PPgroundspace}}
\begin{lemma}[Error bound on perturbed groundspace (restatement)]
Let $\tilde{H}$ and $\tilde{H}'$ be two Hamiltonians.
Per Def.~\ref{defn:gap}, let $\tilde{P}$ project onto a quasi-groundspace of $\tilde{H}$ with energy spread $\tilde{w}$ and quasi-spectral gap $\gamma$.
Assume $\tilde{w}\le 1/2$ and $\|\tilde{H}' - \tilde{H}\| \le \kappa$, where $\kappa \le (1-\tilde{w})\gamma/8$.
Then there is a quasi-groundspace projector $\tilde{P'}$ of $\tilde{H}'$ with quasi-spectral gap at least $\gamma'$, comprised of eigenstates of $\tilde{H}'$ up to energy at most $\lambda_1(\tilde{H}') + \tilde{w}'\gamma'$, where
\begin{equation}
\gamma' > \gamma-2\kappa, \quad
\tilde{w}'\gamma' \le \tilde{w}\gamma + 2\kappa,
\quad \text{and} \quad
\|\tilde{P}'-\tilde{P}\| < \frac{32\kappa}{\gamma}.
\end{equation}
\end{lemma}
\addtocounter{lemma}{-1}
}

\begin{proof}[\textbf{Proof of Theorem~\ref{thm:imposs-dilute}}]
Suppose for the sake of contradiction that there is a polynomial-time classical algorithm that
given a $k$-local $n$-qubit classical Hamiltonian $H$, runs in $O(\poly(n))$ time to
find an $[r,O(n^{k-\xi}),J]$-diluter of $H$ with unfaithfulness $\delta < 1/\sqrt{2}$, energy spread $\tilde{w}\le 1/2$, and some encoding $V$ described by $O(n^{k-\xi})$ classical bits.
We can then construct an oracle communication protocol to decide $k$-$\VC$:
\begin{enumerate}
\item Player 1 takes the input $(\G,m)$ for $k$-$\VC$, where $\mathcal{G}=(\V,\E)$  is a $k$-uniform  $n$-vertex hypergraph, and encodes $\G$ as a  $k$-local  $n$-qubit classical Hamiltonian $H$.
Specifically, we can encode the problem as the following Hamiltonian
\begin{gather}
H = 2H_\text{con} + H_\text{count}, \nonumber \\
\text{where} \quad H_\text{con} =  \sum_{(i_1,\ldots,i_k) \in \mathcal{E}} \ketbra{0}^{(i_1)}\otimes \ketbra{0}^{(i_2)}\otimes \cdots \otimes \ketbra{0}^{(i_k)}
\quad \text{and}
\quad
H_\text{count} = \sum_{i=1}^n \ketbra{1}_i.
\end{gather}
The eigenstates of $H$ are classical strings $\ket{z_1z_2\cdots z_n}$,
where $z_i=1$ means that vertex $i$ is chosen to be in the vertex cover.
Note $H_\text{con}$ ensures that every hyperedge in $\mathcal{E}$ is covered by some vertex, and $H_\text{count}$ penalizes any extra vertices used to cover the hypergraph.
Note that for any computational basis state with energy $E$ not representing a vertex cover (i.e. for some hyperedge $e\in \mathcal{E}$, $z_i=0$ for all $i\in e$), there is a state with energy $E-1$ (by changing one of the $z_i=1$ for some $i\in e$ to cover the hyperedge $e$).
Hence, the groundstates of $H$ represent minimum vertex covers on $\mathcal{G}$.
Let $P=\sum_\mu \ketbra{g_\mu}$ be the projector onto groundstates $\ket{g_\mu}$ of $H$, where each $\ket{g_\mu}$ is written in the computational basis and represents a minimum vertex cover on $\mathcal{G}$.
Observe that $H$ is spectrally gapped with energy spread $w=0$ and spectral gap $\gamma=1$, since any state not representing a minimum vertex cover would incur an energy penalty of at least 1.
\item Player 1 then uses the supposed polynomial-time classical algorithm to generate a diluter $\tilde{H}$ of $H$ with encoding $V$, unfaithfulness  $\delta$ and energy spread $\tilde{w}$, such that $\tilde{H}=\sum_{i=1}^M \tilde{H}_i$, where $\tilde{H}_i$ are $O(1)$-local and $M=O(n^{k-\xi})$.
Furthermore, Player 1 takes each term $\tilde{H}_i$ and rewrites it as $\tilde{H}_i'$ in $s$-bit precision, producing $\tilde{H}'=\sum_i \tilde{H}_i'$.
Noting that $\|\tilde{H}'-\tilde{H}\| \le O(M/2^s)$, we can simply choose $s=O(\log_2 (M/\kappa)) = O(\log_2 (n/\kappa))$ to ensure $\|\tilde{H}'-\tilde{H}\|\le \kappa$, for some
\begin{equation}
\kappa < \min\{(1-\tilde{w})/8, (1/\sqrt{2} - \delta)/64\}.
\end{equation}
Player 1 then communicates $(\tilde{H}', V, n, m)$ to Player 2, incurring a protocol cost of $O(Ms) = O(n^{k-\xi}\log n)$.

\item Player 2 uses their unbounded computational resource to diagonalize $\tilde{H}'$ and find all its eigenstates.
Let us first bound the distance between the groundspace of $\tilde{H}'$ and that of $\tilde{H}$.
To that end, let us denote $\tilde{P}$ as the quasi-groundspace projector of $\tilde{H}$ with quasi-spectral gap $\gamma=1$ and energy spread $\tilde{w}$.
According to Lemma~\ref{lem:PPgroundspace}, since $\tilde{w}\le 1/2$ and $\|\tilde{H}'-\tilde{H}\|\le \kappa \le (1-\tilde{w})/8$,
there is a quasi-groundspace projector $\tilde{P}'$ of $\tilde{H}'$ comprised of eigenstates of $\tilde{H}'$ with energy at most $\lambda_1(\tilde{H}')+ \tilde{w}+2 \kappa$, satisfying $\|\tilde{P}'-\tilde{P}\| < 32\kappa$.
Furthermore, since $\tilde{H}$ is a diluter of $H$ with unfaithfulness $\delta$, we have $\|\tilde{P}-V P V^\dag \tilde{P}\|\le \delta$.
Denoting $P'=V P V^\dag$, this implies $\|\tilde{P}'-P'\tilde{P}'\| \le \|\tilde{P}'-\tilde{P} + \tilde{P} - P'\tilde{P}+  P'\tilde{P} - P'\tilde{P}'\| \le \delta + 64\kappa$.
Note since Player 1 had chosen $\kappa < (1/\sqrt{2}-\delta)/64$, we have $\delta+64\kappa < 1/\sqrt{2}$, which implies $\|\tilde{P}'-VPV^\dag \tilde{P}'\| = \|\tilde{P}'- P' \tilde{P}'\|^2 < 1/2$.

We now show how Player 2 can decide whether $(\G,m) \in k$-$\VC$ by using any groundstate $\ket{\tilde{g}}$ of $\tilde{H}'$.
Let us denote $Q_{\le m}^n$ as the projector onto all computational basis states on $n$ qubits with $\le m$ qubits in the state $\ket{1}$.
Then we claim
\begin{equation}
(\G,m) \in k\text{-}\VC
~ \Longleftrightarrow ~
\braket{\tilde{g}|V(Q_{\le m}^n \otimes \Id_\anc)V^\dag |\tilde{g}} \ge \frac12
\text{ for any groundstate } \ket{\tilde{g}} \text{ of } \tilde{H}'.
\label{eq:Player2-decision}
\end{equation}
This is because:
\begin{itemize}
\item[(a)] If $(\G,m)\in k$-$\VC$, the original groundspace projector $P$ of $H$ contains only states with $\le m$ 1's on the original $n$ qubits.
Then $Q_{\le m}^n - P$ is positive semi-definite, and
$\braket{\psi|V Q_{\le m}^n V^\dag |\psi} \ge \braket{\psi|V P V^\dag|\psi}$ for any state $\ket{\psi}$. Note for any groundstate $\ket{\tilde{g}}$ of $\tilde{H}'$, we have $\tilde{P}'\ket{\tilde{g}}=\ket{\tilde{g}}$, which means
\begin{gather}
\frac12 > \|(\tilde{P}'-VPV^\dag \tilde{P}')\ket{\tilde{g}}\|^2 = 1 - \braket{\tilde{g}|VPV^\dag |\tilde{g}}
\nonumber \\
 \Longrightarrow \quad
\braket{\tilde{g}|V Q_{\le m}^n V^\dag|\tilde{g}} \ge \braket{\tilde{g}|V P V^\dag|\tilde{g}} > \frac12.
\end{gather}
\item[(b)] If $(\G,m)\not \in k$-$\VC$, the original groundspace projector $P$ contains only states with $> m$ 1's on the original $n$ qubits.
Then $\Id - Q_{\le m}^n - P$ is positive semi-definite, and $\braket{\psi|Q_{\le m}^n|\psi} \le 1-\braket{\psi|P|\psi}$ for any state $\ket{\psi}$.
For any groundstate $\ket{\tilde{g}}$ of $\tilde{H}'$, we have similarly
\begin{gather}
\frac12 > \|(\tilde{P}'-V P V^\dag\tilde{P}')\ket{\tilde{g}}\|^2 = 1-\braket{\tilde{g}|V P V^\dag|\tilde{g}} \ge \braket{\tilde{g}|V Q_{\le m}^n V^\dag|\tilde{g}} \nonumber \\
\Longrightarrow \quad
\braket{\tilde{g}|V Q_{\le m}^n V^\dag|\tilde{g}} \not\ge \frac12.
\end{gather}
\end{itemize}
In other words, Player 2 may take any groundstate $\ket{\tilde{g}}$ of $\tilde{H}'$, compute its expectation value of $V Q_{\le m}^n V^\dag$, and transmit the decision yes to Player 1 if and only if $\braket{\tilde{g}|V Q_{\le m}^n V^\dag|\tilde{g}} \ge 1/2$.
\item Player 1 receives the decision from Player 2, which decides whether $(\G,m) \in k$-$\VC$.
\end{enumerate}
Since this oracle communication protocol can decide the vertex cover problem with $O(n^{k-\xi}\log n) = O(n^{k-\xi/2})$ cost for any $\xi > 0$, it directly contradicts Lemma~\ref{lem:nocompression}.
Hence, for any $\delta < 1/\sqrt{2}$ and $\tilde{w}\le 1/2$, no polynomial-time classical algorithm exists to find diluters of classical Hamiltonians with $\delta$-unfaithfulness,  energy spread $\tilde{w}$, and encoding $V$ that is described by $O(n^{\xi-k})$ classical bits.
\end{proof}

\section{Incoherent Degree-Reduction and Dilution}

\subsection{Degree-Reduction of Classical Hamiltonians (Proposition~\ref{prop:classical-deg-reduct})\label{sec:classical-degree-reduction}}

In the classical world, degree-reduction of Constraint Satisfaction Problems is famously used for proving the important result of PCP theorem~\cite{dinur}.
Here we show that the same construction can be used to degree-reduce $k$-local ``classical" Hamiltonians, where all terms are diagonal in the computational basis, albeit \emph{incoherently}.

{
\renewcommand{\theprop}{\ref{prop:classical-deg-reduct}}
\begin{prop}[Incoherent degree-reduction of classical Hamiltonian]
Consider an $n$-qudit $k$-local \emph{classical} Hamiltonian $H = \sum_{S\subset \{1,\ldots, n\}} C_S$, where each $C_S:\{z_i:i\in S\} \to [0, 1]$ is a function of $d$-ary strings of length $|S|\le k$ representing states of qudits in $S$.
Let the number of terms in $H$ be $M_0=|\{S\}|=O(n^k)$.
Then there is a $k$-local $[3,O(kM_0),O(1)]$-degree-reducer of $H$
with $0$-unfaithfulness, no energy spread, and trivial encoding $V=\Id$.
\end{prop}
\addtocounter{prop}{-1}
}

\begin{proof}
Since the original Hamiltonian $H$ is diagonalizable in the computational $z$-basis, from now on we will slightly abuse notation to use $z_i$ to denote both the operator $z_i=\sum_{z=0}^{d-1}\ketbra{z}_i$ and (the state of) the $i$-th qudit interchangeably.
To construct a degree-reducer of $H$, we first replace each original qudit $z_i$ with a cluster of $r_i=|\{S\ni i: S\}|$ qudits, each of which we denote $\tilde{z}_{i,S\ni i}$, corresponding to the term $C_S$ that qudit $i$ participates in the original Hamiltonian.
We then add a ring of equality constraints to make sure that all  $r_i$ qudits in each cluster $i$ have the same value in the computational basis.
Subsequently, by substituting each $k$-local term in the original Hamiltonian with a $k$-local term that acts on one qudit per the cluster corresponding to an original qudit, we are able to produce an equivalent Hamiltonian with degree of at most 3.
More concretely,
if we denote $\gamma$ as the spectral gap of the original Hamiltonian $H$, this construction yields the following sparsifier Hamiltonian:
\begin{gather}
\tilde{H} = H_\text{con} + H_\text{eq},\nonumber \\
\text{where} \quad H_\text{con} = \sum_S C_S(\{ \tilde{z}_{i,S} : i\in S\}) \quad \text{and} \quad  H_\text{eq} = \frac{\gamma}{4} \sum_{i=1}^n \sum_{j=1}^{r_i-1} (\tilde{z}_{i,j}-\tilde{z}_{i, j+1})^2.
\end{gather}
Note that in $H_\text{eq}$, we denote $j$ as some index labeling the $r_i$ qudits in the $i$-th cluster corresponding to different $S\ni i$.
It can then be seen that every qudit $\tilde{z}_{i,S}$ has at most degree 3, appearing in the term $C_S$ as well as up to two 2-local terms in a ring of equality constraints.
Since for each $k$-local term $C_S$, we introduce $k$ qudits, then the number of terms (as well as the number of qudits) in the sparsifier is $M=O(kM_0)$.
Additionally, since the constraint functions $\|C_S\|\le 1$ are bounded in the original Hamiltonian by assumption, then both $\gamma$ and thus strength of individual terms in sparsifier $\tilde{H}$ is also bounded by $J=O(1)$.

It remains to show that $\tilde{H}$ gap-simulates $H$.
Let us denote $\ket{\vect{z}} = \bigotimes_{i=1}^n \ket{z_i}$, for every $\vect{z}\in \mathds{Z}_d^n$, which are computational basis states of the original Hilbert space $\H$ and are also eigenstates of $H$.
Then let us denote a corresponding state $\ket{\tilde{\vect{z}}} = \bigotimes_{i=1}^n \bigotimes_{S\ni i} \ket{\tilde{z}_{i,S\ni i} = z_i}$ in the extended Hilbert space.
We can then show that for any groundstate  $\ket{\vect{z}}$ of the original Hamiltonian, where $H\ket{\vect{z}}=E_{\vect{z}}\ket{\vect{z}}$ for $E_{\vect{z}} \le w\gamma$, the corresponding state $\ket{\tilde{\vect{z}}}$ is an eigenstate of $\tilde{H}$ with the same energy $\ket{\tilde{\vect{z}}} = E_{\vect{z}}\ket{\tilde{\vect{z}}}$.
Any other state will have energy at least $\gamma$ since it must either correspond to an excited state of $H$ or violate a constraint in $H_\text{eq}$.
Hence, $\tilde{H}$ reproduces the spectrum of $H$ up to $\gamma$ (we can easily increase the range of spectrum it reproduces by increasing the strength of $H_\text{eq}$).
Lastly, by identifying one qudit in $i$-th cluster as the original qudit $i$ (and hence yielding a trivial encoding), we can easily see that $\delta \equiv \|\tilde{P}-P\tilde{P}\|=0$ where $P=\sum_{E_{\vect{z}}\le w\gamma} \ketbra{\vect{z}}$ and $\tilde{P} = \sum_{E_{\vect{z}}\le w\gamma} \ketbra{\tilde{\vect{z}}}$.
Therefore, $\tilde{H}$ is a $k$-local $[3,O(km),O(1)]$-degree-reducer of $H$ with zero unfaithfulness and identical low-energy spectrum.
\end{proof}

\subsection{Incoherent Tree-Graph Diluter and Degree-Reducer of $H_\oneone$ (Proposition~\ref{prop:incoherent-tree})\label{sec:incoherenttree}}

While Proposition~\ref{prop:classical-deg-reduct} shows that it is always possible to \emph{constructively} degree-reduce classical Hamiltonians incoherently,
we know this is not possible for dilution due to Theorem~\ref{thm:imposs-dilute}.
Nevertheless, in some cases, such as our example Hamiltonian $H_\oneone$,
we show that a non-trivial incoherent diluter exists.
The key idea behind our construction of this sparsifier of $H_\oneone$ is to use additional ancilla and constraints to simulate a counting operation, so that only states with less than two excitations on the original set of qubits do not lead to violation of any constraints.
By using ancilla qubits as memory and arranging them in a recursive, tree-like geometry, we are able to limit the maximum degree and the number of interactions required in the sparsified Hamiltonian to perform this counting operation.

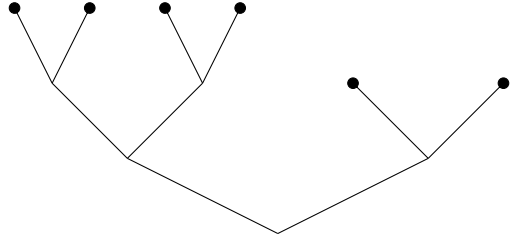
\begin{figure}[htb]
\centering
\begin{tikzpicture}[level distance=10mm,grow'=up]
  \tikzstyle{level 1}=[sibling distance=40mm]
  \tikzstyle{level 2}=[sibling distance=20mm]
  \tikzstyle{level 3}=[sibling distance=10mm]
 \coordinate
  child {
    child {
      child {[fill] circle (2pt)}
      child {[fill] circle (2pt)}
    }
    child {
      child {[fill] circle (2pt)}
      child {[fill] circle (2pt)}
    }
  }
  child {
    child {[fill] circle (2pt)}
    child {[fill] circle (2pt)}
  };
\end{tikzpicture}
\caption{\label{fig:tree}Example incoherent degree-reducer/diluter geometry for $H_\oneone$ on 6 qubits, where $\bullet$ on the leaf nodes denote the original qubits, and the ancilla qubits are located at the internal nodes.}
\end{figure}

{
\renewcommand{\theprop}{\ref{prop:incoherent-tree}}
\begin{prop}[$0$-unfaithfulness incoherent dilution and DR for $H_\oneone$]
There is a 3-local incoherent $[2,n-1,1]$-diluter of $H_\oneone$ with 0-unfaithfulness, energy spread $\tilde{w}=0$, and trivial encoding.
This is also an incoherent $[2,n-1,1]$-degree-reducer of $H_\oneone$.
\end{prop}
\addtocounter{prop}{-1}
}

\begin{proof}
Let us now describe our construction.
For the original Hamiltonian $H_\oneone$ with $n$ qubits, the incoherent sparsifier we propose involves an additional $n-1$ ancilla qubits.
We arrange them on a binary tree of height $\lceil\log_2 n\rceil$, with the original qubits placed on the leaf nodes, as shown in Fig.~\ref{fig:tree}.
The sparsifier we propose consists of a sum of 3-local, commuting, and positive semi-definite terms, one per branching at each internal node:
\begin{equation}
\tilde{H}_\oneone^\text{tree} = \sum_{i=1}^{n-1} \tilde{H}_{\oneone,i}^\text{tree} \quad \text{with} \quad \tilde{H}_{\oneone,i}^\text{tree} = \sum_{l_i,r_i,b_i=0}^1 h(l_i,r_i,b_i)\ketbra{l_ir_ib_i}
\end{equation}
where $\ket{b_i}$ is the qubit state at $i$-th internal node, and $\ket{l_i}$ ($\ket{r_i}$) is the qubit state of its left (right) child node.
The energy cost function $h(l,r,b)$ is
\begin{equation}
h\Bigg(
\begin{tikzpicture}[baseline=1.6ex, level distance = 10mm, sibling distance = 8mm]
\node {$b$}[grow'=up] {child {node {$l$}} child {node {$r$}}};
\end{tikzpicture}
\Bigg) =
\begin{cases}
0, & \begin{tikzpicture}[baseline=1.6ex, level distance = 7mm, sibling distance = 4mm]
    \node {$b$}[grow'=up] {child {node {$l$}} child {node {$r$}}};
    \end{tikzpicture}
    \in \Big\{
     \begin{tikzpicture}[baseline=1.6ex, level distance = 7mm, sibling distance = 4mm]
     \node {0}[grow'=up] {child {node {0}} child {node {0}}};
     \end{tikzpicture},
     \begin{tikzpicture}[baseline=1.6ex, level distance = 7mm, sibling distance = 4mm]
     \node {1}[grow'=up] {child {node {0}} child {node {1}}};
     \end{tikzpicture},
     \begin{tikzpicture}[baseline=1.6ex, level distance = 7mm, sibling distance = 4mm]
     \node {1}[grow'=up] {child {node {1}} child {node {0}}};
     \end{tikzpicture}
    \Big\} \\
1, & \text{otherwise}
\end{cases}
\end{equation}
which imposes a constraint that forces the parent node of any branching to be excited whenever any of its children nodes contains an excitation.
This information is passed down towards the root node at the bottom, and it's easy to see that we are effectively counts the number of excitations among the original $n$ qubits.
No constraint is violated if and only if there are no more than one excitations.
In other words, the zero-energy groundspace of $\tilde{H}_\oneone^\text{tree}$ consists of the $n+1$ states corresponding to zero or one excitations in the original $n$ qubits.
Since the $\tilde{H}_\oneone^\text{tree}$ is commuting, it is easy to solve for all the eigenstates.
By inspection, $\tilde{H}_\oneone^\text{tree}$ is spectrally gapped with energy spread 0 and gap 1, where the excited manifold consists of both states with more than one excitations among the original qubits as well as ``illegal'' states that violate the constraints unnecessarily.
The groundspace consists of $n+1$ states $\{\ket{\tilde{g}_i}\}_{i=0}^n$, where $\ket{\tilde{g}_0}$ has zero excitations among the original qubits, and $\ket{\tilde{g}_i}$ has the $i$-th original qubits excited.
Some example groundstate configurations are shown below in Fig.~\ref{fig:treeincoherent}.

Let us now analyze the performance of this sparsifier construction.
Note that this construction contains $2n-1$ qubits, each involved in at most 2 Hamiltonian terms -- a maximum degree of $r=2$.
Furthermore, there are only $M=n-1$ terms in the sparsifier Hamiltonian, compared to the original Hamiltonian of $n(n-1)/2$ terms.
Each term has bounded strength $J=\|\tilde{H}_{\oneone,i}^\text{tree}\|= 1$.
Additionally, $\tilde{H}_{\oneone}^\text{tree}$ yields the same energy spread of $\tilde{w}=w=0$ and gap of $\gamma=1$ as the original.
Furthermore, since the groundspace of $\tilde{H}_\oneone$ faithfully reproduces the original groundspace configuration on the $n$ original qubits, we have $P\tilde{P}=\tilde{P}$ and $\delta=\|\tilde{P}-P\tilde{P}\|=0$.
Hence, $\tilde{H}_\oneone^\text{tree}$ gap-simulates $H_\oneone$ with 0-unfaithfulness.
Therefore, our construction of $\tilde{H}_\oneone^\text{tree}$ is a 3-local $[2,n-1,1]$-gap-simulator of $H_\oneone$, with energy spread $\tilde{w}=0$ and zero unfaithfulness.
\end{proof}

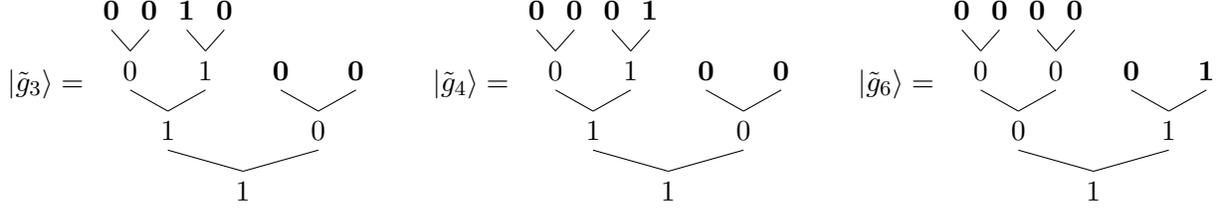
\begin{figure}[h]
\centering $\ket{\tilde{g}_3}=$
\begin{tikzpicture}[baseline=8ex,level distance=8mm,grow'=up]
  \tikzstyle{level 1}=[sibling distance=20mm]
  \tikzstyle{level 2}=[sibling distance=10mm]
  \tikzstyle{level 3}=[sibling distance=5mm]
 \node {1}
  child { node {1}
    child { node {0}
      child {node {\textbf0}}
      child {node {\textbf0}}
    }
    child { node {1}
      child {node {\textbf1}}
      child {node {\textbf0}}
    }
  }
  child { node {0}
    child {node {\textbf0}}
    child {node {\textbf0}}
  };
\end{tikzpicture}
\hspace{15pt} $\ket{\tilde{g}_4}=$
\begin{tikzpicture}[baseline=8ex,level distance=8mm,grow'=up]
  \tikzstyle{level 1}=[sibling distance=20mm]
  \tikzstyle{level 2}=[sibling distance=10mm]
  \tikzstyle{level 3}=[sibling distance=5mm]
 \node {1}
  child { node {1}
    child { node {0}
      child {node {\textbf0}}
      child {node {\textbf0}}
    }
    child { node {1}
      child {node {\textbf0}}
      child {node {\textbf1}}
    }
  }
  child { node {0}
    child {node {\textbf0}}
    child {node {\textbf0}}
  };
\end{tikzpicture}
\hspace{15pt} $\ket{\tilde{g}_6}=$
\begin{tikzpicture}[baseline=8ex,level distance=8mm,grow'=up]
  \tikzstyle{level 1}=[sibling distance=20mm]
  \tikzstyle{level 2}=[sibling distance=10mm]
  \tikzstyle{level 3}=[sibling distance=5mm]
 \node {1}
  child { node {0}
    child { node {0}
      child {node {\textbf0}}
      child {node {\textbf0}}
    }
    child { node {0}
      child {node {\textbf0}}
      child {node {\textbf0}}
    }
  }
  child { node {1}
    child {node {\textbf0}}
    child {node {\textbf1}}
  };
\end{tikzpicture}
\caption{\label{fig:treeincoherent}Three example groundstates of $\tilde{H}_\oneone^\text{tree}$ for 6 original qubits. Note $\ket{\tilde{g}_3}$ and $\ket{\tilde{g}_4}$ have the same ancilla states, but $\ket{\tilde{g}_6}$ has a different ancilla state, so the sparsifier is incoherent.}
\end{figure}

\paragraph{Remark}---
Unfortunately, the incoherence of this sparsifier construction is unavoidably high, since different groundstates on the original $n$ qubits can be strongly correlated with different ancilla states, as seen in Fig.~\ref{fig:treeincoherent}.
To lower bound the incoherence, we note that among the groundstates $\ket{\tilde{g}_i}$, only $\ket{\tilde{g}_{2\nu-1}}$ and $\ket{\tilde{g}_{2\nu}}$ share the same ancilla state, which we denote $\ket{a_\nu}^\anc$, for $\nu=1,\ldots,\lceil n/2 \rceil$.
(Note when $n$ is odd, we simply set $\ket{g_{2\nu}}=0$ when $\nu=\lceil n/2 \rceil$.)
Let us also denote $\ket{a_0}^\anc = \ket{0^{n-1}}^\anc$ as the ancilla configuration in the groundstate $\ket{\tilde{g}_0}$.
Then, we write $P_\nu = \ketbra{g_{2\nu-1}} + \ketbra{g_{2\nu}}$ for $\nu=1,\ldots,\lceil n/2 \rceil$, and $P_0 = \ketbra{g_0}$ as projectors onto disjoint subsets of groundstates of the original Hamiltonian.
From this, we can write $\tilde{P}=\sum_{\nu=0}^{\lceil n/2 \rceil} P_\nu\otimes \ketbra{a_\nu}$.
Since $P_\nu P_\nu' = \delta_{\nu\nu'} P_\nu$, we note the sum over $\sum_\nu P_\nu\otimes M_\nu$ has a block diagonal structure, where each block is $M_\nu$ supported by some state in $P_\nu$.
Hence,
\begin{eqnarray*}
\tilde{P}-P\otimes P_\anc &=& \sum_{\nu=0}^{\lceil n/2 \rceil} P_\nu\otimes(\ketbra{a_\nu} - P_\anc) = \bigoplus_\nu (\ketbra{a_\nu} - P_\anc) \\
\epsilon \ge \| \tilde{P}-P\otimes P_\anc \| &=& \max_\nu \|\ketbra{a_\nu} - P_\anc\|.
\end{eqnarray*}
Let us denote $\ket{a}$ as some state that $P_\anc$ projects onto ($P_\anc\ket{a}=\ket{a}$). Noting that $\braket{a_\nu|a_{\nu'}}=\delta_{\nu\nu'}$, we have
\begin{eqnarray}
\epsilon^2 &\ge& \max_\nu \|\ketbra{a_\nu} - P_\anc\|^2 \ge \min_{\ket{a}} \max_\nu \braket{a|(\ketbra{a_\nu}-1)^2|a} = \min_{\ket{a}} \max_\nu 1-\left|\braket{a|a_\nu}\right|^2  \nonumber \\
&\ge& \min_{\ket{a}} \frac{1}{1+\lceil n/2\rceil}\sum_{\nu=0}^{\lceil n/2 \rceil}  (1-\left|\braket{a|a_\nu}\right|^2 = 1-\frac{\max_{\ket{a}} \sum_{\nu} \left|\braket{a|a_\nu}\right|^2}{1+\lceil n/2 \rceil} \ge 1 - \frac{1}{1+\lceil n/2\rceil}
\end{eqnarray}
where in the beginning of the last line, we used the fact that the maximum is at least as large as the average of a set of numbers.
Hence, we have necessarily large (close-to-1) incoherence $\epsilon \ge \sqrt{1-1/(1+\lceil n/2 \rceil)} \approx 1 - O(1/n)$.

\section{Coherent DR with Polynomial Interaction Strength}

In this Appendix, we provide constructions that perform degree-reduction using interaction strength $J=\poly(n,\epsilon^{-1})$ for a large class of Hamiltonian.
In Sec.~\ref{sec:degree-reduction-poly}, we describe how to perform coherent degree-reduction for any Hamiltonians whose spectral gap decays at most inverse-polynomially with system size;
in fact, the resultant degree-reducer Hamiltonian not only gap-simulates, but also full-spectrum-simulates the original Hamiltonian.
Later in Sec.~\ref{sec:prop-circuit}, we show that for our example $H_\oneone$, coherent dilution is also possible with polynomial strength.

\subsection{Constructive Coherent DR with Polynomial Interaction Strength (Theorem~\ref{thm:degree-reduction-poly})\label{sec:degree-reduction-poly}}

In this section, we show how to perform coherent degree-reduction for Hamiltonians whose quasi-spectral gap scales as $\gamma=\Omega(1/\poly(n))$.

{
\renewcommand{\thethm}{\ref{thm:degree-reduction-poly}}
\begin{thm}[Coherent DR with polynomial interaction strength]
Suppose $H$ is an $O(1)$-local Hamiltonian with a quasi-groundspace projector $P$, which has quasi-spectral gap $\gamma=\Omega(1/\poly(n))$ and energy spread $w$.
Also assume $\|H\|=O(\poly(n))$.
Then for every $\epsilon>0$, one can construct an $O(1)$-local $[O(1), O(\poly(n)/\epsilon^2), O(\poly(n,\epsilon^{-1}))]$-degree-reducer of $H$
with incoherence $\epsilon$, energy spread $w+O(1/\poly(n))$, and trivial encoding.
\end{thm}
\addtocounter{thm}{-1}
}

To prove the above Theorem, we first prove two smaller Lemmas~\ref{lem:circuit-ham-simul} and \ref{lem:circuit-idling} about different aspects of using Kitaev's circuit-to-Hamiltonian construction~\cite{KSV02} for Hamiltonian simulation.
The following two concepts will be useful in the discussion:

\begin{defn}[history states]
Let $U=U_T \cdots U_2 U_1$ be a quantum circuit acting on $n+m$ qudits.
Then for any input state $\ket{\psi_\mu}\in \C^{d^n}$, the \emph{history state} with respect to $U$ and $\ket{\psi_\mu}$ is the following
\begin{equation}
\ket{\eta_\mu} = \frac{1}{\sqrt{T+1}}\sum_{t=0}^T \Big( U_{t}\cdots U_2 U_1 \ket{\psi_\mu}\ket{0^m}^\anc \Big)\ket{1^t 0^{T-t}}^\clock
\end{equation}
\end{defn}
\begin{defn}[circuit degree]
The degree of a quantum circuit $U=\prod_{t=1}^T U_t$ is $\deg(U) = \max_i |\{U_t: U_t \textnormal{ acts nontrivially on qudit } i\}|$.
\end{defn}

We now prove the first of the two Lemmas, which describes a circuit-to-Hamiltonian transformation that can be used for full-spectrum-simulation, assuming an appropriate energy penalty Hamiltonian $H_{out}$ can be constructed.

\begin{lemma}[Circuit-Hamiltonian simulation]
\label{lem:circuit-ham-simul}
Consider an orthonormal basis of states $\{\ket{\psi_\mu}\}_{\mu=1}^{d^n}$ on $n$ qudits.
Let $U=\prod_{t=1}^T U_t$ be a quantum circuit where $U_t$ is a $k$-local gate.
Let $\L = \spn\{\ket{\eta_\mu}\}_{\mu=1}^{d^n}$ be the space of history states with respect to some and $\{\ket{\psi_\mu}\}$, and let $H_\eff = \sum_\mu \lambda_\mu \ketbra{\eta_\mu}$.
Suppose there exists a Hamiltonian $H_{out}$ such that
$\|H_\eff - H_{out}|_{\L}\| \le \xi/2$.
Then for any $\eta>0$, we can construct a Hamiltonian $\tilde{H}_\circuit$ from the description of $U$ such
that $\tilde{H}_\circuit$ full-spectrum-simulates $H_\eff$ to precision $(\eta,\xi)$ below some energy cut-off $\Delta\ge O (\xi^{-1}\|H_{out}\|^2 + \eta^{-1}\|H_{out}\|)$, per Def.~\ref{defn:CMPsimul}.
The constructed $\tilde{H}_\circuit$ is $(k+3)$-local, has $O(\deg(U))$ degree, $O(T)$ number of terms, and $O(\poly(n,T,\xi^{-1},\eta^{-1}, \|H_{out}\|))$ interaction strength.
\end{lemma}
\begin{proof}
For a given circuit $U = U_T \cdots U_2 U_1$, where $T=O(\poly(n))$,
the corresponding circuit-Hamiltonian is
\begin{align}
\tilde{H}_{\circuit} &= H_0 +  H_{out}\\
\text{where} \quad H_0 &= J_{clock} H_{clock} + J_{prop} H_{prop} + J_{in} H_{in}
\end{align}
The role of $H_0$ is to isolate $\L=\spn\{\ket{\eta_\mu}\}$ as the zero-energy groundspace separated by a large spectral gap.
The first part of $H_0$ is
\begin{equation}
\label{eq:Hclock}
H_{clock} = \sum_{t=1}^{T-1}\ketbra{01}^\clock_{t,t+1},
\end{equation}
which sets the legal state configurations in the clock register to be of the form $\ket{t}^\clock \equiv \ket{1^t 0^{T-t}}^\clock$.
Then, we simulate the state propagation under the circuit using
\begin{eqnarray}\label{eq:Hprop}
H_{prop} &=& \sum_{t=1}^T H_{prop,t},\\
\text{where} \quad H_{prop,t} &=& \Id \otimes \ketbra{100}^\clock_{t-1,t,t+1} - U_t\otimes\ketbrat{110}{100}^\clock_{t-1,t,t+1} \nonumber\\
&& - U_t^\dag\otimes \ketbrat{100}{110}^\clock_{t-1,t,t+1} + \Id\otimes\ketbra{110}^\clock_{t-1,t,t+1}
\quad \text{for} \quad 1<t<T, \nonumber\\
H_{prop,1} &=& \Id\otimes \ketbra{00}_{12}^\clock - U_1 \otimes\ketbrat{10}{00}_{12}^\clock - U_1^\dag \otimes \ketbrat{00}{10}_{12}^\clock + \Id\otimes\ketbra{10}_{12}^\clock, \nonumber\\
\text{and} \quad H_{prop,T} &=& \Id\otimes \ketbra{10}_{T-1,T}^\clock - U_T \otimes\ketbrat{11}{10}_{T-1,T}^\clock - U_T^\dag\otimes\ketbrat{10}{11}_{T-1,T}^\clock + \Id\otimes\ketbra{11}_{T-1,T}^\clock.\nonumber
\end{eqnarray}
These terms check the propagation of states from time $t-1$ to $t$ is correct.
Now, we also need to ensure that the ancilla qudits are in the state $\ket{0^{m}}^\anc$ when the clock register is $\ket{0^T}^\clock$, using
\begin{gather}
H_{in} = \sum_{i=1}^{m} (\Id - \ketbra{0})_i^\anc\otimes\ketbra{0}^\clock_{t_{\min}(i)},
\\
\text{where}
\quad
t_{\min}(i) = \min \{t: U_t \text{ acts nontrivially on ancilla qudit } i\}.
\nonumber
\end{gather}
In other words, for each ancilla qudit $i$, $H_{in}$ penalizes the ancilla if it's not in the state $\ket{0}$ before it is first used by the $t_{\min}(i)$-th gate.
Note that $\tilde{H}_\circuit$ has $O(T)$ terms, each of which is most $(k+3)$-local when $U_t$ are $k$-local.
It also has $O(\deg(U))$ degree, since each computational qubit is involved in at most $O(\deg(U))$ terms in $H_{prop}$, and each clock qubit is involved in $O(1)$ terms.

It is easy to see that $H_0\L=0$.
But we also need to lower bound the spectral gap of $H_0$, i.e. $\lambda_1(H_0|_{\L^\perp} )$.
To that end, let us denote the following subspaces:
\begin{align}
\S_{clock} &= \spn\{ \ket{\psi}\ket{y}\ket{1^t 0^{T-t}}:\ket{\psi}\in \C^{d^n} \text{ and } \ket{y} \in \C^{d^m}, 0\le t\le T\},
\\
\S_{prop} &= \spn\{\ket{\eta_\mu, y} \equiv \frac{1}{\sqrt{T+1}}\sum_{t=0}^T \Big(U_t  \cdots U_2 U_1 \ket{\psi}\ket{y}  \Big) \ket{1^t 0^{T-t}}: 1\le \mu \le d^n, 0\le y \le d^{n-1}\}.
\end{align}
Note that $ \L \subset \S_{prop} \subset \S_{clock}$.
Let us denote $\tilde{A}=\A\cap \L^\perp$ for any subspace $\A$.
Note $H_{clock} \S_{clock} = 0$, $H_{prop} \S_{prop} =0$, $H_{in}\L=0$.
We will use the following Projection Lemma~\ref{lem:projection}:
\begin{lemma}[Projection Lemma, adapted from \cite{KKR06}]
 \label{lem:projection}
 Let $H=H_1+H_2$ be sum of two Hamiltonians operating on some Hilbert space $\S_0=\S\oplus\S^\perp$.
 Assuming that $H_2$ has a zero-energy eigenspace $\S\subseteq \S_0$ so that $H_2\S=0$, and that the minimum eigenvalue $\lambda_1(H_2|_{\S^\perp})\ge J > 2\|H_1\|$, then
 \begin{equation}
 \lambda_1(H_1|_\S) - \frac{\|H_1\|^2}{J-2\|H_1\|} \le \lambda_1 (H) \le \lambda_1 (H_1|_\S).
 \end{equation}
 In particular, if $J \ge K\|H_1\|^2 + 2\|H_1\|=O(K\|H_1\|^2)$, we have
 $
 \lambda_1(H_1|_\S) - \frac{1}{K} \le \lambda_1 (H) \le \lambda_1 (H_1|_\S).
 $
\qedextra
\end{lemma}
Applying the above Lemma successively to $H_0$, we obtain
\begin{align}
\lambda_1(H_0|_{\L^\perp} ) &\ge \lambda_1\left[(J_{prop} H_{prop} + J_{in} H_{in} )|_{\tilde\S_{clock}} \right] - \frac{1}{K}
	\quad \text{if} \quad J_{clock} = O(K\|J_{prop} H_{prop} + J_{in} H_{in}\|^2) \\
&\ge \lambda_1\left[( J_{in} H_{in})|_{\tilde\S_{prop}} \right] - \frac{2}{K}
	\quad \text{if} \quad J_{prop}/T^2 = O(K\| J_{in} H_{in}\|^2)
\end{align}
where we used the fact that $\lambda_1(H_{clock}|_{\S_{clock}^\perp})\ge 1$,
and $\lambda_1(H_{prop}|_{\S_{prop}^\perp}) \ge c/T^2$ for some constant $c$.
We now lower bound the last term.
Let us denote $\hat{n}=\Id - \ketbra{0}$.
Then within $\S_{clock}$, we can rewrite
\begin{align}
H_{in}|_{\S_{clock}} &= \sum_{i=1}^m \hat{n}_{i}^\anc \otimes \sum_{0\le t \le t_{\min}(i)} \ketbra{t}^\clock = \sum_{t=0}^{\max_i t_{\min}(i)} H_{in,t} \\
\text{where} \quad H_{in,t} &= \sum_{\{i:~t \le t_{\min}(i)\}} \hat{n}_i^\anc \otimes \ketbra{t}^\clock. \nonumber
\end{align}
In particular, $H_{in,t=0}=\sum_{i=1}^m \hat{n}_i^\anc \otimes \ketbra{t=0}$.
Thus, for any $\ket{\eta_\mu,y},\ket{\eta_\nu,y'}\in \tilde{\S}_{prop}$, where necessarily $y,y'>0$, we have
\begin{eqnarray}
\braket{\eta_\nu,y'|H_{in,t=0}|\eta_\mu,y} &=&  \frac{1}{T+1} \bra{\psi_\nu}\bra{y'} H_{in,t=0} \ket{\psi_\mu}\ket{y} \nonumber \\
&=& \frac{1}{T+1} \delta_{\mu\nu} \braket{y'|\sum_{i=1}^{m}\hat{n}_i^\anc |y}=  \frac{1}{T+1} \delta_{\mu\nu}\delta_{y,y'} \times w(y),
\end{eqnarray}
where $w(y)$ is the Hamming weight of $y$ in $d$-ary representation, which is at least 1 for any $y>0$.
Hence, the minimum eigenvalue of $H_{in,t=0}|_{\L^\perp}$ is $1/(T+1)$.
Since $H_{in}$ consists of only positive semi-definite terms, we have $\lambda_1(H_{in}|_{\tilde{\S}_{prop}}) \ge \lambda_1( H_{in,t=0}|_{\tilde{\S}_{prop}}) \ge 1/(T+1)$.
Thus, to ensure that $H_0$ has spectral gap $\lambda_1(H_0|_{\L^\perp}) \ge \Delta$, we simply choose $J_{in} = O(\Delta(T+1))$, $J_{prop} = O(K T^2 J_{in}^2 m^2)$,
and $J_{clock} = O(KJ_{prop}^2 T^2)=O(\poly(n, T, \Delta))$.

We now show that $\tilde{H}_\circuit$ full-spectrum-simulates $H_\eff$ with only polynomial overhead in energy.
To this end, we use the following result regarding perturbative reductions adapted from Lemma 4 of \cite{BravyiHastingsSim} (also Lemma 35 of \cite{UniversalHamiltonian}):
\begin{lemma}[First-order reduction, adapted from \cite{BravyiHastingsSim}]
Suppose $\tilde{H}=H_0+H_1$, defined on Hilbert space $\tilde\H=\L \oplus \L^\perp$ such that $H_0\L=0$ and $\lambda_1(H_0|_{\L^\perp})\ge \Delta$.
Suppose $H_\eff$ is a Hermitian operator and $V$ is an isometry such that $\| V H_\eff V^\dag - H_1|_{\L} \| \le \xi/2$, then
$\tilde{H}$ full-spectrum-simulates $H_\eff$ to precision $(\eta,\xi)$ below energy cut-off $\Delta/2$, as long as $\Delta \ge O(\xi^{-1}\|H_1\|^2 + \eta^{-1}\|H_1\|)$, per Def.~\ref{defn:CMPsimul}.
In other words, $\|\tilde{H}_{\le\Delta/2}  - \tilde{V} H_\eff \tilde{V}^\dag\| \le \xi$ for some isometry $\tilde{V}$ where $\|\tilde{V}-V\|\le \eta$.
\qedextra
\end{lemma}

We apply the above Lemma  with $H_1=H_{out}$ and $V=\Id$. Note that
we are given in the premise of our Lemma~\ref{lem:circuit-ham-simul} that
\begin{equation}
\|H_\eff - H_1|_{\L}\| = \|H_\eff - H_{out}|_{\L}\| \le \xi/2.
\end{equation}
Hence, $\tilde{H}_\circuit = H_0+H_{out}$ full-spectrum-simulates $H_{\eff}$ to precision $(\eta,\xi)$ below energy cut-off $\Delta/2\ge  O(\xi^{-1}\|H_{out}\|^2 + \eta^{-1}\|H_{out}\|)$.
The maximum interaction strength in $\tilde{H}_\circuit$ is $J_{clock} = O(\poly(n, T, \Delta))= O(\poly(n,T,\xi^{-1},\eta^{-1}, \|H_{out}\|))$.
This concludes the proof of our Lemma~\ref{lem:circuit-ham-simul}.
\end{proof}

We now prove the second Lemma, which shows that in order to ensure the circuit-Hamiltonian simulates a given quasi-groundspace coherently, one only need to add $O(\poly(n)/\epsilon^2)$ identity gates to the end of a polynomial-sized circuit before transforming the circuit to a Hamiltonian.

\begin{lemma}[Idling to enhance coherence]
\label{lem:circuit-idling}
Consider an uncomputed quantum circuit $U_{D}\cdots U_1=\Id$.
Suppose we add $L$ identity gates to the end of the circuit, so that we obtain a new circuit $U=\Id^L U_D\cdots U_1$ with length $T=D+L$.
Let $P=\sum_{\mu=1}^q \ketbra{\psi_\mu}$ and $Q=\sum_{\mu=1}^q \ketbra{\eta_\mu}$, where $\ket{\eta_\mu}$ are history states with respect to $U$ and $\ket{\psi_\mu}$.
If we choose $L=O(D/\epsilon^2)$, then
$\|Q-P\otimes P_\anc\| \le \epsilon$ for some ancilla projector $P_\anc$, regardless of $q$.
\end{lemma}
\begin{proof}
Note that we can write
\begin{align}
\ket{\eta_\mu} &= \sqrt{1-\chi^2}\ket{\psi_\mu}\otimes\ket{\alpha} + \chi \ket{\beta_\mu} \\
\text{where} \quad
\ket{\alpha} &= \frac{1}{\sqrt{L+1}} \ket{0^m}^\anc \otimes \sum_{t=D}^{D+L} \ket{1^t 0^{T-t}}^\clock \\
\ket{\beta_\mu} &= \frac{1}{\sqrt{D}}\sum_{t=0}^{D-1}\Big( U_{t}\cdots U_2 U_1 \ket{\psi_\mu}\ket{0^{m}}^\anc \Big) \ket{1^t 0^{T-t}}^\clock\\
\chi &= \sqrt{D/(D+L+1)}.
\end{align}
Observe that $\bra{\beta_\mu}(\ket{\psi_\nu}\ket{\alpha}) = 0$ since the clock register are at different times, and
\begin{equation}
\braket{\beta_\mu|\beta_\nu} = \frac{1}{D}\sum_{t=0}^{D-1}\braket{\psi_\mu|\psi_\nu} = \delta_{\mu\nu}.
\end{equation}

Hence, for any normalized state $\ket{\phi}\in Q$, we write $\ket{\phi}=\sum_\mu c_\mu\ket{\eta_\mu}$, and find that
\begin{equation}
\left\|(Q-P\otimes\ketbra{\alpha})\ket{\phi}\right\| = \left\|\chi\sum_{\mu} c_\mu \ket{\beta_\mu} \right\| = \chi \left\|\sum_{\mu} c_\mu \ket{\beta_\mu}\right\| = \chi.
\label{eq:bound-in}
\end{equation}
We now use the technical Lemma proved in Appendix~\ref{sec:uniqueGS}, which we restate here:
{
\renewcommand{\thelemma}{\ref{lem:proj-diff}}
\begin{lemma}[Projector Difference Lemma (restatement)]
Consider two Hermitian projectors $\Pi_A$ and $\Pi_B$, such that $\rank(\Pi_A)\le \rank(\Pi_B)$.
Suppose that for all normalized $\ket{\phi}\in \tilde{\Pi_B}$, $\|(\Pi_B-\Pi_A)\ket{\phi}\| \le \delta$.
Then $\|\Pi_B-\Pi_A\| \le \sqrt{2}\delta/\sqrt{1-\delta^2}$.
\qedextra
\end{lemma}
\addtocounter{lemma}{-1}
}
Since $\rank(Q) = \rank(P\otimes \ketbra{\alpha})=q$, then by Lemma~\ref{lem:proj-diff} (identifying $\Pi_B \to Q$, $\Pi_A\to P\otimes\ketbra{\alpha}$, and $\delta\to\chi$), we have
\begin{equation}
\|Q-P\otimes \ketbra{\alpha}\| \le \sqrt{2}\chi/\sqrt{1-\chi^2}
\end{equation}
To make sure $\sqrt{2}\chi/\sqrt{1-\chi^2} \le \epsilon$, it is sufficient to choose $L=O(D/\epsilon^2)$.
\end{proof}

We are now ready to prove our theorem:

\begin{proof}[\textbf{Proof of Theorem~\ref{thm:degree-reduction-poly}}]

Let us denote the normalized eigenstates of $H$ as $\ket{\psi_\mu}$, with corresponding eigenvalues $E_\mu$.
We assume they are ordered such that $E_1  \le E_2 \le E_3 \le \cdots \le E_{d^n}$.
Since the energy spread is $w$ and quasi-spectral gap is $\gamma$, we have
$E_1\le E_\mu\le E_1 + w\gamma$ for $1\le \mu \le q$, and $E_\mu \ge E_1 + \gamma$ for $\mu \ge q+1$, where $q=\rank(P)$ is the quasi-groundspace degeneracy.

\vspace{2pt}
\noindent
\textbf{Part I - Energy measurement via Phase Estimation Circuit}---
Let us first consider the idealized version of the quantum phase estimation algorithm $U_{\rm PE}^\ideal$ for measuring energy with respect to Hamiltonian $H$.
Here, the circuit uses the evolution operator $u_j=e^{-iH\tau 2^{j-1}}$ under $H$, and writes phase of the eigenvalues of $u_1=e^{-iH\tau}$ on some ancilla qubits.
The eigenvalues of $u_1$ are $e^{i2\pi\varphi_\mu}$, where $\varphi_\mu=E_\mu \tau/(2\pi)$;
we choose $\tau$ to satisfy $\tau \le 2\pi/\|H\|$, so that $0 \le \varphi_\mu\le 1$ and we can write $\varphi_\mu = 0.\varphi_{\mu,1}\varphi_{\mu,2}\varphi_{\mu,3}\cdots$.
Ideally, the action of the phase-estimation circuit on input states $\{\ket{\psi_\mu}\ket{0^m}\}_{\mu=1}^{d^n}$ is
\begin{align}
\label{eq:ideal-PE}
U_{\rm PE}^\ideal \ket{\psi_\mu}\ket{0^m} = \ket{\psi_\mu}\ket{E_{\mu}} \ket{\rest_\mu},
\end{align}
where $\ket{\tilde{E}_\mu}=\ket{\varphi_{\mu,1}\varphi_{\mu,2}\varphi_{\mu,3}\cdots\varphi_{\mu,s}}$ is the $s$-bit string representation of the eigenvalue phase $\varphi_\mu$.
Correspondingly, let us denote $\tilde{E}_\mu=2\pi\tilde{\varphi}_\mu/\tau$, where $\tilde{\varphi}_\mu = 0.\varphi_{\mu,1}\varphi_{\mu,2}\varphi_{\mu,3}\cdots\varphi_{\mu,s}$, as approximate values of the energy $E_\mu$.
In the ideal case, $E_\mu=\tilde{E}_\mu$ for some sufficiently large $s$.

In reality, there are two sources of errors that cause the phase-estimation circuit to deviate from $U_{\rm PE}^\ideal$.
The first is due to the fact that the energy don't generally have finite-bit-precision representation,
i.e., $|E_\mu-\tilde{E}_\mu|=O(2^{-s})$ is non-zero.
In other words, since $\varphi_\mu \neq \tilde{\varphi}_\mu$, there's additional error from imprecise phase estimation.
Let us consider a phase estimation circuit $U_{\rm PE}$ implemented to $p$-bit precision, where $p > s$.
Let $b_\mu$ be the integer in the range $[0,2^{p}-1]$ such that $0\le \varphi_\mu - b_\mu/2^p \le 2^{-p}$.
It is well-known~\cite{NielsenChuang} that the action of $U_{\rm PE}$ on any input state $\ket{\psi_\mu}\ket{0}$ result in the following state
\begin{align}
U_{\rm PE} \ket{\psi_\mu}\ket{0^m} &= \ket{\psi_\mu}\ket{\rest_\mu'} \otimes \frac{1}{2^{p}}\sum_{k,\ell=0}^{2^p-1} e^{-i2\pi  k \ell/2^p} e^{i 2\pi \varphi_\mu k} \ket{\ell} = \ket{\psi_\mu}\ket{\rest_\mu'}\sum_{\ell=0}^{2^p-1}\alpha_\ell^\mu \ket{\ell}
\end{align}
where $\ket{\ell}=\ket{\ell_1\cdots \ell_p}$ is the binary representation of $\ell$, and
\begin{align}
\alpha_\ell^\mu &= \frac{1}{2^p}\sum_{k=0}^{2^p-1}[e^{i 2\pi (\varphi_\mu-\ell/2^p)}]^k = \frac{1}{2^p} \left[\frac{1-e^{i 2\pi (2^{s}\varphi_\mu-\ell)}}{1-e^{i 2\pi (\varphi_\mu-\ell/2^p)}} \right]
\end{align}
The analysis from Sec.~5.2.1 in \cite{NielsenChuang} shows that the probability of getting a state that is a distance of $e$ integer away is
\begin{equation}
p_\mu^\text{error}(e) \equiv \sum_{|\ell-b_\mu| > e} |\alpha_\ell^\mu|^2 \le \frac{1}{2(e-1)}
\end{equation}
Note that we only care about the first $s<p$ bits, so we can choose $e=2^{p-s}-1$.Hence,
\begin{align}
U_{\rm PE} \ket{\psi_\mu}\ket{0^m} &=  \ket{\psi_\mu}\ket{\rest_\mu'} \otimes \left[ \sum_{|\ell-b_\mu|\le e} \alpha_\ell^\mu \ket{\ell}  + \sum_{|\ell-b_\mu|> e} \alpha_\ell^\mu \ket{\ell} \right]  \nonumber \\
&= \ket{\psi_\mu}\ket{\rest_\mu'}\otimes\left(\sqrt{1-p_\mu^\text{error}}\ket{\tilde{E}_\mu}\ket{\rest^1_\mu} + \sqrt{p_\mu^\text{error}}\ket{\rest^2_\mu}\right)
\end{align}
Comparing this with the idealized output in Eq.~\eqref{eq:ideal-PE}, we can identify $\ket{\rest_\mu}=\ket{\rest_\mu'}\ket{\rest_\mu^1}$, and observe that
\begin{align}
(U_{\rm PE} - U_{\rm PE}^\ideal) \ket{\psi_\mu}\ket{0^m} = \ket{\psi_\mu}\ket{\text{error}_\mu}, \quad
\text{where}
\quad
\left\|\ket{\text{error}_\mu}\right\|^2 \le 2p^\text{error}_\mu = O(2^{-(p-s)})
\end{align}
Thus, for any normalized state $\ket{\psi}=\sum_\mu c_\mu \ket{\psi_\mu}\ket{0^m}$, we have
\begin{align}
\|(U_{\rm PE} - U_{\rm PE}^\ideal)\sum_\mu c_\mu \ket{\psi_\mu}\ket{0^m}\|^2 = O(2^{-(p-s)}) = O(1/\poly(n))
\end{align}
where we assume that we can choose, for example, $p=2s$ and $s=O(\log(n))$, and thus make this first source of error due to imprecision to be polynomially small.

The second source of error is due to the fact that we need to implement the circuit using only local gates, in order to ensure the corresponding circuit-Hamiltonian is local,
The only non-local gates that we need to address are $u_j=e^{-iH\tau_j}$, where $\tau_j=2^{j-1}\tau$.
This can be implemented with local gates via ``Trotterization''.
Specifically, we write $H=\sum_a H_a$, where $H_a$ is a $k$-local term,
and implement $\tilde{u}_j=(\prod_a e^{-iH_a \tau_j/r_j})^{r_j}$ for some integer $r_j$, so that $\|\tilde{u}_j-u_j\| \le O(\tau_j^2/r_j)$.
Assuming $s=O(\log n)$ and so $\tau_j = O(\poly(n))$, we can then choose $r_j = O(\tau_j^2 \poly(n))= O(\poly(n))$ to ensure each such error is polynomially small.
The error from Trotterization is bounded by
\begin{equation}
\|U_{\rm PE}^\local - U_{\rm PE}\| \le \sum_{j=1}^s O(\tau_j^2/r_j) = O(1/\poly(n)).
\end{equation}

In sum, we can choose any $\zeta=O(1/\poly(n)$, and construct the actual phase estimation circuit $U_{\rm PE}^\local$ in such a way that it is $\zeta$-close to $U_{\rm PE}^\ideal$ on any valid input state $\ket{\psi}\ket{0^m}$:
\begin{align}
\|(U_{\rm PE}^\local - U_{\rm PE}^\ideal) \ket{\psi}\ket{0^m} \| \le \zeta \equiv O(1/\poly(n)). \label{eq:PE-error}
\end{align}

\noindent
\textbf{Part II -- Constructing degree-reducer Hamiltonian from circuit}---
We first replace $U_{\rm PE}^\local$ with a sparsified version, so that $\deg(U_{\rm PE}^\local)=O(1)$.
This can be done by adding swap gates and ancilla qudits after each gate, so that the computational states are mapped to the new ancilla qudits.
Assuming each gate is $k$-local, this only increase the total number of gates and qudits by a factor of $k$, and the error from idealized phase-estimation is still bounded by Eq.~\eqref{eq:PE-error}

Suppose the sparsified circuit $U_{\rm PE}^\local$ now has $t_0$ gates.
Then we construct the following circuit
\begin{align}
U_\circuit = (\Id)^L U_{\rm PE}^{\local\dag} (\Id)^s U_{\rm PE}^\local.
\end{align}
Note we add $U_{\rm PE}^{\local \dag}$ for uncomputing and $s+L$ idling identity gates, making the entire circuit gate count $T=2t_0+s+L$.
The $s$ identity gates are used for local measurements of energy to $s$-bit precision, and $L=O((2t_0+s)/\epsilon^2)=O(\poly(n)/\epsilon^2)$ identity gates are used to ensure $\epsilon$-incoherence as in Lemma~\ref{lem:circuit-idling}.
The history states with respect to eigenstate $\ket{\psi_\mu}$ of $H$ and this circuit are
\begin{align}
\ket{\eta_\mu} =\frac{1}{\sqrt{T+1}}\sum_{t=0}^T \Big( U_t \cdots U_2 U_1 \ket{\psi_\mu}\ket{0^m} \Big) \ket{1^t 0^{T-t}}
\end{align}

We can convert the circuit to Hamiltonian $\tilde{H}_\circuit$ using the method described in Lemma~\ref{lem:circuit-ham-simul}, where $H_{out}$ is chosen to be
\begin{align}
H_{out} = (T+1)\sum_{b=1}^s \frac{2\pi}{\tau} 2^{-b}\ketbra{1}_b^\anc \otimes P^\clock(t=t_0+b).
\end{align}
Here, we denote $P^\clock(t)=\ketbra{110}_{t-1,t,t+1}^\clock$ as the effective projector onto legal clock states corresponding to time step $t$.

To show that $\tilde{H}_\circuit$ simulates the original Hamiltonian $H$, we first show that $H_{out}$ restricted to the subspace of history states $\L=\spn\{\ket{\eta_\mu}:1\le \mu \le d^n \}$ can be approximated by the following effective Hamiltonian
\begin{equation}
H_\eff = \sum_\mu \tilde{E}_\mu \ketbra{\eta_\mu}.
\end{equation}
Consider arbitrary states $\ket{\eta}\in \L_-$.
We write $\ket{\eta} = \sum_{\mu} a_\mu \ket{\eta_\mu}$, and observe
\begin{align}
\braket{\eta| H_{out}|\eta}  &= \frac{2\pi}{\tau} \sum_{b=1}^s 2^{-b} \left[\sum_{\nu} a_\nu^* \bra{\psi_\nu}\bra{0^m} U_{\rm PE}^{\local\dag} \right] \ketbra{1}_b \left[\sum_{\mu} a_\mu U_{\rm PE}^{\local} \ket{\psi_\mu}\ket{0^m}\right] \nonumber \\
&= \frac{2\pi}{\tau} \sum_{b=1}^s 2^{-b} \left[\sum_{\nu} a_\nu^* \bra{\psi_\nu}\bra{\tilde{E}_\nu}\right] \ketbra{1}_b \left[\sum_{\mu} a_\mu \ket{\psi_\mu}\ket{\tilde{E}_\mu} \right] + O(\|H\|\|(U_{\rm PE}^{\local}  - U_{\rm PE}^\text{ideal})\ket{\psi}\ket{0^m}\|) \nonumber \\
&= \sum_\mu |a_\mu|^2 \tilde{E}_\mu + O(\zeta\|H\|)
\end{align}
Note the extra factor of $\|H\|$ comes from $\tau=O(1/\|H\|)$.
Since we can choose $\zeta=O(1/\poly(n))$ to be arbitrarily polynomially small, we write $\xi/2=O(\zeta\|H\|)=O(1/\poly(n))$.
Hence,
\begin{equation}
\left|\braket{\eta |  H_{out} - H_\eff|\eta}\right| \le \xi/2 \quad \forall \ket{\eta} \in \L_- \quad
\Longrightarrow \quad
\|H_\eff - H_{out}|_{\L_-} \| \le \xi/2
\end{equation}
Thus by Lemma~\ref{lem:circuit-ham-simul}, for any $\eta>0$, the constructed $\tilde{H}_\circuit$ full-spectrum-simulates $H_\eff$ to precision $(\eta,\xi)$ below some energy cut-off $\Delta=O(\xi^{-1}\|H_{out}\|^2 + \eta^{-1}\|H_{out}\|)$.

Finally, we show that a slightly rescaled $\tilde{H}_\circuit$ gap-simulates $H$.
Note since we had chosen $O(log n)$-bit precision, $\|\tilde{E}_\mu - E_\mu|\le O(1/\poly(n)$.
Thus, similar to the given Hamiltonian $H$, $H_{\eff}$ has the quasi-groundspace projector $Q=\sum_{\mu=1}^q \ketbra{\eta_\mu}$ with quasi-spectral gap $\gamma+O(1/\poly(n))$ and energy spread $w+O(1/\poly(n))$.
Thus, we can choose $\xi=\Theta(\epsilon\gamma)=\Omega(\epsilon/\poly(n))$ and $\eta=\Theta(\epsilon)$ and $\alpha=\frac43 + O(1/\poly(n))$, and then use Lemma~\ref{lem:relate-CMP} to show that $\tilde{H}_\circuit' = \alpha \tilde{H}_\circuit$, gap-simulates $H_\eff$ with incoherence $\epsilon/2$, while the overall energy scale is still polynomial as $\Delta\ge (\xi^{-1}\|H_{out}\|^2 + \eta^{-1}\|H_{out}\|) = O(\poly(n,\epsilon^{-1}))$.
In other words, if $\tilde{P}$ is the quasi-groundspace projector of $\tilde{H}_\circuit'$ with quasi-spectral gap $\gamma$, then $\|\tilde{P} - Q\| \le \epsilon/2$.
Furthermore, by Lemma~\ref{lem:circuit-idling}, we have $\|Q-P\otimes P_\anc\|\le \epsilon/2$ since we have added $L=O(\poly(n)/\epsilon^2)$ identity gates to the end of the circuit.
This implies that
\begin{equation}
\|\tilde{P} - P\otimes P_\anc\| \le \|\tilde{P} - Q\| + \|Q-P\otimes P_\anc\| \le \epsilon
\end{equation}
This means $\tilde{H}_\circuit'$ gap-simulates $(H,P)$ with incoherence $\epsilon$ and energy spread $w+O(1/\poly(n))$.
Finally, we note that
$\tilde{H}_\circuit$ has degree $O(\deg(U))=O(1)$, $O(T)=O(\poly(n)/\epsilon^2)$ number of terms, and $O(\poly(n,T,\xi^{-1},\eta^{-1}, \|H_{out}\|)) = O(\poly(n,\epsilon^{-1}))$ interaction strength.
This concludes our proof of Theorem~\ref{thm:degree-reduction-poly}.
\end{proof}

\subsection{Coherent Dilution and DR of $H_\oneone$ with Polynomial Interaction Strength (Proposition~\ref{prop:circuit})\label{sec:prop-circuit}}

Recall that we have developed an incoherent diluter and degree-reducer of $H_\oneone$ in Prop.~\ref{prop:incoherent-tree} from Appendix~\ref{sec:incoherenttree}.
Although the construction only need bounded interaction strength, it was completely incoherent.
As shown in Theorem~\ref{thm:degree-reduction-poly}, coherent degree-reduction is possible by allowing polynomial-strength interactions.
It turns out that for some special cases such as $H_\oneone$, coherent dilution is also possible with polynomial-strength interaction.

This construction is done by constructing a similar circuit as in Prop.~\ref{prop:incoherent-tree} that counts the number of excitations locally and arranging them on a tree geometry so that individual qubit/qudit has constant degree.
Since the circuit has $O(n)$ gates and constant degree, we can apply the circuit-to-Hamiltonian mapping as in Lemma~\ref{lem:circuit-ham-simul}, producing a coherent diluter and degree-reducer of $H_\oneone$.

{
\renewcommand{\theprop}{\ref{prop:circuit}}
\begin{prop}[constant-coherence dilution and DR for $H_A$ with polynomial interaction strength]
There is a 6-local $[6, O(n/\epsilon^2),O(\poly(n,\epsilon^{-1}))]$-degree-reducer of $H_\oneone$ with $\epsilon$-incoherence, energy spread $\tilde{w}=0$, and trivial encoding. This is also a $[6, O(n/\epsilon^2),O(\poly(n,\epsilon^{-1}))]$-diluter of $H_\oneone$.
\end{prop}
\addtocounter{prop}{-1}
}

\begin{proof}

First, let us construct a circuit to count the number of excitations.
Similar to Fig.~\ref{fig:tree}, we add $n-1$ ancilla qutrits to the system of $n$ qubits and arrange them in a tree.
In this arrangement, the original qubits are on the leaf nodes, and the ancilla qutrits are on the internal nodes.
Note qutrits are particles with three possible states: $\ket{0},\ket{1},\ket{2}$.

Starting from the internal nodes right below the leaf nodes, we label each internal node with an index $t$ and work our way down to the root node, so that no parent node has an index smaller than its children. For $t=1,2,\ldots,n-2$, we apply a gate $U_t$ for $t$-th internal node, where $U_t$ is a 3-local unitary satisfying
\begin{gather}
U_t\ket{00}\ket{z} = \ket{00}\ket{z},\quad
U_t\ket{10}\ket{z} = \ket{10}\ket{z\oplus1},\quad
U_t\ket{01}\ket{z} = \ket{01}\ket{z\oplus1}
\\U_t\ket{xy}\ket{z} = \ket{xy}\ket{2-z} \quad \forall xy\neq00,01,10,
\end{gather}
where $\ket{lr}\ket{b}$ denote the state where the internal node qutrit is in the state $\ket{b}$, while its left (right) child node is in the state $\ket{l}$ ($\ket{r}$).
Here, we denote $\oplus$ as addition modulo 3.
For $t=n-1$, we apply $U_{n-1}$ on the root node and its two children where $U_{n-1}$ satisfies
\begin{eqnarray}
U_{n-1}\ket{xy}\ket{z}&=&\ket{xy}\ket{z} \quad \quad \text{for} \quad xy=00,01,10,\\
\text{and} \quad
U_{n-1}\ket{xy}\ket{z}&=&\ket{xy}\ket{z\oplus1} \quad \textnormal{otherwise.}
\end{eqnarray}
Assuming the ancilla qutrits are initialized at $\ket{0}$, this circuit checks how many excitations (1s) are among the $n$ system qubits, and keep the root node qutrit at $\ket{0}$ if and only if the circuit accepts the input when less than two excitations are counted.

Since there are $n-1$ internal nodes, there are also $n-1$ such gates.
However, since we need to maintain coherence on the ancilla by un-computing, the full circuit should be
\begin{equation} \label{eq:treecircuit}
U_\textnormal{circuit} = U_1^\dag U_2^\dag \cdots U^\dag_{n-2} U_{n-1} U_{n-2}\cdots U_2 U_1.
\end{equation}
Note this circuit consists of $D\equiv 2(n-1)-1=2n-3$ gates, where $U_{n-1}$ is the only gate that acts on the root node.
It is then clear that for input states of the form $\ket{x_1\cdots x_n}\otimes\ket{0^{n-1}}^\anc$, where $x_i\in\{0,1\}$ are states of the original qubits, the final output of the full circuit is $\ket{x_1\cdots x_n}\ket{0^{n-2}}\ket{0}$ if there are zero or one excitations amongst $x_i$, and $\ket{x_1\cdots x_n}\ket{0^{n-2}}\ket{1}$ otherwise.

\begin{figure}[t]
\centering
\includegraphics[width=0.9\textwidth]{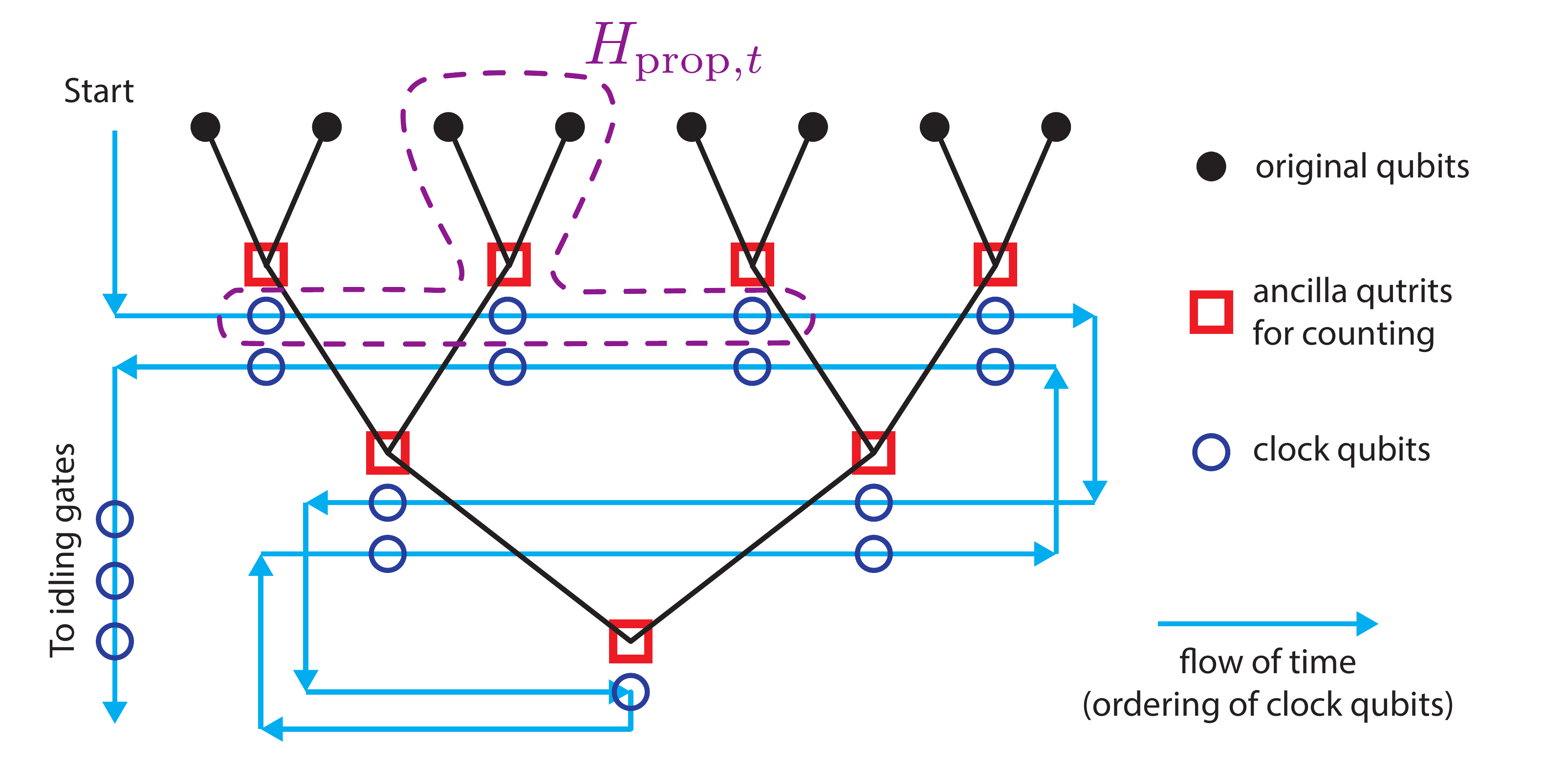}
\caption{\label{fig:circuit}Schematic diagram of a constant-degree, poly-strength diluter and degree-reducer of $H_\oneone$ via circuit Hamiltonian on a tree, for $n=8$. $H_{prop,t}$ are 6-local term that connects all the qudits at each branching as well as the nearby clock qubits on the flow of time. The clock qubits are linked together in the blue line that represents flow of time through $H_{clock}$, which consists of 2-local terms acting on nearest neighbors on the line. Each ancilla qutrit is also connected to the clock qubit directly below through $H_{in}$ and $H_{out}$.}
\end{figure}

Per Lemma~\ref{lem:circuit-idling}, we add $L=O(D/\epsilon^2)$ identity gates to the end of $U_\circuit$ to ensure $\epsilon$-incoherence.
We then convert the resultant circuit with $T=D+L$ gates to a Hamiltonian $\Hcirc=H_0 + H_\oneone^{out}$ using the construction laid out in Lemma~\ref{lem:circuit-ham-simul}.
To ensure that our circuit only accepts outputs where the root ancilla qutrit is $\ket{0}$ once the circuit reaches it (with $U_{n-1}$ at $t=n-1$), we use the following Hamiltonian
\begin{equation}
H_\oneone^{out} = J_{out} (\ketbra{1}_{n-1}^\anc+\ketbra{2}_{n-1}^\anc)\otimes \ketbra{1}^\clock_{n-1}.
\end{equation}

Let $\ket{\psi_\mu}$ be eigenstates of $H_\oneone$.
For $1\le \mu\le n+1$, let $\ket{\psi_\mu}$ be the groundstates of $H_\oneone$.
We consider the restriction to the subspace of history states, $\L=\spn\{\ket{\eta_\mu}:1\le\mu \le 2^n\}$.
Let us write $\tilde{P}=\sum_{\mu=1}^{n+1} \ketbra{\eta_\mu}$ as the projector onto history states corresponding to groundstate input.
It is easy to see that $\Hcirc\tilde{P} = H_\oneone^{out} \tilde{P}=0$, so $\tilde{P}$ is in fact the groundspace projector of $\Hcirc$.
Let us denote $U_{t\leftarrow0}=U_t\cdots U_2 U_1$.
We note that for any $\ket{\psi_\mu},\ket{\psi_\nu} \in \L$ but $\perp \tilde{P}$, we have
\begin{align}
\braket{\eta_\mu|H_\oneone^{out}|\eta_\nu} &= \frac{J_{out}}{T+1}\sum_{t',t=0}^T \left(\bra{\psi_\mu} \bra{0^{n-1}}\bra{t'}\right) U_{t'\leftarrow 0}^\dag H_{out} U_{t\leftarrow 0} \left(\ket{\psi_\nu}\ket{0^{n-1}}\ket{t}\right) \nonumber \\
&= \frac{J_{out}}{T+1}\sum_{t',t=n-1}^T \left(\bra{\psi_\mu} \bra{0^{n-1}}\right) U_{t'\leftarrow 0}^\dag U_{t\leftarrow 0} \left(\ket{\psi_\nu}\ket{0^{n-1}}\right) \delta_{t,t'}\nonumber \\
&= \frac{J_{out}}{T+1}\sum_{t=n-1}^T \braket{\psi'|\psi} = J_{out} \frac{T+2-n}{T+1}\braket{\psi'|\psi} = J_{out} \frac{n-1+L}{2(n-1)+L}\braket{\psi'|\psi}.
\end{align}
Consequently, $H_\oneone^{out}|_{\L}$ is diagonal in the basis of history states, and we can write
\begin{equation}
H_\oneone^{out}|_\L = \tilde\gamma (\Id - Q) \equiv H_\eff, \quad \text{where}\quad
\tilde\gamma \equiv J_{out}\frac{n-1+L}{2(n-1)+L} \ge J_{out}/2
\end{equation}
We can then apply Lemma~\ref{lem:circuit-ham-simul} to show that for any $\eta,\xi>0$, the constructed $\Hcirc$ full-spectrum-simulates $H_\eff$ to precision $(\eta,\xi)$ below energy cut-off $\Delta\ge O(\xi^{-1} J_{out}^2 + \eta^{-1}J_{out})$, with trivial encoding.
By choosing $\xi=O(\epsilon J_{out})$ and $\eta=O(\epsilon$), we can apply Lemma~\ref{lem:relate-CMP} and show that $\frac{4}{3}\Hcirc$ gap-simulates $H_\eff$ with incoherence $\epsilon$ and trivial encoding.
We note that by choosing $J_{out}=2$, we ensure that the spectral gap of $\Hcirc$ is $\ge 1$.
Furthermore, we note that $\Hcirc$ in fact has energy spread $\tilde{w}=0$ since $\Hcirc \tilde{P}=0$.
Since $\|\tilde{P}-P\otimes P_\anc\|\le \epsilon$ by Lemma~\ref{lem:circuit-idling}, we have shown that
$\frac43 \Hcirc$ gap-simulates $H$ with incoherence $\epsilon$, energy spread $\tilde{w}=0$, and trivial encoding, whose maximum interaction strength is $O(\poly(n,\epsilon^{-1})$.

A schematic illustrating the connectivity of the qubits/qutrits in $\tilde{H}_\oneone^\textnormal{circuit}$ is shown in Fig.~\ref{fig:circuit}.
By inspection we see that the computation qubits have max degree 2, the ancilla qutrits have max degree 5, and the clock qubits have max degree 6.
The Hamiltonian consists of $M=O(n+T)=O(n/\epsilon^2)$ terms, each at most 6-local.
Hence, $\tilde{H}_\oneone^\textnormal{circuit}$ is a 6-local $[6, O(n/\epsilon^2),O(\poly(n,\epsilon^{-1}))]$-sparsifier of $H_\oneone$ with $\epsilon$-incoherence.
\end{proof}

\section{Bounding error of perturbed groundspace  (Lemma~\ref{lem:PPgroundspace} and \ref{lem:PPgsv2})\label{sec:PPgroundspace-proof}}

In this Appendix, we prove Lemma~\ref{lem:PPgroundspace} which is used earlier in Appendix~\ref{sec:comp-defns} and \ref{sec:imposs-dilute}.
In fact, we prove a more general version than the version previously stated.
This is Lemma~\ref{lem:PPgsv2}, of which Lemma~\ref{lem:PPgroundspace} is a special case with the further restriction that $w\le1/2$.
Note this Lemma uses essentially the same technique as Lemma A.1 from Ref.~\cite{OliveiraTerhal}.

\begin{lemma}[Error bound on perturbed groundspace]
\label{lem:PPgsv2}
Let $H$ and $\tilde{H}$ be two Hamiltonians on the same Hilbert space.
Per Def.~\ref{defn:gap}, let $P$ project onto a quasi-groundspace of $H$ with energy spread $w$ and quasi-spectral gap $\gamma$.
If $\|\tilde{H}-H\| \le \kappa$, and $\kappa < (1-w)\gamma/4$, then there is a quasi-groundspace projector $\tilde{P}$ of $\tilde{H}$ with quasi-spectral gap at least $\tilde{\gamma}$, comprised of eigenstates of $\tilde{H}$ up to energy at most $\lambda_1(\tilde{H}) + \tilde{w}\tilde{\gamma}$, where
$\tilde{\gamma} \ge \gamma-2\kappa$ and $\tilde{w}\tilde{\gamma} \le w\gamma + 2\kappa$. Furthermore,
\begin{equation}
\|\tilde{P}- P \| < \kappa \left(\frac{4}{(1-w)\gamma} +  \frac{2(1+w)/(1-w)}{(1-w)\gamma-4\kappa} \right).
\end{equation}
In particular, if $w\le 1/2$ and $\kappa \le (1-w)\gamma/8$, then $\|\tilde{P}-P\| < 32\kappa/\gamma$.
\end{lemma}

To prove the above Lemma, we borrow the Green's function techniques from Ref.~\cite{KKR06,OliveiraTerhal} to bound error due to perturbations.

First, let us establish some notations, similar to that in Ref.~\cite{KKR06,OliveiraTerhal}.
We consider Hamiltonians of the form $\tilde{H} = H + V$, defined on some Hilbert space $\tilde{\H}$.
(Note the symbol $V$ in this Appendix refers to a Hermitian operator, not an isometry.)
Furthermore, we assume $H$ has a gap of width $\Delta>0$ in the spectrum centered at some value $\lambda_*$;
in other words, no eigenvalue of $H$ lies between $\lambda_- = \lambda_* - \Delta/2$ and $\lambda_+ = \lambda_* + \Delta/2$.
We decompose the Hilbert space $\tilde{\H}=\L_+\oplus\L_-$, where $\L_-$ is the low-energy subspace of eigenstates of $H$ with eigenvalue  $\le \lambda_-$, and $\L_+$ corresponds to high-energy eigenstates of $H$ with eigenvalues $\ge \lambda_+$.
Correspondingly, we denote $\Pi_\pm$ as projectors onto subspaces $\L_\pm$.
Furthermore, let us denote the operator-valued Green's functions $G(z) = (z-H)^{-1}$ and $\tilde{G}(z) = (z - \tilde{H})^{-1}$.
We can decompose all operators in the Hilbert space $\tilde{H}$ into four blocks according to $\L_{\pm}$:
\begin{equation}
\begin{aligned}
H &= \begin{pmatrix}
H_{+} & 0 \\
0 & H_{-} \\
\end{pmatrix}, \quad
&V &= \begin{pmatrix}
V_{++} & V_{+-} \\
V_{-+} & V_{--} \\
\end{pmatrix}, &\quad
\tilde{H} =
\begin{pmatrix}
\tilde{H}_{++} & \tilde{H}_{+-} \\
\tilde{H}_{-+} & \tilde{H}_{--}
\end{pmatrix}, \\
G &= \begin{pmatrix}
G_{+} & 0 \\
0 & G_- \\
\end{pmatrix},
\quad
&\tilde{G} &=
\begin{pmatrix}
\tilde{G}_{++} & \tilde{G}_{+-} \\
\tilde{G}_{-+} & \tilde{G}_{--}
\end{pmatrix}.
\end{aligned}
\label{eq:gadget-convention-1}
\end{equation}
We denote $A_{\pm\pm}=\Pi_{\pm}A\Pi_{\pm}$ and $A_{\pm\mp}= \Pi_{\pm}A\Pi_{\mp}$ as parts of operator $A$ restricted to mapping between the corresponding subspaces.
In cases when the operator is block diagonal in this basis, we simplify notation by denoting $G_+\equiv G_{++}$ for example.
Finally, we define the self-energy $\Sigma_-(z)$ as the following operator acting on the subspace $\L_-$:
\begin{equation}
\Sigma_-(z) = z - \tilde{G}_{--}^{-1}(z) = H_{-} + V_{--} +  \sum_{p=0}^\infty V_{-+} (G_+ V_{++})^p G_+ V_{+-}
\label{eq:gadget-convention-2}
\end{equation}
where after the last equality, we wrote out the series expansion of $\Sigma_-(z)$ that will be very useful.

Before we proceed to the proof of Lemma~\ref{lem:PPgsv2}, we first state a useful result proved in Ref.~\cite{KKR06}:

\begin{lemma}[Error bound on perturbed eigenvalues, Theorem 3 of \cite{KKR06}]
\label{lem:gadget-eigenvalue}
Consider a Hamiltonian $\tilde{H}= H + V$.
Let us denote a precision parameter $\eps>0$, and assume the existence of a Hermitian operator $H_\eff$ whose eigenvalues lie in some range $[a,b]$.
Suppose that all the following conditions are satisfied:
\begin{itemize}
\item For some constants $\Delta>0$ and $\lambda_* > b+\eps$, $H$ has no eigenvalues between $\lambda_- = \lambda_* - \Delta/2$ and $\lambda_+ = \lambda_* + \Delta/2$.
\item $\|V\| < \Delta/2$.
\item For all $z\in [a-\eps,b+\eps]$, the following inequality holds for the self-energy:
\begin{equation}
\|\Sigma_-(z) - H_\eff\| \le \eps
\end{equation}
\end{itemize}
Let $\tilde{H}|_{<\Delta/2}$ denote the operator $\tilde{H}$ restricted to eigenstates with eigenvalues $<\lambda_*$.
Then
\begin{equation}
\left|\lambda_j\left(\tilde{H}|_{<\lambda_*}\right) - \lambda_j\left(H_\eff\right)\right| \le \eps \quad \forall j
\end{equation}
where $\lambda_j(X)$ is the $j$-th eigenvalue of Hermitian operator $X$.
\end{lemma}

\begin{proof}[\textbf{Proof of Lemma~\ref{lem:PPgsv2}}]
Let $E^g=\lambda_1(H)$.
WLOG let us assume $E^g=0$, since otherwise we can simply redefine $H\mapsto H'=H-E^g$ and $\tilde{H} \mapsto \tilde{H}'=\tilde{H}-E^g$, which have the same spectrum as the original $H$ and $\tilde{H}$ except with the eigenvalues shifted by $E^g$.

Note by Def.~\ref{defn:gap}, $H$ has no eigenvalue between $\lambda_- =  w\gamma$ and $\lambda_+= \gamma$, so there is a gap $\Delta=(1-w)\gamma$ in the spectrum of $H$ centered at $\lambda_* = \frac12(\lambda_+ + \lambda_-) = (1+w\gamma)/2$.
Thus let us decompose the Hilbert space $\H=\L_+ \oplus \L_-$, where $\L_+$ ($\L_-$) corresponds to eigenstates of $H$ with eigenvalue $\ge \lambda_+$ ($\le \lambda_-$).
We note that $P=\Pi_-$.
Now, let us denote $V = \tilde{H}-H$, which satisfies $\|V\|\le \kappa$ by assumption.
We consider a region $R = \{z\in \mathds{C}:|z|\le (1+w)\gamma/2\}$ in the complex plane, which is a disk centered at $z=0$ with radius $r=(1+w)\gamma/2$.
For any $z\in R$, we have
$\|G_+(z)\| = \| \Pi_+(z-H)^{-1} \Pi_+ \| \le 2/[(1-w)\gamma] = 2/\Delta$.
Thus, treating $H_-\equiv \Pi_- H \Pi_-$, which is $H$ restricted to $\L_-$, as the effective Hermitian operator $H_\eff=H_-$, we have
\begin{eqnarray}
\Sigma_-(z) - H_-  &=& V_{--} + \sum_{p=0}^\infty V_{-+} (G_+ V_{++})^p G_+ V_{+-} \nonumber \\
\|\Sigma_-(z) - H_- \| &\le & \|V_{--}\|+ \sum_{p=0}^\infty \|V_{-+}\|^2 \|G_+\|^{p+1} \|V_{++}\|^p   \nonumber \\
&\le& \kappa + \sum_{p=0}^\infty \frac{\kappa^{p+2}}{(\Delta/2)^{p+1}} = \frac{\kappa\Delta}{\Delta -2\kappa} \equiv \eps
\label{eq:PP-self-energy-bound}
\end{eqnarray}
Observe that the region $R$ includes the interval $[-\eps, w\gamma+\eps]$, if
\begin{equation}
\eps < \frac12 (1-w)\gamma = \frac{\Delta}{2}
\quad \Longleftrightarrow \quad
\kappa < \frac{1}{4}(1-w)\gamma = \frac{\Delta}{4},
\label{eq:small-kappa}
\end{equation}
which is what we assumed in the premise of the Lemma.

Since $\|\tilde{H}-H\|\le \kappa$, then by Weyl's inequality we have $|\lambda_j(\tilde{H})-\lambda_j(H)|\le \kappa$ for all $j$.
Note $\lambda_1(H)=0$, so $|\lambda_1(\tilde{H})|\le \kappa$.
Let us denote $\tilde{P}$ as the projector onto the corresponding eigenstates of $\tilde{H}$ with eigenvalue $< \lambda^*$.
Note the eigenstates in $P$ of $H$ have maximum eigenvalue $\le \lambda_-$;
consequently, the eigenstates in $\tilde{P}$ of $\tilde{H}$ has maximum eigenvalue $\le \lambda_-+\kappa$, and all other eigenstates of $\tilde{H}$ have eigenvalue $\ge \lambda^+-\kappa  = \gamma - \kappa$.
Hence, $\tilde{P}$ corresponds to a quasi-groundspace of $\tilde{H}$ with quasi-spectral gap $\tilde{\gamma}$ and energy spread $\tilde{w}$, given by:
\begin{align}
\tilde{\gamma} &\ge \lambda_+-\kappa -\lambda_1(\tilde{H}) \ge \gamma-2\kappa \\
\tilde{w}\tilde{\gamma} &\le \lambda_- + \kappa -\lambda_1(\tilde{H}) \le w\gamma + 2\kappa
\end{align}

Now, let us bound the error $\|\tilde{P}-P\|$ between the groundspace projectors, following
the same idea in Ref.~\cite{OliveiraTerhal}.
We first bound
\begin{equation}
\|\tilde{P}-\Pi_- \tilde{P}\Pi_-\| = \|\tilde{P} - \Pi_- \tilde{P} + \Pi_- \tilde{P} - \Pi_- \tilde{P}\Pi_-\| = \|\Pi_+ \tilde{P} + \Pi_- \tilde{P} \Pi_+\| \le 2\|\Pi_+ \tilde{P}\|
\end{equation}
Furthermore, we can bound the quantity $\|\Pi_+ H \tilde{P}\|$ from above:
\begin{eqnarray}
\|\Pi_+ H \tilde{P}\| &=& \|\Pi_+ (\tilde{H}-V) \tilde{P}\| \le \|\Pi_+ \tilde{H} \tilde{P}\| + \|\Pi_+ V \tilde{P}\| \le \|\Pi_+ \tilde{P} \tilde{H}\tilde{P}\| + \|V\| \nonumber \\
&\le& (\lambda_- + \eps)\|\Pi_+ \tilde{P}\| + \|V\|.
\end{eqnarray}
where we used the fact that $\|\tilde{P}\tilde{H}\tilde{P}\| \le \lambda_-+\eps$.
Using $\|\Pi_+ H\Pi_+\| \ge \lambda_+$, we can also bound $\|\Pi_+ H \tilde{P}\|$ from below:
\begin{equation}
\|\Pi_+ H \tilde{P}\| = \|\Pi_+ H \Pi_+ \tilde{P}\| \ge \lambda_+ \|\Pi_+ \tilde{P}\|
\end{equation}
Thus
\begin{equation}
\|\tilde{P}-\Pi_- \tilde{P}\Pi_-\|  \le \frac{2\|V\|}{\lambda_+ - \lambda_- - \eps} = \frac{2\|V\|}{\Delta - \eps} \le \frac{2\kappa}{\Delta-\eps} < \frac{4\kappa}{\Delta} = \frac{4\kappa}{(1-w)\gamma},
\label{eq:PP-bound-1}
\end{equation}
where we used the assumption in Eq.~\eqref{eq:small-kappa} that $\eps < \Delta/2$, so $(\Delta-\eps)^{-1} < 2/\Delta$.
Now, let us bound $\|\Pi_- \tilde{P} \Pi_- - P\|$.
To this end, let us denote $C$ as the contour in the complex plane around the region $R$, centered at $z=0$ with radius $r=(1+w)\gamma/2$.
Due to Lemma~\ref{lem:gadget-eigenvalue}, we know all eigenvalues of $\tilde{H}$ that correspond to $\tilde{P}$ are enclosed by $C$.
Using the Cauchy integral formula, we have
\begin{equation}
\Pi_- \tilde{P} \Pi_- = \Pi_-\left(\frac{1}{2\pi i} \oint_C \tilde{G}(z) dz\right)\Pi_- = \frac{1}{2\pi i} \oint_C \tilde{G}_{--}(z)dz = \frac{1}{2\pi i} \oint_C (z-\Sigma_-(z))^{-1}dz.
\end{equation}
Also, observe that
\begin{equation}
P =   \frac{1}{2\pi i} \oint_C G_-(z) =  \frac{1}{2\pi i} \oint_C (z-H_-)^{-1} dz.
\end{equation}
To bound the difference between the two operators, we use the following identity
\begin{equation}
\|(A-B)^{-1} - A^{-1}\| = \|(\Id-A^{-1}B)^{-1}A^{-1} - A^{-1}\|
\le
\|A^{-1}\|\left( (1-\|A^{-1}\| \|B\|)^{-1} - 1\right),
\end{equation}
which is true if $\|A^{-1}\|\|B\|< 1$.
By choosing $A=z-H_-$ and $B = \Sigma_-(z) - H_-$ for $z\in C$ (on the contour), we have $\|B\|\le \eps$ by Eq.~\eqref{eq:PP-self-energy-bound}, and $\|A^{-1}\| \le (r-w\gamma)^{-1}=2/\Delta$ by inspection.
Hence $\|A^{-1}\|\|B\| \le 2\eps/\Delta$, which is $<1$ since we assumed $\eps < \Delta/2$ as in Eq.~\eqref{eq:small-kappa}.
Then, we can apply the aforementioned identity and bound
\begin{eqnarray}
\sup_{z\in C} \|(z-\Sigma_-(z))^{-1} - (z-H_\eff)^{-1}\| \le
\frac{4\eps}{\Delta(\Delta-2\eps)}= \frac{4\kappa}{\Delta(\Delta-4\kappa)}.
\end{eqnarray}
Consequently, we have
\begin{eqnarray}
\|\Pi_-\tilde{P}\Pi_- - P\| &=& \left\| \frac{1}{2\pi i} \oint_C [(z-\Sigma_-(z))^{-1} - (z-H_\eff)^{-1}]dz \right \| \nonumber \\
&\le&  \frac{4\kappa r}{\Delta(\Delta-4\kappa)}  =
 \frac{2\kappa(1+w)/(1-w)}{(1-w)\gamma-4\kappa} ,
\end{eqnarray}
where we plugged in $r=(1+w)\gamma/2$ and $\Delta=(1-w)\gamma$. Combining with the first bound in Eq.~\eqref{eq:PP-bound-1}, we have
\begin{equation}
\|\tilde{P}- P \| < \kappa \left(\frac{4}{(1-w)\gamma} +  \frac{2(1+w)/(1-w)}{(1-w)\gamma-4\kappa} \right)
\end{equation}

Let us now consider the particular case when $\kappa \le (1-w)\gamma/8$ and $w\le 1/2$.
Given these constraints, we have $(1-w)^{-1}\le 2$ and $1/[(1-w)\gamma-4\kappa] \le 2/[(1-w)\gamma]$, which implies
\begin{equation}
\|\tilde{P}-P\| < \kappa \left(\frac{4}{(1-w)\gamma} + \frac{4(1+w)}{(1-w)^2\gamma} \right) \le \frac{32\kappa}{\gamma}.
\end{equation}
\end{proof}

\section{General Coherent DR with Exponential Interaction Strength \\ (Theorem~\ref{thm:degree-reduction-exp})\label{sec:degree-reduction-unbounded} }

In this Appendix, we prove Theorem~\ref{thm:degree-reduction-exp}, which shows that given unbounded interaction strength, one can perform degree-reduction for arbitrary local Hamiltonians.
The proof makes heavy use of perturbative gadgets.
Specifically, we use versions of subdivision, 3-to-2-to-local, and fork gadgets first presented in Ref.~\cite{OliveiraTerhal} to construct a 2-local coherent degree-reducer for any given local Hamiltonian.
The analyses in Ref.~\cite{KKR06, CaoImprovedOTGadget} have also provided some inspirations.

The proof can be divided into two sections.
In Sec.~\ref{subsec:perturbativeGtools}, we will show that the three above mentioned perturbative gadget tools can indeed be used for gap-simulation (Definition~\ref{defn:hamsimul}).
To this end we first prove Lemma~\ref{lem:gadget-ground-space}
which is a cousin of  Lemma~\ref{lem:PPgsv2} used previously,
providing error bound on perturbed groundspace.
Then, Claims~\ref{claim:subdiv}, \ref{claim:3-to-2} and \ref{claim:fork} prove the applicability of the three tools of perturbative gadgets to our coherent gap-simulation framework, respectively.

Subsequently, in Sec.~\ref{sec:proof-theorem-gadget} we use these
tools in a fairly straight-forward sequence of mappings to degree-reduce any $O(1)$-local Hamiltonian.

\subsection{Gap-simulation by Perturbative Gadgets}\label{subsec:perturbativeGtools}

Perturbative gadgets are Hamiltonians of the form $\tilde{H}=H_\anc + V$, where $V$ contains perturbations that act on highly degenerate groundstates of $H_\anc$ and produce effective interactions that mimic some target Hamiltonian $H_\eff$.
(We emphasize that the symbol $V$ in this Appendix refers to a Hermitian operator, not an isometry.)
Generally, the quality of how well the gadget Hamiltonian simulates the target Hamiltonian is given by a precision parameter $\eps \ll 1$ that one can freely choose at the end of the construction (see for example the statement of Lemma~\ref{lem:gadget-eigenvalue}).

To prove the results in this section,
we use the same Green's function machinery described above in Appendix~\ref{sec:PPgroundspace-proof}, which studies perturbation theory on $\tilde{H}=H+V$.
The notation we use is the same as in Eq.~\eqref{eq:gadget-convention-1}\eqref{eq:gadget-convention-2},
except we change $H\to H_\anc$.
Note that, Lemma~\ref{lem:gadget-eigenvalue} already allow us to bound eigenvalues of $\tilde{H}$ relative to $H_\eff$, which can allow us to satisfy condition 1 of gap-simulation Definition~\ref{defn:hamsimul}
We still need to bound errors of perturbed (quasi-)groundspace to satisfy
condition 2 of the definition.
To that end, we prove the following Lemma, which is a cousin of Lemma~\ref{lem:PPgsv2}.
The proof uses essentially the same arguments as in Lemma A.1 from Ref.~\cite{OliveiraTerhal} and Lemma~\ref{lem:PPgsv2},
but adapted to prove statements more directly useful for the goals in the section.

\begin{lemma}[Gadget groundspace error bound, modified from Ref.~\cite{OliveiraTerhal}]
\label{lem:gadget-ground-space}
Suppose we are given a target Hamiltonian $H_\targ$ defined on Hilbert space $\H$, and let $E^g = \lambda_1(H_\targ)$.
Let us denote its quasi-groundspace projector $P$, energy spread $w$, and quasi-spectral gap $\gamma$ per Def.~\ref{defn:gap}.
Additionally, let us denote $q\equiv \textnormal{rank}(P)$ as the degeneracy of the quasi-groundspace.
Now, consider a gadget Hamiltonian $\tilde{H} = H_\anc + V$ acting on Hilbert space $\tilde{\H}=\H\otimes \H_\anc$, and some precision parameter $\eps$ such that $0<\eps < (1-w)\gamma/2$.
Suppose the following conditions are satisfied:
\begin{itemize}
\item $H_\anc$ acts trivially on $\H$. When restricted to the ancilla Hilbert space $\H_\anc$, we denote $P_\anc$ as the projector onto the eigenstates of $H_\anc$ with eigenvalue $\lambda_-=0$, and all other eigenvalues are $\ge \lambda_+ = \Delta$.
In other words, the subspace $\L_- = \textnormal{range}(\Id\otimes P_\anc)$.
\item The conditions of Lemma~\ref{lem:gadget-eigenvalue} are satisfied with $H_\eff =  (H_\targ \otimes P_\anc) |_{\L_-}$ and precision parameter $\eps$.
\item Consider again the self energy $\Sigma_-(z) \equiv z - \tilde{G}_{--}^{-1}(z)$ now generalized for $z\in \C$.
For some constant $r$ satisfying $w\gamma + \eps < r \le \frac12(1+w)\gamma$, we have for all $|z - E^g | \le r$,
\begin{equation}
\|\Sigma_-(z) - H_\eff\|\le \eps
\end{equation}
\end{itemize}
Let $\tilde{P}$ be the projector onto the $q$ lowest eigenstates of $\tilde{H}$. Then
\begin{equation} \label{eq:ground-space-error}
\| \tilde{P} - P\otimes P_\anc\| \le \frac{2 \|V\|}{\Delta- (|E^g| + w\gamma+\eps)} + \frac{\eps r}{(r-w\gamma)(r-w\gamma-\eps)}
\end{equation}
In particular, if we choose $r=(1+w)\gamma/2$, then
\begin{equation}
\| \tilde{P} - P\otimes P_\anc\| \le \frac{2 \|V\|}{\Delta- (|E^g| + w\gamma+\eps)}+ 2\eps\frac{1+w}{(1-w)[(1-w)\gamma-2\eps]}
\end{equation}
\end{lemma}

\begin{proof}
Per our convention, we denote $\Pi_-=\Id\otimes P_\anc$ and $\Pi_+ = \Id-\Pi_-$ as projectors onto low- and high- energy subspace $\L_\mp$ of $H_\anc$.
The proof proceeds in two parts: bounding $\|\tilde{P} - \Pi_- \tilde{P} \Pi_-\|$ and $\|\Pi_-\tilde{P} \Pi_- - P\otimes P_\anc\|$, which together yields Eq.~\eqref{eq:ground-space-error}.

\emph{Part 1}---
Using the triangle inequality for the spectral norm, we can bound
\begin{equation}
\|\tilde{P} - \Pi_- \tilde{P} \Pi_-\| = \|\tilde{P} - \Pi_- \tilde{P} + \Pi_-\tilde{P}  - \Pi_- \tilde{P}\Pi_- \| = \|\Pi_+\tilde{P} + \Pi_- \tilde{P} \Pi_+\| \le 2\|\Pi_+ \tilde{P}\|.
\end{equation}
Observe that since $\lambda_1(H_\eff)=E^g$ and $\lambda_q(H_\eff) \le E^g + w\gamma$, Lemma~\ref{lem:gadget-eigenvalue} tells us $\lambda_1(\tilde{H}) \ge E^g -\eps$ and $\lambda_q(\tilde{H}) \le E^g + w\gamma+\eps$, which means $\|\tilde{P} \tilde{H} \tilde{P}\|\le \max\{|\lambda_1(\tilde{H})|, |\lambda_q(\tilde{H})| \} \le |E^g| + w\gamma + \eps$.
This allows us to bound the quantity from $\|\Pi_+ H_\anc \tilde{P}\|$ above:
\begin{eqnarray}
\|\Pi_+ H_\anc \tilde{P}\| &=& \|\Pi_+ (\tilde{H} - V)\tilde{P}\| \le \|\Pi_+ \tilde{H}\tilde{P}\| +  \| \Pi_+ V \tilde{P} \| = \|\Pi_+ \tilde{P}\tilde{H}\tilde{P}\| +  \| \Pi_+ V \tilde{P} \| \nonumber \\
&\le& (|E^g| + w\gamma + \eps)\|\Pi_+ \tilde{P}\| +  \|V\|.
\end{eqnarray}
Also, we can bound the same quantity from below
\begin{equation}
\|\Pi_+ H_\anc \tilde{P}\| = \|\Pi_+ H_\anc \Pi_+ \tilde{P}\| \ge \lambda_+ \|\Pi_+ \tilde{P}\| = \Delta\|\Pi_+ \tilde{P}\|.
\end{equation}
Consequently, we obtain the bound
\begin{equation}
\|\tilde{P}- \Pi_- \tilde{P} \Pi_-\| \le \frac{2\|V\|}{\Delta - (|E^g| + w\gamma + \eps)}
\end{equation}

\emph{Part 2}---
Consider a circular contour $C$ in the complex plane centered at $z=E^g$ with radius $r$ satisfying the assumption $w\gamma + \eps < r \le \frac12(1+w)\gamma$.
Due to Lemma~\ref{lem:gadget-eigenvalue}, we know all eigenvalues of $\tilde{H}$ corresponding to $\tilde{P}$ is inside $C$.
Using the Cauchy integral formula, we have
\begin{equation}
\Pi_- \tilde{P} \Pi_- = \Pi_-\left(\frac{1}{2\pi i} \oint_C \tilde{G}(z) dz\right)\Pi_- = \frac{1}{2\pi i} \oint_C \tilde{G}_{--}(z)dz = \frac{1}{2\pi i} \oint_C (z-\Sigma_-(z))^{-1}dz.
\end{equation}
Also, observe that
\begin{equation}
P\otimes P_\anc = \frac{1}{2\pi i} \oint_C (z-H_\eff)^{-1} dz.
\end{equation}
To bound the difference between the two operators, we use the following identity
\begin{equation}
\|(A-B)^{-1} - A^{-1}\| = \|(\Id-A^{-1}B)^{-1}A^{-1} - A^{-1}\|
\le
\|A^{-1}\|\left( (1-\|A^{-1}\| \|B\|)^{-1} - 1\right),
\end{equation}
which is true if $\|A^{-1}\|\|B\|< 1$.
Let us choosing $A=z-H_\eff$ and $B = \Sigma_-(z) - H_\eff$ for $z\in C$ (on the contour).
Observe that we have $\|B\|\le \eps$ by assumption.
Since we assumed $r\le (1+w)\gamma/2$, we have $r-w\gamma \le \frac12 (1-w)\gamma \le \gamma - r$, and thus$\|A^{-1}\| \le (r-w\gamma)^{-1}$.
Also, since we assumed $w\gamma+\eps < r$, we have $\|A^{-1}\|\|B\| \le \eps/(r-w\gamma) < 1$.
Therefore, we can apply the aforementioned identity and bound
\begin{eqnarray}
\sup_{z\in C} \|(z-\Sigma_-(z))^{-1} - (z-H_\eff)^{-1}\| \le \frac{\eps}{(r-w\gamma)(r-w\gamma - \eps)}.
\end{eqnarray}
Consequently, using the estimation lemma for contour integrals,  we have
\begin{equation}
\|\Pi_-\tilde{P}\Pi_- - P\otimes P_\anc\| = \left\| \frac{1}{2\pi i} \oint_C [(z-\Sigma_-(z))^{-1} - (z-H_\eff)^{-1}]dz \right \| \le \frac{\eps r}{(r-w\gamma)(r-w\gamma -\eps)}.
\end{equation}
Eq.~\eqref{eq:ground-space-error} thus follows.

\end{proof}

Now we will use the above Lemma~\ref{lem:gadget-eigenvalue} and \ref{lem:gadget-ground-space} to prove three Claims about how different gadget reductions can produce gap-simulating Hamiltonians.
We will then use these Claims to prove Theorem~\ref{thm:degree-reduction-exp}.
In the following, we denote $X\equiv \sigma_x$, $Y \equiv \sigma_y$, $Z\equiv \sigma_z$ for convenience.

\begin{claim}[gap-simulation by subdivision gadget]
\label{claim:subdiv}
Given an $n$-qubit $k$-local Hamiltonian $H_\targ$ with a quasi-groundspace projector $P$ with quasi-spectral gap $\gamma$ and energy spread $w$.
Let us write it as
\begin{equation}
H_\targ = H_\els +  \sum_{\mu=1}^m c_{\mu} \sigma_{\mu_1}^{(s_{\mu_1})}\otimes \sigma_{\mu_2}^{(s_{\mu_2})} \otimes \cdots \otimes \sigma_{\mu_k}^{(s_{\mu_k})} , \quad \sigma_{\mu_i} \in \{ X, Y, Z\}
\end{equation}
where $H_\els$ is some $k'$-local for $k'\le  \lceil k/2 \rceil + 1$.
where $|c_\mu| \le J$.
Let $A_\mu = \sigma_{\mu_1}^{(i_1)} \otimes \cdots \otimes \sigma_{\mu_{j}}^{(i_j)}$, and $B_\mu = \sigma_{\mu_{j+1}}^{(i_{j+1})}\otimes \cdots  \otimes \sigma_{\mu_{k}}^{(i_k)}$, where $j=\lceil k/2\rceil$.
Considering adding an ancilla qubit $a_\mu$ for each $1\le \mu \le m$,
and write the following $\eps$-precision $(\lceil k/2 \rceil +1)$-local \emph{subdivision gadget} Hamiltonian on $n+m$ qubits
\begin{gather}
\tilde{H}_\gadget = H_\anc + V, \quad
H_\anc = \sum_{\mu} \Delta \ketbra{1}^{(a_\mu)}, \\
V = H_\els + \sum_\mu \left[ \sqrt{\frac{|c_\mu|\Delta}{2}} ( \sgn(c_\mu)A_\mu -  B_\mu ) \otimes X^{(a_{\mu})} + |c_\mu| \right]
\end{gather}
assuming we choose $\eps \ll (1-w)\gamma$ and
\begin{equation}
\Delta = O\left(\frac{m^2 J(m^4J^2 + \|H_\els\|)}{\eps^2}\right)
\end{equation}
Let $c=\gamma/(\gamma-2\eps)=O(1)$, then $\tilde{H} = c\tilde{H}_\gadget$  gap-simulates $(H_\targ,P)$ with incoherence $\epsilon=\O(\eps/(1-w)^2\gamma)$ and energy spread $\tilde{w}\le w+2\eps/\gamma$.
\end{claim}

\begin{proof}
Let us denote $P_\anc = \bigotimes_\mu \ketbra{0}^{(a_\mu)}=\ketbra{\vect{a}=0}$ which projects onto the ancilla state described by the binary string $\vect{a}=0$.
It is also the groundspace projector of $H_\anc$.
Let us denote $\Pi_- = \Id\otimes P_\anc$ and $\Pi_+ = \Id-\Pi_-$ be two projectors that partition the full Hilbert space into $\L_-$ and $\L_+$ respectively.
We now follow the same convention laid out in Eq.~\eqref{eq:gadget-convention-1} and \eqref{eq:gadget-convention-2}.
Note $G_+(z) = \Pi_+ (z-H_\anc)^{-1} \Pi_+ = \sum_{\vect{a}\neq 0} \ketbra{\vect{a}}/(z-h(\vect{a})\Delta)$, where $h(\vect{a})$ is the Hamming weight of the binary string $\vect{a}$.
Observe that
\begin{eqnarray}
V_{--} &=& (H_\els + \sum_\mu |c_\mu|)\otimes P_\anc \\
V_{-+} G_+ V_{+-} &=& \sum_\mu \frac{|c_\mu|\Delta}{2(z-\Delta)} ( \sgn(c_\mu)  A_\mu - B_\mu )^2 \otimes P_\anc = \frac{\Delta}{z-\Delta} \sum_\mu (|c_\mu|-c_\mu A_\mu \otimes B_\mu ) \otimes P_\anc
\nonumber \\
&=& \sum_{\mu} (c_\mu A_\mu \otimes B_\mu - |c_\mu|)\otimes P_\anc  + \frac{z}{z-\Delta} \sum_\mu (|c_\mu|-c_\mu A_\mu \otimes B_\mu ) \otimes P_\anc.
\end{eqnarray}
We used the fact that in the second-order perturbation $V_{-+} G_+ V_{+-}$, there's no ``cross-gadget'' term because $\sum_{\mu,\mu'} \Pi_- X^{(a_\mu)} G_+ X^{(a_\mu')}\Pi_- = \delta_{\mu,\mu'}\Pi_-/(z-\Delta)$.
Noting that $H_{\anc,-}=0$, we have
\begin{eqnarray}
\Sigma_-(z) = H\otimes P_\anc +\underbrace{\frac{z}{z-\Delta} \sum_\mu (|c_\mu|-c_\mu A_\mu \otimes B_\mu ) \otimes P_\anc }_{E_1} + \underbrace{\sum_{p=1}^\infty V_{-+}(G_+ V_{++})^p G_+ V_{+-}}_{E_2}
\end{eqnarray}
We want to show $\|\Sigma_-(z) - H_\eff\|< \eps$ for an appropriate range of $z$ if $\Delta$ is sufficiently large.
Consider $|z| \le \|H\|+\eps$, which is sufficient for applying Lemma~\ref{lem:gadget-eigenvalue} and \ref{lem:gadget-ground-space}.
Let us assume we choose $\Delta \gg 2\|H\|$, and consequently $\Delta \gg J$.
Then we can bound the first error term $\|E_1\| \le O(mJ/\Delta)$.
Note we have $\|G_+(z)\| \le 1/(\Delta - \|H\|) \le 2/\Delta$, $\|V_{-+}\| \le O(m \sqrt{J\Delta})$ and $\|V_{++}\| \le \|V\|  = \|H_\els\| + O(mJ) + O(m\sqrt{J\Delta}) = \|H_\els\| + O(m\sqrt{J\Delta})$, and thus
\begin{equation}
\|E_2\| \le \sum_{p=1}^\infty \frac{\|V^{(0)}_{-+}\|^2\|V^{(0)}_{++}\|^p}{(\Delta/2)^{p+1}} = \frac{4 \|V_{-+}\|^2 \|V_{++}\|}{\Delta(\Delta - 2\|V_{++}\|)} \le  O\left(\frac{m^3 J^{3/2}}{\Delta^{1/2}} \right) + O\left(\frac{m^2 J \|H_\els\|}{\Delta}\right)
\end{equation}
To make sure $\|E_1\| + \|E_2\| \le \eps$, we need $\Delta = \Omega(m^6 J^3/\eps^2)$ and $\Delta=\Omega(m^2 J \|H_\els\|/\eps)$.
Hence, a sufficient choice for $\Delta$ is
\begin{equation}
\Delta = O\left(\frac{m^2 J(m^4J^2 + \|H_\els\|)}{\eps^2}\right)
\end{equation}
Note this choice would also ensure $\|V\|/\Delta \ll \eps$.
Let us denote $\tilde{P}$ as the quasi-groundspace projector of $\tilde{H}_\gadget $ corresponding to the lowest $\rank(P)$ eigenstates.
By applying Lemma~\ref{lem:gadget-eigenvalue}, we can see that the corresponding $H_\targ$ and $\tilde{H}_\gadget$  differ by at most $\eps$.
By rescaling $\tilde{H}_\gadget\mapsto \tilde{H} = c \tilde{H}_\gadget$, where $c=\gamma/(\gamma-2\eps)$, we can ensure $\tilde{H}$ has $\tilde{P}$ as a quasi-groundspace projector with quasi-spectral gap at least $\gamma$.
Furthermore, assuming $\eps < (1-w)\gamma/2$, we can bound the energy spread of $\tilde{P}$ in $\tilde{H}$ is by
\begin{equation}
\tilde{w} \le 	\frac{w \gamma + 2\eps}{\gamma} = w + \frac{2\eps}{\gamma}.
\end{equation}
By applying Lemma~\ref{lem:gadget-ground-space} with $r=(1+w)\gamma/2$, noting that $|E^g| \le \|H_\targ\| \ll \Delta$ and $\eps \ll (1-w)\gamma/2$, we can bound incoherence by
\begin{equation}
\|\tilde{P} - P\otimes P_\anc\| \le \O\left(\frac{\eps}{(1-w)^2\gamma}\right).
\end{equation}
\end{proof}

\begin{claim}[gap-simulation by 3-to-2-local gadget]
\label{claim:3-to-2}
Given a $n$-qubit 3-local Hamiltonian $H_\targ$ with a quasi-groundspace projector $P$ with quasi-spectral gap $\gamma$ and energy spread $w$.
Let us write it as
\begin{equation}
H_\targ = H_\els + \sum_{\mu=1}^m c_\mu A_\mu^{(i_\mu)} \otimes B_\mu^{(j_\mu)} \otimes C_\mu^{(k_\mu)}
\quad A_\mu, B_\mu, C_\mu \in \{X,Y,Z\}
\end{equation}
where $|c_\mu|\le \sqrt{\Delta_0}$ for some $\Delta_0$.
We assume $H_\els$ contains only 2-local terms and $\|H_\els\| \le O(m \Delta_0)$.
Consider adding an ancilla qubit $a_\mu$ for each $1\le \mu \le m$.
Then consider the following $\eps$-precision 2-local gadget Hamiltonian on $n+m$ qubits
\begin{equation}
\begin{aligned}
& \tilde{H}_\gadget = H_\anc + V, \quad H_\anc = \sum_\mu \Delta \ketbra{1}^{(a_\mu)}, \quad V = V_1 + V_2 \\
& V_1 =  H_\els +  \sum_\mu \left[ \frac12 \Delta^{1/3}(A_\mu^{(i_\mu)} - B_\mu^{(j_\mu)})^2 + c_\mu C_\mu^{(k_\mu)} \otimes \ketbra{0}^{(a_\mu)}   \right], \\
& V_2 = \Delta^{2/3} \sum_\mu \left[\frac{1}{\sqrt{2}} (A_\mu^{(i_\mu)} - B_\mu^{(j_\mu)})\otimes X^{(a_\mu)}  - c_\mu C_\mu^{(k_\mu)} \otimes \ketbra{1}^{(a_\mu)}\right] \\
\end{aligned}
\end{equation}
where $\Delta = O(m^{12} \Delta_0^3 /\eps^3)$, where we assume $\eps \ll (1-w)\gamma/2$.
Let $c=\gamma/(\gamma-2\eps)=O(1)$, then $\tilde{H}=c\tilde{H}_\gadget$
gap-simulates $(H_\targ, P)$ with incoherence $\epsilon = \O(\eps/(1-w)^2\gamma)$ and energy spread $\tilde{w} \le w + 2\eps/\gamma$,
\end{claim}

\begin{proof}
Let us denote $P_\anc = \bigotimes_\mu \ketbra{0}^{(a_\mu)}=\ketbra{\vect{a}=0}$ which projects onto the ancilla state described by the binary string $\vect{a}=0$.
It is also the groundspace projector of $H_\anc$.
Let us denote $\Pi_- = \Id\otimes P_\anc$ and $\Pi_+ = \Id-\Pi_-$ be two projectors that partition the full Hilbert space into $\L_-$ and $\L_+$ respectively.
We now follow the same convention laid out in Eq.~\eqref{eq:gadget-convention-1} and \eqref{eq:gadget-convention-2}.
Note $G_+(z) = \Pi_+ (z-H_\anc)^{-1} \Pi_+ = \sum_{\vect{a}\neq 0} \ketbra{\vect{a}}/(z-h(\vect{a})\Delta)$, where $h(\vect{a})$ is the Hamming weight of the binary string $\vect{a}$.

In the following, we will simplify notation by denoting $A_\mu=A_\mu^{(i_\mu)}, B_\mu\equiv B_\mu^{(j_\mu)}, C_\mu \equiv C_\mu^{(k_\mu)}$, and $X_\mu \equiv X_\mu^{(a_\mu)}$.
Observe that
\begin{eqnarray}
V_{--} &=& V_{1,--} = H_\els \otimes P_\anc +  \sum_\mu \left[ \frac12 \Delta^{1/3}(A_\mu - B_\mu)^2 + c_\mu C_\mu   \right]\otimes P_\anc \\
V_{-+} &=& \Delta^{2/3} \sum_\mu \left[\frac{1}{\sqrt{2}} (A_\mu - B_\mu)\otimes \ketbrat{0}{1}^{(a_\mu)}\right]
\end{eqnarray}
Then
\begin{eqnarray}
V_{-+} G_+ V_{+-} &=& \frac{\Delta^{4/3}}{z-\Delta} \sum_\mu \frac{1}{2}(A_\mu-B_\mu)^2\otimes P_\anc  \nonumber\\
V_{--} + V_{-+} G_+ V_{+-} &=& \left[H_\els + \sum_\mu c_\mu C_\mu\right]\otimes P_\anc + \underbrace{\frac{z\Delta^{1/3}}{2(z-\Delta)}	\sum_\mu (A_\mu - B_\mu)^2  \otimes P_\anc}_{E_1}
\end{eqnarray}
At the third order perturbation theory, the only allowed virtual transition on the ancilla qubits are of the form $\ket{0\cdots0}\to\ket{0\cdots010\cdots0} \to\ket{0\cdots010\cdots0}  \to \ket{0\cdots0}$.
In other words, the only non-zero terms at this order involve exciting some ancilla $a_\mu$ from $\ket{0}$ to $\ket{1}$ by $V_{+-}$, keeping it at $\ket{1}$ by $V_{++}$, and then return it to $\ket{0}$ by $V_{-+}$.
Hence
\begin{eqnarray*}
&& V_{-+} G_+ V_{++} G_+ V_{+-} = -\frac{\Delta^2}{2(z-\Delta)^2}\sum_\mu c_\mu (A_\mu-B_\mu)^2C_\mu \otimes P_\anc
\nonumber  \\
&& + \underbrace {\frac{1}{(z-\Delta)^2}\Bigg[ \frac{\Delta^{4/3}}{2} \sum_\mu (A_\mu-B_\mu) H_\els (A_\mu - B_\mu)+\sum_{\mu,\mu'} \frac{\Delta^{5/3}}{4} (A_\mu - B_\mu)^2(A_{\mu'} - B_{\mu'})^2
\Bigg]\otimes P_\anc}_{E_2}
\end{eqnarray*}
Let us denote $\xi = \Delta^2/[(z-\Delta)^2]$, then
\begin{eqnarray}
&& V_{--} + V_{-+} G_+ V_{+-} + V_{-+} G_+ V_{++} G_+ V_{+-}\nonumber \\
&=& \left[H_\els + \sum_\mu \left(c_\mu C_\mu  -\frac{\xi}{2} c_\mu (A_\mu-B_\mu)^2C_\mu\right)\right] \otimes P_\anc + E_1 + E_2 \nonumber \\
&=& \left[H_\els + \xi \sum_\mu c_\mu A_\mu B_\mu C_\mu  + (1-\xi)\sum_\mu c_\mu C_\mu
\right]\otimes P_\anc + E_1 + E_2 \nonumber \\
&=& \underbrace{[H_\els + \sum_\mu c_\mu A_\mu B_\mu C_\mu]}_{H_\targ} \otimes P_\anc + \underbrace{(1-\xi) \sum_\mu   c_\mu(C_\mu - A_\mu B_\mu C_\mu)\otimes P_\anc}_{E_3} + E_1 + E_2
\end{eqnarray}
Thus, we can write the self-energy as
\begin{eqnarray}
\Sigma_-(z) &=& H_\targ\otimes P_\anc + E_1 + E_2 + E_3 + \underbrace{\sum_{p=2}^\infty V_{-+}(G_+ V_{++})^p G_+ V_{+-}}_{E_4}
\end{eqnarray}

Suppose we choose the relevant range of $|z|\le z_{\max} = \|H_\els\| = O(m\Delta_0)$. Assuming we will choose $\Delta \gg \|H_\els\|$, we have
\begin{eqnarray}
\|E_1\| &\le&  O\left(m \frac{ \|H_\els\| }{\Delta^{2/3}} \right) = O\left(\frac{m^2\Delta_0}{\Delta^{2/3}}\right) \\
\|E_2\| &\le& O\left(\frac{m z_{\max}}{\Delta^{2/3}}\right) + O\left(\frac{m^2}{\Delta^{1/3}}\right)  = O\left(\frac{m^2\Delta_0}{\Delta^{2/3}}\right) + O\left(\frac{m^2}{\Delta^{1/3}}\right)\\
\|E_3\| &\le& 2m \sqrt{\Delta_0}|1-\xi| = 2m \sqrt{\Delta_0} \frac{(2\Delta - \|H_\els\| )\|H_\els\| }{(\Delta-\|H_\els\| )^2} \le  O \left(\frac{m^2\Delta_0^{3/2}}{\Delta}\right) \\
\|E_4\| &\le& \sum_{p=2}^\infty \frac{\|V_{-+}\|^2 \|V_{++}\|^p}{(\Delta/2)^{p+1}} = \frac{8\|V_{-+}\|^2 \|V_{++}\|^2}{\Delta^2(\Delta-2\|V_{++}\|)} \le O\left(\frac{m^4 \Delta_0}{\Delta^{1/3}}\right)+ O\left(\frac{m^4\Delta_0^2}{\Delta^{5/3}}\right)
\end{eqnarray}
where we used the fact that
$\|V_{-+}\|= O(m\Delta^{2/3})$ and $\|V_{++}\|\le \|V\| = O(m \sqrt{\Delta_0} \Delta^{2/3}) + \|H_\els\| = O(m\sqrt{\Delta_0} \Delta^{2/3}) + O(m\Delta_0)$.
Hence, as sufficient choice for $\Delta$ to ensure that $\|\Sigma_-(z) - H_\eff\|\le \eps\|$ would be
\begin{equation}
\Delta = O(m^{12} \Delta_0^3 /\eps^3)
\end{equation}
Let us denote $\tilde{P}$ as the quasi-groundspace projector of $\tilde{H}_\gadget$ corresponding to the lowest $\rank(P)$ eigenstates.
By applying Lemma~\ref{lem:gadget-eigenvalue}, we can see that the corresponding $H_\targ$ and $\tilde{H}_\gadget$  differ by at most $\eps$.
By rescaling $\tilde{H}_\gadget\mapsto \tilde{H} = c \tilde{H}_\gadget$, where $c=\gamma/(\gamma-2\eps)$, we can ensure $\tilde{H}$ has $\tilde{P}$ as a quasi-groundspace projector with quasi-spectral gap at least $\gamma$.
Furthermore, assuming $\eps < (1-w)\gamma/2$, we can bound the energy spread of $\tilde{P}$ in $\tilde{H}$ is by
\begin{equation}
\tilde{w} \le 	\frac{w \gamma + 2\eps}{\gamma} = w + \frac{2\eps}{\gamma}.
\end{equation}
By applying Lemma~\ref{lem:gadget-ground-space} with $r=(1+w)\gamma/2$, noting that $|E^g| \le \|H_\targ\| \ll \Delta$ and $\eps \ll (1-w)\gamma/2$, we can bound incoherence by
\begin{equation}
\|\tilde{P} - P\otimes P_\anc\| \le \O\left(\frac{\eps}{(1-w)^2\gamma}\right).
\end{equation}

\end{proof}

\begin{claim}[degree-reduction via fork gadget]
\label{claim:fork}
Consider a 2-local Hamiltonian $H_\targ$ of the form
\begin{equation}
H_\targ = H_\els + \sum_{i=1}^n \sum_{\alpha=x,y,z}\sum_{\kappa_{\alpha,i}=1}^{r_{\alpha,i}} \lambda_{\kappa_{\alpha,i}} \sigma^{(i)}_{\alpha} \otimes X^{(\kappa_{\alpha,i})}
\end{equation}
which contains $n$ original qubits interacting only with ancilla qubits $\kappa_{\alpha,i}$ through $\sigma_\alpha\otimes X$.
Let $r_0 \equiv \max_{i,\alpha}r_{\alpha,i} = O(\poly n)$ be the maximum ``Pauli degree'',
and $|\lambda_{\kappa_{\alpha,i}}| \le \sqrt{\Delta_0}$.
We assume $H_\els$ does not act on the $n$ original qubits, and contains $O(n r_0)$ terms with interaction strength at most $\Delta_0$.
Lastly, we assume the ancilla (non-original) qubits in $H_\targ$ has degree at most 5.
Now let $P$ be a quasi-groundspace projector of $H_\targ$ with energy spread $w$ and quasi-spectral gap $\gamma$.
Then for some precision parameter $\eps \ll (1-w)\gamma$, there is a Hamiltonian $\tilde{H}$ that gap-simulates $(H_\targ, P)$ with incoherence $\epsilon = O((\eps/\gamma) \log n)$ and  energy spread $\tilde{w} \le w + O((\eps/\gamma) \log n)$, and has maximum degree 6, $O(nr_0)$ terms with interaction strength $J=O((\poly(n) \Delta_0/\eps^2)^{\poly(n)})$.
\end{claim}

\begin{proof}
We can reduce the degrees of original qubits with serial application of the ``fork gadget''~\cite{OliveiraTerhal}.
We will apply the fork gadget in $S=O(\log_2(r_0))=O(\log(n))$ iterative steps, starting with the target Hamiltonian $H_\targ^{(1)} = H_\targ$, producing gadget Hamiltonian that gap-simulates the target, which then becomes the target Hamiltonian for the next step:
\begin{equation}
[H_\targ \equiv H_\targ^{(1)} ]\to [\tilde{H}_\gadget^{(1)} \equiv H_\targ^{(2)} ] \to \cdots
\to [\tilde{H}_\gadget^{(S-1)} \equiv H_\targ^{(S-1)}] \to \tilde{H}_\gadget^{(S)} = \tilde{H}
\end{equation}
At each step $s=1,\ldots,S$,
our target Hamiltonian is of the form
\begin{eqnarray}
H_\targ^{(s)} \equiv  \tilde{H}_\gadget^{(s-1)} = H_\els^{(s)} + \sum_{i=1}^n \sum_{\alpha=x,y,z}\sum_{\kappa_{\alpha,i}^s=1}^{r_{\alpha,i}^{(s)}} \lambda_{\kappa_{\alpha,i}^s} \sigma^{(i)}_{\alpha} \otimes X^{(\kappa_{\alpha,i}^s)}
\end{eqnarray}
where $H_\els^{(s)}$ contains all terms that are not 2-local terms that involve an original qubit.
Note in the last sum, $\kappa_{\alpha,i}^s$ indexes all ancilla qubits that interacts with the original qubit $i$ with $\sigma_\alpha$ at the beginning of step $s$.
We (self-consistently) assume that the degrees of these ancilla qubits $\kappa_{\alpha,i}^s$ are at most 5 in $H_{\targ}^{(s)}$, and furthermore (denoting $M_0=nr_0$) that
\begin{equation}
\text{for } 1\le s \le S, \quad
\Delta_s \gg M_0 \Delta_{s-1}, \quad
\|H_\els^{(s)}\| \le \|H_\targ^{(s)}\| \le O(M_0 \Delta_{s-1}), \quad
\lambda_{\kappa_{\alpha,i}^s} = O(\sqrt{\Delta_{s-1}}).
\label{eq:fork-self-consistent-assumption}
\end{equation}
where $\Delta_s$ is thought of as the interaction strength of $\tilde{H}_\gadget^{(s)}$.

To simplify notation, we will denote $\kappa\equiv\kappa_{\alpha,i}^s$ from now on with the implicit understanding that the index $\kappa$ depends on the Pauli type $\alpha$, original qubit index $i$, as well as step index $s$.
Then, the \emph{fork gadget} Hamiltonian that roughly halves the Pauli degrees of original qubits is
\begin{gather}
\tilde{H}_\gadget^{(s)} = H_\anc^{(s)} + V^{(s)}, \quad
H_\anc^{(s)} = \Delta_s \sum_{i=1}^n \sum_{\alpha=x,y,z} \sum_{\kappa=1}^{\lfloor r_{\alpha,i}^{(s)} /2 \rfloor} \ketbra{1}^{()} , \quad
V^{(s)} = V_1^{(s)}+V_2^{(s)}\\
V_1^{(s)} =  H_\els^{(s)}{}' + \sum_{i,\alpha}\sum_{\kappa=1}^{\lfloor r_{\alpha,i}^{(s)} /2 \rfloor}  \left[ \lambda_{2\kappa-1}\lambda_{2\kappa} X^{(2\kappa-1)} X^{(2\kappa)} + \frac12(\lambda_{2\kappa-1}^2 + \lambda_{2\kappa}^2 + 1)\right]\\
V_2^{(s)} =  \sqrt{\frac{\Delta_s}{2}} \sum_{i,\alpha}\sum_{\kappa=1}^{\lfloor r_{\alpha,i}^{(s)} /2 \rfloor}  (\sigma_\alpha^{(i)} - \lambda_{2\kappa-1} X^{(2\kappa-1)} - \lambda_{2\kappa} X^{(2\kappa)})\otimes X^{(a_\kappa)}
\end{gather}
Here, we introduce an extra ancilla qubit $a_\kappa$ for each pair of relevant ancilla qubit $(2\kappa-1, 2\kappa)$ where the fork gadget is applied, and it has degree 3.
For every odd $r_{\alpha,i}^{(s)}$,  we add the left-over term $\lambda_{\kappa} \sigma_{\alpha}^{(i)}\otimes X^{(\kappa)}$ for $\kappa = r_{\alpha,i}^{(s)}$ to $H_\els^{(s)}$, which gives us $H_\els^{(s)}{}'$.
Each original qubit $i$ thus has its $\alpha$-Pauli-degree reduced to $r_{\alpha,i}^{(s+1)} = \lceil r_{\alpha,i}^{(s)} /2 \rceil$.
However, the pre-existing ancilla qubits $2\kappa-1$ and $2\kappa$ acquire an extra ``edge'' (interaction term) in $V_1^{(s)}$, so their degrees increase by one, but will not increase further as they are unaffected by subsequent gadget applications.
Therefore, since we assumed these qubits have degree at most 5 in $H_\targ^{(s)}$, their degree in $\tilde{H}_\gadget^{(s)}$ is at most 6.
Note when we consider $\tilde{H}_\gadget^{(s)}\equiv H_\targ^{(s+1)}$, the last two of the three self-consistent assumptions we made in Eq.~\eqref{eq:fork-self-consistent-assumption} are satisfied.

\textit{Effective simulation of $H_\targ^{(s)}$}---
Now we've written down the gadget Hamiltonian, we want to show that it reproduces the effective Hamiltonian $H_\eff^{(s)} = H_\targ^{(s)}\otimes P_\anc|_{\L_-}$ in its self-energy $\Sigma_-(z)$.
Let us denote $P_\anc^{(s)}=\bigotimes_{a_\kappa}\ketbra{0}^{(a_\kappa)}$ as the ground space projector of $H_\anc^{(s)}$ restricted to Hilbert space of newly added ancilla.
We also denote $\Pi_-^{(s)}=\Id\otimes P_\anc$, $\Pi_+^{(s)} = \Id - \Pi_-^{(s)}$ as projectors the partition the full Hilbert space into low and high energy subspace with respect to $H_\anc^{(s)}$.
Following the same convention outlined earlier in Eq.~\eqref{eq:gadget-convention-1}, we denote $V_{\pm\pm}^{(s)}$, $V_{\pm\mp}^{(s)}$, and $G_+^{(s)}(z)\equiv \Pi_+^{(s)}(z-H_\anc^{(s)})^{-1}\Pi_+^{(s)}$, etc.
We then apply the perturbation series expansion for self-energy in Eq.~\eqref{eq:gadget-convention-2} to find
\begin{equation} \label{eq:self-energy-fork}
\Sigma_-^{(s)}(z) =  H^{(s)}_{\anc,-}  + V^{(s)}_{--} +  V^{(s)}_{-+} G^{(s)}_+ V^{(s)}_{+-} + \sum_{p=1}^\infty V^{(s)}_{-+} (G^{(s)}_{+} V^{(s)}_{++})^p G^{(s)}_+ V^{(s)}_{+-}
\end{equation}
Here $H_{\anc,-}^{(s)}=0$, but
\begin{eqnarray}
&& V_{--}^{(s)}+
V_{-+}^{(s)}G^{(s)}_+ V^{(s)}_{+-} = V_1^{(s)}\otimes P_\anc^{(s)} + \frac{\Delta_s}{2(z-\Delta_s)} \sum_{i,\alpha}\sum_{\kappa=1}^{\lfloor r_{\alpha,i}^{(s)} /2 \rfloor} (\sigma_\alpha^{(i)} - \lambda_{2\kappa-1} X^{(2\kappa-1)} - \lambda_{2\kappa} X^{(2\kappa)})^2 \otimes P_\anc^{(s)} \nonumber \\
&=& H_\targ^{(s)} \otimes P_\anc^{(s)}  + \underbrace{\frac{z}{2(z-\Delta_s)} \sum_{i,\alpha}\sum_{\kappa=1}^{\lfloor r_{\alpha,i}^{(s)} /2 \rfloor} (\sigma_\alpha^{(i)} - \lambda_{2\kappa-1} X^{(2\kappa-1)} - \lambda_{2\kappa} X^{(2\kappa)})^2 \otimes P_\anc^{(s)}}_{E_1} \label{eq:second-order-fork}
\end{eqnarray}
where we used the fact that there is no ``cross-gadget'' term because $\sum_{\kappa,\kappa'} \Pi_- X^{(a_\kappa)} G_+ X^{(a_\kappa')}\Pi_- = \delta_{\kappa,\kappa'}\Pi_-/(z-\Delta_s)$.

\textit{Optimizing $\Delta_s$}---
We need to find the optimal choice of interaction strength $\Delta_s$ that allows $\tilde{H}_\gadget^{(s)}$ to gap-simulate $H_\targ^{(s)}$.
To that end, we want to show $\|\Sigma_-^{(s)}(z) - H_\eff^{(s)}\|\le \eps$, for some range of $|z| \le  z_{\max} = \|H_\targ^{(s)}\|$ and appropriately large $\Delta_s$.
Recall our self-consistent assumptions in Eq.~\eqref{eq:fork-self-consistent-assumption} that:
\begin{equation}
\Delta_s \gg M_0 \Delta_{s-1}, \quad
\|H_\els^{(s)}\| \le \|H_\targ^{(s)}\| \le O(M_0 \Delta_{s-1}), \quad
\lambda_{\kappa_{\alpha,i}^s} = O(\sqrt{\Delta_{s-1}})
\end{equation}
Note the error term $E_1$ in Eq.~\eqref{eq:second-order-fork} can be bounded by
\begin{equation}
\|E_1\| \le O(M_0\sqrt{\Delta_{s-1}} z_{\max}/\Delta_s) = O(M_0^2\Delta_{s-1}^{3/2}/\Delta_s).
\end{equation}
We also have $\|V_{-+}^{(s)}\| , \|V_{2,++}^{(s)}\| \le \|V^{(s)}\| = O(M_0\sqrt{\Delta_s\Delta_{s-1}})$, $ \|V_{1,++}^{(s)}\|  = O(M_0\Delta_{s-1})$,
and $\|G_+^{(s)}\| \le 1/(\Delta_s - z_{\max}) \le 2/\Delta_s$.
We are now ready to bound the higher order terms ($p\ge 1$) in Eq.~\eqref{eq:self-energy-fork}.
In order to obtain a better overall error bound, let us bound the $p=1$ term (third-order perturbation) in Eq.~\eqref{eq:self-energy-fork} separately.
Observe that at this order, with only three possible applications of $V$, the only possible virtual transition on the ancilla level is of the form $\ket{0\cdots0}\to\ket{0\cdots010\cdots0} \to\ket{0\cdots010\cdots0}  \to \ket{0\cdots0}$.
Therefore,
\begin{eqnarray}
V_{-+}^{(s)} G_+^{(s)} V_{++}^{(s)} G_+^{(s)} V_{+-}^{(s)} &=& V_{-+}^{(s)} G_+^{(s)} V_{1,++}^{(s)} G_+^{(s)} V_{+-}^{(s)}, \nonumber \\
\|E_2\| \equiv \|V_{-+}^{(s)} G_+^{(s)} V_{++}^{(s)} G_+^{(s)} V_{+-}^{(s)}\| &\le &  \frac{\|V_{-+}^{(s)} \|^2 \|V_{1,++}^{(s)}\|}{(\Delta_s - z_{\max})^2} = O\left(\frac{M_0^3\Delta_{s-1}^2}{\Delta_s}\right).
\end{eqnarray}
Also, $\|V_{++}^{(s)}\| \le \|V_{1,++}^{(s)}\| + \|V_{2,++}^{(s)}\|  = O(M_0\sqrt{\Delta_s\Delta_{s-1}})$.
Consequently, we can bound the remaining terms in Eq.~\eqref{eq:self-energy-fork}:
\begin{eqnarray}
\|E_3\| \equiv\left\|\sum_{p=2}^\infty V^{(s)}_{-+} (G^{(s)}_{+} V^{(s)}_{++})^p G^{(s)}_+ V^{(s)}_{+-}\right\|
&\le& \sum_{p=2}^\infty \frac{\|V^{(s)}_{-+}\|^2\|V^{(s)}_{++}\|^p}{(\Delta_s/2 )^{p+1}} = \frac{8\|V_{-+}^{(s)}\|^2 \|V_{++}^{(s)}\|^2}{\Delta_s^2(\Delta_s- 2\|V_{++}^{(s)}\|)} \nonumber \\
&\le& O\left(\frac{M_0^4 \Delta_{s-1}^2}{\Delta_s}\right).
\end{eqnarray}
Hence, for $H_\eff^{(s)} = H_\targ^{(s)}\otimes P_\anc^{(s)}$, we have
\begin{equation}
\|\Sigma_-^{(s)}(z) - H_\eff^{(s)}\|\le \|E_1\|+\|E_2\| + \|E_3\| =O\left(\frac{M_0^4 \Delta_{s-1}^2}{\Delta_s}\right)
\end{equation}
Furthermore, in order to bound the groundspace projector error to $\O(\eps)$ per Lemma~\ref{lem:gadget-ground-space}, we also need $\|V^{(s)}\|/\Delta_s \le\O(M_0\sqrt{\Delta_{s-1}/\Delta_s}) \le \eps_s$.
Hence, a sufficient choice of $\Delta_s$ that satisfy all these bounds is
\begin{equation} \label{eq:gadget-norm2}
\Delta_s = O \left(\frac{M_0^4\Delta_{s-1}^2}{\eps^2}\right).
\end{equation}

\textit{Analysis of gap-simulation}---
We now analyze the gap-simulation of $H_\targ^{(s)}$ by $\tilde{H}_\gadget^{(s)}$, for $s=1,\ldots, S$.
By Lemma~\ref{lem:gadget-eigenvalue}, one can see that corresponding eigenvalues of $\tilde{H}_\gadget^{(s)}$ and $H_\targ^{(s)}$ differ by at most $\eps$.
Let $\tilde{P}^{(0)}=P$ be the given quasi-groundspace projector of $H_\targ = H_\targ^{(1)}$ with energy spread $w$ and quasi-spectral gap $\gamma$.
We let $\tilde{P}^{(s)}$ as the quasi-groundspace projector onto $\rank(P)$ lowest eigenstates of $\tilde{H}_\gadget^{(s)}$,
with energy spread $\tilde{w}_s$ and quasi-spectral gap $\gamma_s$.
We also generalize to the case of $s=0$ by denoting $\tilde{w}_0 = w$ and $\gamma_0=\gamma$.
To ensure $\gamma_s \ge \gamma$, we simply scale $\tilde{H}_\gadget^{(s)} \mapsto c \tilde{H}_\gadget^{(s)}$ where $c = \frac{\gamma}{\gamma-2\eps}=O(1)$.
Assuming $\eps < (1-\tilde{w}_{s-1})\gamma/2$, the energy spread of $\tilde{H}_\gadget^{(s)}$ can be bounded by
\begin{equation}
\tilde{w}_s\le \frac{\tilde{w}_{s-1}\gamma + 2 \eps}{\gamma} =\tilde{w}_{s-1} + \frac{2\eps}{\gamma}.
\label{eq:fork-energy-spread}
\end{equation}
Let $E^g_0$ be the groundstate energy of $H_\targ$ and $E^g_s$ be the groundstate energy of $\tilde{H}_\gadget^{(s)}$, then Lemma~\ref{lem:gadget-eigenvalue} tells us that $|E^g_s| \le |E^g_0| + \sum_{r=1}^{s} \eps_r \le O(\|H_\targ\|) + \O(S \eps) \ll \Delta_s$.
Our choice of $\Delta_s$ satisfies $\|V^{(s)}\|/\Delta_s \le \O(\eps_s)$, we can use Lemma~\ref{lem:gadget-ground-space} with $r=(1+\tilde{w}_{s-1})\gamma/2$ and $\eps \ll (1-\tilde{w}_{s-1})\gamma/2$ to bound the error in the ground space projector by
\begin{equation} \label{eq:fork-ground-space}
\|\tilde{P}^{(s)} - \tilde{P}^{(s-1)} \otimes P_\anc^{(s)}\| \le \O\left(\frac{\eps_s}{(1-\tilde{w}_{s-1})^2\gamma}\right).
\end{equation}

\textit{Analysis of output Hamiltonian}---
After $S=O(\log r_0) = O(\log n)$ iterations of fork gadgets, we obtained the final gadget Hamiltonian $\tilde{H} = \tilde{H}_\gadget^{(S)}$.
Here, each original qubit has final Pauli-$\alpha$-degree of at most $2$.
Since the different Pauli couplings are handled independently with different groups of ancilla qubits, the maximum degree of each original qubit is $2\times3=6$.
As noted earlier, each ancilla qubit has degree at most 6 in $\tilde{H}_\gadget^{(s)}$.
Hence $\tilde{H}$ has maximum degree 6.
Moreover, since we had assumed $H_\targ$ has $O(n r_0)$ terms, so does $\tilde{H}$.
By Eq.~\eqref{eq:fork-energy-spread}, $\tilde{H}$ gap-simulates $(H_\targ, P)$ with energy spread
\begin{equation}
\tilde{w}_S \le w + \sum_{s=1}^S \frac{2\eps}{\gamma} = w + \frac{2S\eps}{\gamma} =  w + O(\frac{\eps}{\gamma} \log n)
\end{equation}
By choosing $P_\anc=\bigotimes_{s=1}^S P_\anc^{(s)}$, we can bound
the incoherence of gap-simulation by
\begin{equation}
\|\tilde{P}^{(S)} - P\otimes P_\anc\| \le \sum_{s=1}^{S} \O\left(\frac{\eps}{(1-\tilde{w}_{s-1})^2\gamma} \right) \lesssim \O(S\eps/\gamma) = O(\eps \log n/\gamma).
\end{equation}
Furthermore, solving the recursive relation in Eq.~\eqref{eq:gadget-norm2} yields the maximum required interaction strength in $\tilde{H}$:
\begin{equation}
\Delta_{S}
= O\left(\left(\frac{M_0^4\Delta_0}{\eps^2}\right)^{2^S-1} \right)
= O\left(\left(\frac{M_0^4\Delta_0}{\eps^2}\right)^{\poly(n)} \right)
 = O\left(\left(\frac{\poly(n) \Delta_0}{\eps^2}\right)^{\poly(n)}\right).
\end{equation}
Note that in this final gadget Hamiltonian, the resultant geometry is a cluster of $O(r_0)$ qubits arranged in a tree-like graph that mediate all $r_0$ interactions between the original qubit with the rest.
\end{proof}

\subsection{Proof of Theorem~\ref{thm:degree-reduction-exp}\label{sec:proof-theorem-gadget}}
Now that we have established that perturbative gadgets can be used for gap-simulation, we can use them constructively for degree-reduction of any local Hamiltonian. This is Theorem~\ref{thm:degree-reduction-exp}, which we restate here for convenience:

{
\renewcommand{\thethm}{\ref{thm:degree-reduction-exp}}
\begin{thm}[Coherent DR with exponential interaction strength]
Let $H$ be an $n$-qubit $O(1)$-local Hamiltonian with
$M_0$ terms, each with bounded norm.
Suppose $H$ has quasi-spectral gap $\gamma$ and energy spread $w$ according to Def.~\ref{defn:gap}.
For any $\epsilon>0$, one can construct a $2$-local
$[O(1), O(M_0), O ((\gamma\epsilon)^{-\poly (n)} )]$-degree-reducer of $H$ with incoherence $\epsilon$, energy spread $w+\O(\epsilon)$, and trivial encoding.
\end{thm}
\addtocounter{thm}{-1}
}

\begin{proof}
Let $H$ be the given $k$-local $n$-qubit Hamiltonian, $H$, where $k=O(1)$.
Note we can always write $H$ in the form
\begin{equation}
H = \sum_{\mu=1}^{M_0} \alpha_{\mu} \sigma_{\mu_1}^{(s_{\mu_1})}\otimes \sigma_{\mu_2}^{(s_{\mu_2})} \otimes \cdots \otimes \sigma_{\mu_k}^{(s_{\mu_k})} , \quad \sigma_{\mu_i} \in \{\Id, X, Y, Z\},
\end{equation}
where $M_0 = O(n^{k})$, and $|\alpha_\mu|=O(1)$.
We call these $n$ qubits the ``original qubits'', and they have maximum degree $d_0 = O(n^{k-1})$.
Let us denote the groundspace projector of $H$ as $P$, with spectral gap $\gamma$ and energy spread $w$.
We now want to construct a degree-reducer of $H$ using gadgets that gap-simulates $(H,P)$, and our construction proceeds in four parts:
\begin{enumerate}
\item Reduce locality to 3 by $O(\log k)$ serial applications of subdivision gadget (use Claim~\ref{claim:subdiv}).
\item Reduce locality to 2 by one application 3-to-2 local gadget in parallel (use Claim~\ref{claim:3-to-2}).
\item Isolate each original qubit by one application of subdivision gadget (use Claim~\ref{claim:subdiv}).
\item Reduce maximum degree to 6 by $O(\log n)$ serial applications of fork gadget (use Claim~\ref{claim:fork}).
\end{enumerate}

\paragraph{Part I}---
We apply the subdivision gadget $K$ times to reduce locality to 3.
At iteration $q = 1,\ldots,K$, we have
\begin{equation}
[H\equiv \tilde{H}_{\gadget,0} \equiv H_{\targ,1}] \to [\tilde{H}_{\gadget,1} \equiv H_{\targ,2}] \to \cdots \to \tilde{H}_{\gadget,K}.
\end{equation}
Let us denote the locality of $\tilde{H}_{\gadget,q}$ as $k_q$, its energy spread $\tilde{w}_q$, its incoherence relative to $H$ as $\epsilon_q$, and $\Delta_q$ as the parameter chosen as in Claim~\ref{claim:subdiv}.
We denote $k_0=k$, $\tilde{w}_0=w$, $\epsilon_0=0$, and $\Delta_0=O(1)$.
By Claim~\ref{claim:subdiv}, we have $k_q \le \lceil k_{q-1}/2\rceil +1$,
hence $K=O(\log k) = O(1)$ iterations is sufficient to reduce locality at the end to $k_K=3$.
Note at any given iteration, the number of terms whose locality need to be reduced is $m=O(M_0)$.
At iteration $q$, let us denote interaction strength in front of the $k_q$ local terms of $H_{\targ,q}$ as $J_q$, where $J_1=\max_\mu |\alpha_\mu| = O(1)$.
Then $\Delta_1 = O((M_0^2J_1(M_0^4 J_1^2 + M_0)/\eps^2)= O(M_0^6/\eps^2)$.
For $q\ge 2$, we have
\begin{equation}
J_q = O(\sqrt{J_{q-1}\Delta_{q-1}}), \quad
\Delta_q = O\left(\frac{M_0^2J_q(M_0^4 J_q^2 + M_0 \Delta_{q-1}) }{\eps^2}\right)
\end{equation}
Since $K=O(1)$, $M_0 = O(n^k) = \poly(n)$, we have
\begin{equation}
J_K, \Delta_K = O(\poly(n, \eps^{-1}))
\end{equation}
We note that after these iterative applications, we added $O(k)$ ancilla qubits for each $k$-local term.
A total of $O(k M_0)$ ancilla qubits are added since there are $O(M_0)$ $k$-local terms.
The added ancilla qubits have degree at most 2.
The original qubits still have degree $O(r_0)$, since for every $k$-local term there were a part of, they still need to interact with some other qubit to effectively generate the interactions.
By Claim~\ref{claim:subdiv}, $\tilde{H}_I := \tilde{H}_{\gadget,K}$ gap-simulates $(H,P)$ with incoherence $\epsilon_I$ and energy spread $\tilde{w}_I$, where
\begin{equation}
\begin{aligned}
\epsilon_I &= \O(K\eps/\gamma) = \O(\log k \eps/\gamma) = \O(\eps/\gamma), \\
\tilde{w}_I &= w + 2K\eps/\gamma = w + \O(\log k \eps/\gamma) = w + \O(\eps/\gamma).
\end{aligned}
\end{equation}
The interaction strength of the $(>1)$-local terms in $\tilde{H}_I \equiv$ are $\sqrt{\Delta_I} := \sqrt{J_K\Delta_K} = O(\poly(n,\eps^{-1}))$.

\paragraph{Part II}---
We apply the 3-to-2 local gadget to reduce locality of $H_{\gadget,K}$ to 2.
We note that there are $O(kM_0)=O(n^k)$ 3-local terms after the previous part,
In particular, we note that since we can apply subdivision gadget to even 3-local terms in the previous part, we can make it so that any 3-local term in $\tilde{H}_{\gadget, K}$ contains at least one ancilla qubits, where the 3-local term acts on it with $X$.
Hence, we can apply Claim~\ref{claim:3-to-2} for every 3-local term $\mu$ simultaneously, while choosing $C_\mu=X$.
The parameters in the premise of Claim~\ref{claim:3-to-2} in this context are $m=O(n^k)$ and $\Delta_0 =\Delta_I = O(\poly(n,\eps^{-1})$.
This allows us to generate a 2-local Hamiltonian $\tilde{H}_{II}$ with interaction strength $\Delta_{II} = O(m^{12}\Delta_0^3/\eps^3) = O(\poly(n,\eps^{-1}))$.
$\tilde{H}_{II}$ gap-simulate $(H,P)$ with incoherence $\epsilon_{II}$ and energy spread $\tilde{w}_{II}$, where
\begin{equation}
\eps_{II} \le \epsilon_{I} + \O(\eps/\gamma) = \O(\eps/\gamma),\quad
\tilde{w}_{II} \le \tilde{w}_{I} + 2\eps/\gamma = w + \O(\eps/\gamma).
\end{equation}
Importantly, in this construction, all the original qubits will interact with ancilla only in the form of $\sigma_\alpha^{(i)}\otimes X^{(a_i)}$, where $i$ is an original qubit and $a_i$ is some ancilla qubit.
(They still might interact with other original qubits with some arbitrary 2-local term.)
The ancilla qubits will have maximum degree 4 (they had maximum degree 2 in the 3-local Hamiltonian of Part I, and the 3-to-2-local gadget Hamiltonian results in a maximum of degree 2 per conversion of a 3-local term).
This is useful to keep in mind because it satisfies the assumptions in Claim~\ref{claim:fork}.

\paragraph{Part III}---
Now we want to isolate original qubits from each other with the subdivision gadget, so that each original qubit $i$ only interacts with some set of ancilla qubit $\{a_i\}$ in the form of $\sigma_\alpha^{(i)}\otimes X^{(a_i)}$.
The idea is to that we can write
\begin{equation}
H_{\targ,III} = \tilde{H}_{II} = H_{\els,III} + \sum_{\nu} c_\nu \sigma_{\nu_1} ^{(i_1)} \otimes \sigma_{\nu_2}^{(i_2)}
\end{equation}
where $(i_1, i_2)$ are pairs of original qubits, $|c_\nu|\le \Delta_{II}^{2/3} = O(\poly(n,\eps^{-1}))$, and $H_{\els,III}$ contains all other terms (i.e. 1-local terms, 2-local terms interacting ancilla with original through $\sigma\otimes X$, or ancilla with ancilla).
We note that $\|H_{\els,III}\| \le O(M_0\Delta_{II})$.
Again, since we have not reduced the degree of the original qubit, there are $m=O(M_0)$ interactions between the original qubits that we need to address.
Thus, we can use the following gadget Hamiltonian
\begin{gather}
\begin{aligned}
&\tilde{H}_{III} = H_{\anc,III} + V_{III} , \quad
H_{\anc,III} = \Delta_{III} \sum_\nu \ketbra{1}^{(a_{\nu})},  \\
&V_{III} = H_{\els,III} +  \sum_\nu \left[ \sqrt{\frac{|c_\nu| \Delta_{III}}{2}} ( \sgn(c_\nu)  \sigma_{\nu_1}^{(i_1)} - \sigma_{\nu_2}^{(i_2)} ) \otimes X^{(a_{\nu})} + |c_\nu| \right], \label{eq:VIII}
\end{aligned}
\end{gather}
so that the original qubit $i_1$ and $i_2$ no longer interact directly, but now interacts with a new ancilla qubit $a_\nu$.
By Claim~\ref{claim:subdiv}, it is sufficient to choose $\Delta_{III} = \poly(n,\eps^{-1})$ to ensure $\tilde{H}_{III}$ gap-simulates $(H,P)$ with incoherence $\epsilon_{III}$ and energy spread $\tilde{w}_{III}$, where
\begin{equation}
\eps_{III} \le \epsilon_{II} + \O(\eps/\gamma) = \O(\eps/\gamma),\quad
\tilde{w}_{III} \le \tilde{w}_{II} + 2\eps/\gamma = w +\O(\eps/\gamma).
\end{equation}
Note the added ancilla qubits have degree at most 2, while the ancilla qubits in $H_{\els, III}$ have degree at most 3.

\paragraph{Part IV}---
Note the Hamiltonian $\tilde{H}_{III}$ from the previous part satisfy the assumptions for applying Claim~\ref{claim:fork}.
This is because the original qubits only interact with ancilla qubits through the form $\sigma\otimes X$ in $V_{III}$ as seen in Eq.~\eqref{eq:VIII}, and implicitly so in $H_{\els,III}$ due to our construction in Part II.
Note the maximum Pauli degree is $r_0 \le d_0=O(n^k)$.
By choosing $\Delta_0 \equiv \Delta_{0,IV} = O(|c_\nu|\Delta_{III}) = O(\Delta_{II}^{2/3}\Delta_{III}) = \poly(n,\eps^{-1})$, the rest of the assumptions in Claim~\ref{claim:fork} are satisfied.
Therefore, by Claim~\ref{claim:fork}, there is a Hamiltonian $\tilde{H}_{IV}$ that gap-simulate $(H,P)$ with incoherence $\epsilon_{IV}$ and energy spread $\tilde{w}_{IV}$ where
\begin{equation}
\epsilon_{IV} \le \epsilon_{III} + O(\frac{\eps}{\gamma} \log n) = O(\frac{\eps}{\gamma}\log n), \quad
\tilde{w}_{IV} \le \tilde{w}_{III} + O(\frac{\eps}{\gamma} \log n) = w + O(\frac{\eps}{\gamma}\log n).
\end{equation}
To ensure that $\eps_{IV} \le \epsilon$ for some constant $\epsilon$, we need to have chosen
\begin{equation}
\eps \le O(\frac{\gamma \epsilon}{\log n})
\quad \Longrightarrow \quad \epsilon_{IV} \le \epsilon \quad \text{and} \quad \tilde{w}_{IV} \le w + \O(\epsilon).
\end{equation}
By Claim~\ref{claim:fork}, this gap-simulating Hamiltonian has maximum degree $6$, $O(nr_0)=O(n^k)=O(M_0)$ terms, each with norm (interaction strength) bounded by
\begin{equation}
J = O\left(\left(\frac{\poly(n) \Delta_{0,IV}}{\eps^2}\right)^{\poly(n)}\right) = O\left(\left(\poly(n,\eps^{-1})\right)^{\poly(n)}\right) =  O\left(\left(\poly(n)/\gamma\epsilon)\right)^{\poly(n)}\right).
\end{equation}
This concludes our construction.
\end{proof}

\section{Connection to Quantum PCP\label{sec:qPCP}}

In this Appendix, we draw a connection between our idea of gap-simulation to reductions of Hamiltonian in the context on quantum PCP.
As explained in Sec.~\ref{sec:other-results},
there is a very important distinction between gap-simulating degree-reductions
and degree-reductions used in the PCP
context. In a nut-shell, this distinction
boils down to the difference between
spectral gap and promise gap.
Nevertheless, there is a meaningful setting
in which this difference can be bridged.

We note that classical PCP reduction algorithms are usually
{\it constructive}: they map not only CSPs to CSPs
but also assignments to assignments\cite{dinurgoldreich,BenSassonPCP}.
Moreover, if the CSP is satisfiable,
then any of its non-satisfying assignments is mapped
to one with at least as many violations.
In other words, they preserve the {\it properties of the
assignment}.

In Sec.~\ref{sec:qPCP-implication} below, we suggest a definition of qPCP
reductions which extends this notion
to the quantum world.
Specifically, we requires that the reduction preserve groundstate
properties in the similar sense as gap-simulations do, and
maps excited states of $H$ to high-energy states of $\tilde{H}$.
Very importantly,
we do \emph{not} require preservation of the spectral gap, as this is
the essence of the difference between the qPCP and gap-simulating settings.
We discuss the connection between this definition and gap-simulation in Sec.~\ref{sec:qPCP-gap-simulation}.

With these restrictions, it is possible to connect
the worlds of spectral gap and promise gap.
In Sec.~\ref{sec:imposs-qPCP-degree}, we prove Theorem~\ref{thm:imposs-qPCP} of that shows the impossibility of qPCP-DR and qPCP-dilution with close-to-perfect coherence, based on similar ideas of those in Lemma~\ref{lem:imposs1-DR} and Theorem~\ref{thm:imposs1-dilute}.
Unfortunately, these
impossibility results hold only
for inverse polynomial incoherence.
One would hope to improve these results to {\it constant}
incoherence, similar to Theorem \ref{thm:main}.
Alas, this remains open;
See Sec.~\ref{sec:proof-sketch-main} for discussion about the difficulty in strengthening
to constant $\epsilon$, which is the
relevant regime in the context of PCPs.

\subsection{Definitions of Quantum PCP Reductions\label{sec:qPCP-implication}}

To derive implications of our gap-simulation framework to quantum PCP,
we need to define quantum PCP reductions,
and then restrict those in a way
that will enable bridging the differences
between spectral-gap and promise-gap worlds.

We start by defining general quantum
PCP reductions.
In gap-simulation, one considers reductions of Hamiltonians that preserve all groundstate properties {\it including their
spectral gaps}.
In the context of quantum NP, it is the so-called
\emph{promise gap}, and not the spectral gap, that must be controlled through the reductions.
Formally:

\begin{defn}[$\mu$-qPCP-reduction]
\label{defn:qPCP-reduction0}
An algorithm $\A$ is a $\mu$-\emph{qPCP-reduction} if it takes as its input any $n$-qubit, $O(1)$-local, positive semi-definite Hamiltonian $H$, and two numbers $a,b$ where $0\le a<2^{-\Omega(n)}$, the \emph{promise gap}
$b-a>1/\poly(n)$, and its output $\tilde{H}=\A(H,a,b)$ is an $O(1)$-local, positive semi-definite Hamiltonian on $n+\poly(n)$ qubits such that
(1) if $\lambda_1(H) \le a$, then $\lambda_1(\tilde{H}) \le \mu a$, and
(2) if $\lambda_1(H) \ge b$, then $\lambda_1(\tilde{H}) \ge \mu b$.
\end{defn}

The definition of classical PCP reduction can easily be deduced.
Note that $\A$ above is not required to preserve any properties of the groundspace.
Thus,
as explained in the introduction, there is no hope to prove information-theoretical
impossibility in this setting.
Hence, we consider the following the restrictions on the quantum PCP reductions,
which we believe extend the properties of constructive PCP reductions in the classical context to the quantum world:

\begin{defn}[$(\mu,\delta, \epsilon)$-qPCP*-reduction]
\label{defn:qPCP-reduction}
A $\mu$-qPCP-reduction $\A$ is a $(\mu,\delta, \epsilon)$-\emph{qPCP*-reduction} with encoding $V$ ($V$ is an isometry) if there exists a projector $P_\anc$ on some ancilla so that:
\begin{enumerate}
\vspace{-5pt}
\item Let $P$ be the projector onto eigenstates of $H$ with eigenvalue $\le a$.
Let $\tilde{P}$ be the projector onto eigenstates of $\tilde{H}=\A(H,a,b)$ with eigenvalue $\le \mu a$.
Then they satisfy
$\|\tilde{P} - V (P\otimes \Id_\anc)V^\dag \tilde{P}\| \le \delta$,
and
$\|\tilde{P} - V({P}\otimes P_\anc)V^\dag \| \le \epsilon$.
We call $\delta$ \emph{unfaithfulness} and $\epsilon$ \emph{incoherence}, same as in Definitions~\ref{defn:hamsimul} and \ref{defn:hamsimul-incoherent}.
\item If $\ket{\psi}$ is an eigenstate of $H$ with eigenvalue $\ge b$, then $\forall \ket\alpha\in P_\anc$, let $\ket{\bar{\psi}} = V \ket{\psi}\ket{\alpha}$ which must satisfy $\braket{\bar{\psi}|\tilde{H}|\bar{\psi}} \ge \mu b$.
\end{enumerate}
\end{defn}

Essentially, our definition of the analogously ``constructive'' qPCP reduction requires that the eigenstates of $H$ serving as satisfying assignments to be faithfully and/or coherently mapped to low-energy eigenstates of the output Hamiltonian $\tilde{H}$.
Furthermore, we require any high-energy eigenstates (violating assignments) of $H$ with eigenvalues above the promise gap to also be mapped mapped to high-energy states in $\tilde{H}$.

As a sanity check, note that a $\mu$-qPCP-reduction with the additional condition that its output $\tilde{H}=\A(H,a=0,b)$ gap-simulates $H$ with $\delta$-unfaithfulness, $\epsilon$-incoherence, and energy spread $\tilde{w}=0$ is a $(\mu,\delta,\epsilon)$-qPCP*-reduction, if the output promise gap is not larger than the spectral gap (see Lemma~\ref{lem:connecting-qPCP-gap}
in Sec.~\ref{sec:qPCP-gap-simulation}).
We now aim to study whether we can rule out
DR and dilution in the qPCP context.
We first define these notions:

\begin{defn}[qPCP-DR and qPCP-dilution]
Consider any $n$-qubit input Hamiltonians of the form $H=\sum_{i=1}^{M_0} H_i$, which is a sum of $M_0=M_0(n)$ terms, each of which is $O(1)$-local.
\begin{itemize}
\vspace{-5pt}
\item A $(\mu,\delta, \epsilon)$-qPCP*-reduction $\A$ is a $(\mu,\delta, \epsilon)$-\emph{qPCP-DR} if $\A(H,a,b)$ has O(1) degree.
\vspace{-5pt}
\item A $(\mu,\delta, \epsilon)$-qPCP*-reduction $\A$ is a $(\mu,\delta, \epsilon)$-\emph{qPCP-dilution} if $\A(H,a,b)$ has $o(M_0(n))$ local terms.
\vspace{-5pt}
\end{itemize}
\end{defn}

We note that the construction used in Proposition~\ref{prop:classical-deg-reduct}, which is based on DR for classical PCP\cite{dinur}, directly
implies a $(1,0,1)$-qPCP-DR by the above definition, for all classical Hamiltonians.

\subsection{Impossibility Result on qPCP-DR and qPCP-dilution (Theorem~\ref{thm:imposs-qPCP})\label{sec:imposs-qPCP-degree}}

Recall our example family of 2-local $n$-qubit Hamiltonian that was previously used:
\begin{equation}
H_\oneone  = \left(\J_z+\frac{n}{2} \right)\left(\J_z+\frac{n}{2}-1\right),
\end{equation}
whose $n+1$ groundstates are
\begin{equation}
\ket{00\cdots00}, \ket{00\cdots01}, \ket{00\cdots10}, \ldots, \ket{10\cdots00}.
\end{equation}

Using the above Hamiltonian, we show
that {\it generic} quantum $(\mu,\delta,\epsilon)$-qPCP-DR
and $(\mu,\delta,\epsilon)$-qPCP-dilution
are impossible assuming the encoding $V$ is unitary and localized.
Unfortunately, we are only
able to show this
for polynomially small $\epsilon$:

\begin{thm}[Limitation on qPCP-DR and qPCP-dilution]
\label{thm:imposs-qPCP}
For any localized encoding $V$, it is impossible to have a $(\mu,\delta,\epsilon)$-qPCP-DR or $(\mu,\delta,\epsilon)$-qPCP-dilution algorithm with encoding $V$ that works on the $n$-qubit Hamiltonian $H_\oneone$ and outputs $\tilde{H}_A$ with $\|\tilde{H}_\oneone\|=O(n^p)$, and incoherence $\epsilon(n) \le o(\sqrt{\mu(n)/n^{p}})$.
\end{thm}
\vspace{-5pt}

In particular, if we restrict the output Hamiltonians to have
degree $d$ on $n+m$ number of qubits, with $O(1)$-norm terms and encoded by some localized $V$, then no $(\mu,\delta,\epsilon)$-qPCP-DR or $(\mu,\delta,\epsilon)$-qPCP-dilution algorithm exists that works on $H_\oneone$ with $\mu=\Omega(1)$ and $\epsilon \le o (1/\sqrt{d(n+m)})$.

We prove this Theorem using the same idea as in Lemma~\ref{lem:TotalCoherentImpossible} (that was used to prove Lemma~\ref{lem:imposs1-DR} and Theorem~\ref{thm:imposs1-dilute}), which in the current context yields:

\begin{lemma}
\label{lem:imposs-qPCP}
Suppose there exists a $(\mu,\delta,\epsilon)$-qPCP*-reduction algorithm $\A$ with trivial encoding $V=\Id$.
Let $\tilde{H}_\oneone = \A(H_\oneone,a,b)$ be its output, with $0\le a < b \le 1$.
If either (1) $\epsilon=0$  and $b>2a$, or (2) $\|\tilde{H}_\oneone\| < \frac{\mu (b - 2a - 4\epsilon)}{2\epsilon^2}$, then for every pair of qubits $(i,j)$, $\tilde{H}_\oneone$ must contain a term that acts nontrivially on both qubits.
\end{lemma}

\begin{proof}
For the sake of contradiction, suppose $\tilde{H}_\oneone$ contains no term that acts nontrivially on both qubit $i$ and $j$.
This means we can decompose $\tilde{H}_\oneone$ into two parts: $\tilde{H}_\oneone=\tilde{H}_{\oneone,i}+\tilde{H}_{\oneone,j}$, where $\tilde{H}_{\oneone,i}$ acts trivially on qubit $i$.
In other words, $[\tilde{H}_{\oneone,i},\sigma_i]=0$ for any Pauli operator $\sigma_i$ on qubit $i$.

Let us denote states $\ket{g_0}=\ket{0\cdots0}$ and $\ket{g_i}=X_i\ket{g_0}=\ket{0\cdots01_i0\cdots0}$.
Let us denote $P=\sum_{i=0}^n \ketbra{g_i}$, which is a projector onto eigenstates of $H_\oneone$ with eigenvalue $\le a $.
We also denote $\tilde{P}$ as the projector onto eigenstates of $\tilde{H}_\oneone$ with eigenvalue $\le \mu a$.
By Definition~\ref{defn:qPCP-reduction} there exists $P_\anc$, such that $\|\tilde{P} - P\otimes P_\anc\| \le \epsilon$.
Let $\ket{\alpha}\in P_\anc$, and denote $\ket{\bar{g}_i} \equiv \ket{g_i}\ket{\alpha}$ for $0\le i \le n$.
Observe that
\begin{equation}
\tilde{P} \ket{\bar{g}_i}= (\tilde{P} - P \otimes P_\anc + P \otimes P_\anc)\ket{g_i}\ket{\alpha} =  (\tilde{P} - P \otimes P_\anc)\ket{g_i}\ket{\alpha} + \ket{g_i}\ket{\alpha} = \ket{\epsilon_i} + \ket{\bar{g}_i}
\end{equation}
where $\ket{\epsilon_i} = (\tilde{P}-P\otimes P_\anc)\ket{\bar{g}_i}$ satisfying $\|\ket{\epsilon_i}\|\le \epsilon$.
Thus,
\begin{eqnarray}
\braket{\bar{g}_i | \tilde{H}_\oneone |\bar{g}_i} &=& (\bra{\bar{g}_i} \tilde{P} - \bra{\epsilon_i}) \tilde{H}_\oneone (\tilde{P}\ket{\bar{g}_i}  - \ket{\epsilon_i}) \nonumber\\
&=& \braket{\bar{g}_i |\tilde{P}\tilde{H}_\oneone \tilde{P}|\bar{g}_i} + \braket{\epsilon_i |\tilde{H}_\oneone|\epsilon_i} - 2\Re \braket{\bar{g}_i |\tilde{P}\tilde{H}_\oneone |\epsilon_i} \nonumber \\
&\le& \braket{\epsilon_i |\tilde{H}_\oneone|\epsilon_i}  + \mu a(1+2\epsilon) \le \epsilon^2 \|\tilde{H}_\oneone\| + \mu a(1+2\epsilon),
\end{eqnarray}
where we used the fact that $\|\tilde{P}\tilde{H}_\oneone \| =\|\tilde{P} \tilde{H}_\oneone \tilde{P}\| \le \mu a$.

Now consider the eigenstate $\ket{e_{ij}}=X_iX_j\ket{g_0}$ of $H_\oneone$ with eigenvalue $1 \ge b$, and let $\ket{\bar{e}_{ij}} = \ket{e_{ij}}\ket{\alpha}$.
By the second condition in Definition~\ref{defn:qPCP-reduction} of a $(\mu,\delta,\epsilon)$-qPCP*-reduction algorithm, we must have
\begin{equation}
\braket{\bar{e}_{ij}|\tilde{H}_\oneone|\bar{e}_{ij}} \ge \mu b
\end{equation}
In addition, observe that
\begin{eqnarray}
\braket{\bar{e}_{ij}|\tilde{H}_\oneone|\bar{e}_{ij}} &=& \braket{\bar{g}_0|X_iX_j (\tilde{H}_{\oneone,i}+\tilde{H}_{\oneone,j})X_iX_j|\bar{g}_0} = \braket{\bar{g}_0|X_i\tilde{H}_{\oneone,j}X_i|\bar{g}_0} + \braket{\bar{g}_0|X_j\tilde{H}_{\oneone,i}X_j|\bar{g}_0} \nonumber \\
&=& \braket{\bar{g}_i|\tilde{H}_{\oneone,j}|\bar{g}_i} + \braket{\bar{g}_j|\tilde{H}_{\oneone,i}|\bar{g}_j} \nonumber\\
&=& \braket{\bar{g}_i|\tilde{H}_{\oneone}|\bar{g}_i} + \braket{\bar{g}_j|\tilde{H}_{\oneone}|\bar{g}_j} - \braket{\bar{g}_i|\tilde{H}_{\oneone,i}|\bar{g}_i} - \braket{\bar{g}_j|\tilde{H}_{\oneone,j}|\bar{g}_j} \nonumber \\
&=& \braket{\bar{g}_i|\tilde{H}_{\oneone}|\bar{g}_i} + \braket{\bar{g}_j|\tilde{H}_{\oneone}|\bar{g}_j} - \braket{\bar{g}_0|X_i\tilde{H}_{\oneone,i}X_i|\bar{g}_0} - \braket{\bar{g}_0|X_j\tilde{H}_{\oneone,j}X_j|\bar{g}_0} \nonumber \\
&=& \braket{\bar{g}_i|\tilde{H}_{\oneone}|\bar{g}_i} + \braket{\bar{g}_j|\tilde{H}_{\oneone}|\bar{g}_j} - \braket{\bar{g}_0|\tilde{H}_{\oneone}|\bar{g}_0} \le 2 \epsilon^2 \|\tilde{H}_\oneone\| + 2\mu a(1+2\epsilon).
\end{eqnarray}
where we used the fact that $\braket{\bar{g}_0|\tilde{H}_{\oneone,j}|\bar{g}_0}\ge 0$ (because $\tilde{H}_\oneone$, as an output of a qPCP*-reduction, is positive semi-definite).
This contradicts the previous equation whenever
\begin{equation}
2\epsilon^2 \|\tilde{H}_\oneone\| + 2\mu a(1+2\epsilon) < \mu b
\quad \Longleftrightarrow \quad
\begin{dcases}
b>2a, & \text{if } \epsilon = 0 \\
\|\tilde{H}_\oneone\| < \frac{\mu (b - 2a - 4\epsilon)}{2\epsilon^2},
& \text{if } \epsilon > 0
\end{dcases}
\label{eq:HA-norm-small}
\end{equation}
Hence, if we assume (1) $\epsilon=0$ and $b>2a$, or (2) $\|\tilde{H}_\oneone\| < [\mu (b - 2a - 4\epsilon)]/(2\epsilon^2)$, then $\tilde{H}_\oneone$ must contain a term that acts nontrivially on both qubit $i$ and $j$.
\end{proof}

\begin{proof}[\textbf{Proof of Theorem~\ref{thm:imposs-qPCP}}]
We know $\lambda_1(H_\oneone) = 0$.
Suppose $\A$ is a $(\mu,\delta,\epsilon)$-qPCP-DR or $(\mu,\delta,\epsilon)$-qPCP-dilution algorithm with some localized encoding $V$.
Since $V$ is an encoding supplied by $\A$, we can define $\A'(H,a,b)=V^\dag \A(H,a,b) V$, which effectively has trivial encoding.
We run the qPCP*-reductions with $a=0$ $b=1$, obtaining $\tilde{H}_\oneone = \A(H_\oneone,0,1)$ and $\tilde{H}_\oneone'=\A'(H_\oneone,0,1) = V^\dag \tilde{H}_\oneone V$.
However, if we require the norm of $\|\tilde{H}_\oneone\| =\|\tilde{H}_\oneone'\| = O(n^p)$ and $\epsilon(n)\le o(\sqrt{\mu(n)/n^p})$, where $\mu(n)$ and $\epsilon(n)$ are parameters of $\A$ that we change with the system size $n$, then we have $ O(n^p) = \|\tilde{H}_\oneone'\| < [\mu(1-4\epsilon)]/(2\epsilon^2) = \omega(n^p)$ for sufficiently large $n$.
Then by Lemma~\ref{lem:imposs-qPCP},
for sufficiently large $n$,
each qubit in $\tilde{H}_\oneone'$ must interact with at least $n-1$ other qubits in $\tilde{H}_\oneone'$.
Since if $\tilde{H}_\oneone$ is assumed to be $O(1)$-local, then so is $\tilde{H}_\oneone'$ when $V$ is a localized encoding.
Thus, this means that $\tilde{H}_\oneone'$ must have degree $r'=\Omega(n)$, and $M'=\Omega(n^2)$ terms to cover all $\binom{n}{2}$ required pairwise interactions between qubits.
As $V$ maps local terms in $\tilde{H}_\oneone$ to local terms in $\tilde{H}_\oneone'$, their maximum degree and number of terms is related by a constant factor.
Hence, $\tilde{H}_\oneone$ also must have degree $r=\Theta(r')=\Omega(n)$, and $M=M'=\Omega(n^2)$ local terms.
Therefore, $\tilde{H}_\oneone$ cannot be the output of a qPCP-DR or a qPCP-dilution algorithm with localized encoding,
given the assumptions of polynomial bound on $\|\tilde{H}_\oneone\|$ and inverse polynomially small incoherence $\epsilon$.\end{proof}

\subsection{Relationship of Gap-Simulation to qPCP*-reduction\label{sec:qPCP-gap-simulation}}

\begin{lemma}
\label{lem:connecting-qPCP-gap}
Fix $a=0$. Suppose $\A$ is a $\mu$-qPCP-reduction such that for all input $H$, its output $\tilde{H}=\A(H,a=0,b)$ also gap-simulates $H$ with encoding $V$, unfaithfulness $\delta$, incoherence $\epsilon$, and energy spread $\tilde{w}=0$.
Then $\A$ is also a $(\mu,\delta,\epsilon)$-qPCP*-reduction with encoding $V$, as long as the output promise gap $\mu b \le \gamma(1-\delta^2)$, where $\gamma$ is the spectral gap of $H$.
\end{lemma}
\begin{proof}
Let us consider the two cases, where $\lambda_1(H) = 0$ or $\lambda_1(H) \ge b$.

To clarify notations, let us denote $P_q$ as the projector onto eigenstates of $H$ with eigenvalue $0$. and $\tilde{P}_q$ as the projector onto eigenstates of $\tilde{H}$ with eigenvalue $0$, both of which may be empty.
These are not to be confused with (quasi-)groundspace projector $P$ and $\tilde{P}$ in the gap-simulation context, although we'll show that under certain conditions, they coincide.

\textbf{Case 1: $\lambda_1(H) \ge b$. }---
Because $\lambda_1(H) \ge b \neq 0$, then there is no eigenstate of $H$ with eigenvalue $\le 0$, so the projector $P_q=0$.
Also since $\lambda_1(H) \ge b \Longrightarrow \lambda_1(\tilde{H})\ge \mu b$ according to Def.~\ref{defn:qPCP-reduction0},
there is also no eigenstate of $\tilde{H}$ with eigenvalue $\le \mu\times 0 =0$, so $\tilde{P}_q=0$.
Trivially, $\|\tilde{P}_q - V P_q V^\dag \tilde{P}_q\| = 0\le \delta$, and $\|\tilde{P}_q - V (P_q\otimes P_\anc)V^\dag \|=0\le \epsilon$.
Thus, condition 1 of Def.~\ref{defn:qPCP-reduction} is satisfied.
Condition 2 is trivially satisfied because $\lambda_1(\tilde{H}) \ge \mu b$.

\textbf{Case 2: $\lambda_1(H) = 0$. }---
Let us denote $P$ the groundspace projector of $H$ with energy $0$, and its spectral gap be $\gamma$.
Since $\lambda_1(H) = 0 \Longrightarrow \lambda_1(\tilde{H})=0$ by Def.~\ref{defn:qPCP-reduction0}, let us denote $\tilde{P}$ as the groundspace projector of $\tilde{H}$ with energy $0$.
Note $P_q=P$ and $\tilde{P}_q=\tilde{P}$.
Since by definition of gap-simulation, we have $\|\tilde{P} - VPV^\dag \tilde{P}\|\le \delta$ and $\|\tilde{P} - V(P\otimes P_\anc) V^\dag \| \le \epsilon$ for some $P_\anc$, then condition 1 of Def.~\ref{defn:qPCP-reduction} is satisfied

Let us now check if condition 2 of Def.~\ref{defn:qPCP-reduction} is satisfied.
Consider any $\ket{\alpha}\in P_\anc$, and $\ket{g^\perp}$ be any eigenstate of $H$ with eigenvalue $\ge b$ (which satisfies $P\ket{g^\perp}=0$).
Let $\ket{\bar\psi}= V \ket{g^\perp}\ket{\alpha}$, we have
\begin{equation}
\tilde{P}^\perp \ket{\bar{\psi}} = \ket{\bar{\psi}} - \tilde{P}\ket{\bar{\psi}} = \ket{\bar{\psi}} - \tilde{P} V\ket{g^\perp}\ket{\alpha} = \ket{\bar{\psi}} - (\tilde{P}-\tilde{P}V P V^\dag) V\ket{g^\perp}\ket{\alpha} = \ket{\bar{\psi}} - \ket{\delta},
\end{equation}
where we denoted $\ket{\delta}\equiv(\tilde{P}-\tilde{P}VPV^\dag )\ket{\bar{\psi}}$ satisfying $\|\ket{\delta}\|\le \delta$.
Also observe that $\tilde{P}^\perp \ket{\delta}=0$, which means $\tilde{P}^\perp\ket{\bar\psi}$ and $\ket{\delta}$ are orthogonal.
Then using $1=\|\ket{\bar\psi}\|^2 = \|\tilde{P}^\perp\ket{\bar\psi}\|^2 + \|\ket{\delta}\|^2$, we have $\|\tilde{P}^\perp\ket{\bar\psi}\| \ge \sqrt{1-\delta^2}$.
Furthermore, note that $[\tilde{P},\tilde{H}]=0$ implies $[\tilde{P}^\perp, \tilde{H}]=0$.
Hence, the energy of $\ket{\bar\psi}$ with respect to $\tilde{H}$ can be lower-bounded:
\begin{eqnarray}
\braket{\bar{\psi} |\tilde{H} |\bar{\psi}} &=& \braket{\bar{\psi}|\tilde{P}^\perp \tilde{H}\tilde{P}^\perp|\bar{\psi}} + 2 \Re\braket{\delta| \tilde{H}\tilde{P}^\perp|\bar{\psi}} + \braket{\delta|\tilde{H}|\delta} \nonumber\\
&=& \braket{\bar{\psi}|\tilde{P}^\perp \tilde{H}\tilde{P}^\perp|\bar{\psi}} + 2 \Re\braket{\delta| \tilde{P}^\perp \tilde{H}\tilde{P}^\perp|\bar{\psi}} + \braket{\delta|\tilde{H}|\delta} = \braket{\bar{\psi}|\tilde{P}^\perp \tilde{H}\tilde{P}^\perp|\bar{\psi}}+  \braket{\delta|\tilde{H}|\delta} \nonumber \\
&\ge& \braket{\bar{\psi}|\tilde{P}^\perp \tilde{H}\tilde{P}^\perp|\bar{\psi}} = \braket{\bar{\psi}|\tilde{P}^\perp(\tilde{P}^\perp \tilde{H} \tilde{P}^\perp+\gamma\tilde{P})\tilde{P}^\perp|\bar{\psi}},
\end{eqnarray}
where in the last step we added a term $\gamma \tilde{P}$ (which would evaluate to zero) inside the parenthesis to utilize the fact that the operator $\tilde{P}^\perp(\tilde{P}^\perp \tilde{H} \tilde{P}^\perp+\gamma\tilde{P})\tilde{P}^\perp$ has minimum eigenvalue $\gamma$, by Def.~\ref{defn:hamsimul} of gap-simulation.
Normalizing $\tilde{P}^\perp \ket{\bar\psi}$ by its norm, we can apply the lower bound of eigenvalue to get a lower bound of $\braket{\bar{\psi}|\tilde{H}|\bar{\psi}}$:
\begin{eqnarray}
\braket{\bar{\psi} |\tilde{H} |\bar{\psi}} \ge \frac{\braket{\bar{\psi}|\tilde{P}^\perp(\tilde{P}^\perp \tilde{H} \tilde{P}^\perp+\gamma\tilde{P})\tilde{P}^\perp|\bar{\psi}}}{\left\|\ket{\tilde{P}^\perp|\bar{\psi}}\right\|^2 } \left\|\ket{\tilde{P}^\perp|\bar{\psi}}\right\|^2  \ge \gamma (1-\delta^2).
\end{eqnarray}
We satisfy condition 2 of Def.~\ref{defn:qPCP-reduction} as long as $\gamma(1-\delta^2)\ge \mu b$.
This completes the proof of the Lemma.
\end{proof}

\section{Weak Gap-Simulation and Coherent Weak Dilution of $H_\oneone$ with Constant Interaction Strength\label{sec:weak-sparsifier}}
In this Appendix we introduce the definition of weak gap-simulation, which is an even weaker version of Hamiltonian simulation than our
gap-simulation.
In certain instances, we find that it can be helpful to allow the
Hamiltonian $\tilde{H}$ to simulate the original $H$
not in its ``groundspace" by in an excited space.
In particular, we can consider allowing $\tilde{P}$ to project onto a subspace that is not necessarily the lowest-energy, but rather isolated in the spectrum by a spectral gap of $\gamma$ from both above and below.
This may have physical applications, for example, in Floquet Hamiltonian engineering\cite{floquet,FloquetEngineering, LindnerFloquetTopo, ChoiDynamicalEngineering} where the system is driven time-dependently and periodically, and thus eigenvalues are only well-defined up to a period.
Hence, this motivates a definition of a weaker
version of gap-simulation that we provide below:

\begin{defn}[weak gap-simulation of Hamiltonian]
\label{defn:weaksimul}
Let $H$ and $\tilde{H}$ be two Hamiltonians, defined on Hilbert spaces $\H$ and $\tilde{\H}$ respectively, where $\tilde{\H}$.
Let $V: \H\otimes \H_\anc  \to \tilde{\H}$ be an isometry ($V^\dag V=\Id$), where $\H_\anc$ is some ancilla Hilbert space.
Denote $\tilde{E}^g \equiv \lambda_1(\tilde{H})$.
Per Definition~\ref{defn:gap}, let $P$ be a quasi-groundspace projector of $H$, $\gamma$ its quasi-spectral gap.
We say that $\tilde{H}$ \emph{weakly gap-simulates}
$(H,P)$ with \emph{encoding} $V$, \emph{incoherence} $\epsilon\ge 0$ and \emph{energy spread} $0\le\tilde{w}<1$ if the following conditions are both satisfied:
\begin{enumerate}
\item There exists a Hermitian projector $\tilde{P}$ projecting onto a
subspace of eigenstates of $\tilde{H}$ such that
    \begin{gather}
    [\tilde{H}, \tilde{P}]= 0, \quad \|\tilde{P}(\tilde{H} -\tilde{E}^g)\tilde{P}\|\le \tilde{w}\gamma, \quad \textnormal{and} \quad
    \left|\lambda_j(\tilde{P}^\perp (\tilde{H} - \tilde{E}^g )\tilde{P}^\perp + \gamma \tilde{P})\right| \ge \gamma \quad \forall j.
    \label{eq:weaksimul}
    \end{gather}
	I.e., $\tilde{P}$ projects onto a
 quasi-groundspace of $\tilde{H}$ with quasi-spectral gap not smaller
than that of $P$ in $H$, and energy spread $\tilde{w}$.
\item There exists a Hermitian projector $P_\anc$ acting on $\H_\anc$, so that
\end{enumerate}
\begin{flalign}
\textnormal{[bounded incoherence]} &&
\|\tilde{P} - V(P\otimes P_\anc)V^\dag \| \le \epsilon &&
\phantom{\textnormal{(incoherence)}}
\end{flalign}
When $P$ projects onto the groundspace of $H$, rather than onto a
quasi-groundspace,
we usually omit $P$ and simply
say $\tilde{H}$ \emph{weakly gap-simulates} $H$.

\end{defn}

\noindent The only difference between Definition~\ref{defn:hamsimul} and \ref{defn:weaksimul} is that we replaced Eq.~\eqref{eq:strongsimul} with Eq.~\eqref{eq:weaksimul}.
Correspondingly, any degree-reducer (diluter) that only weakly gap-simulates is called a \emph{weak} degree-reducer (diluter).

\begin{figure}
\centering
\includegraphics[height=5cm]{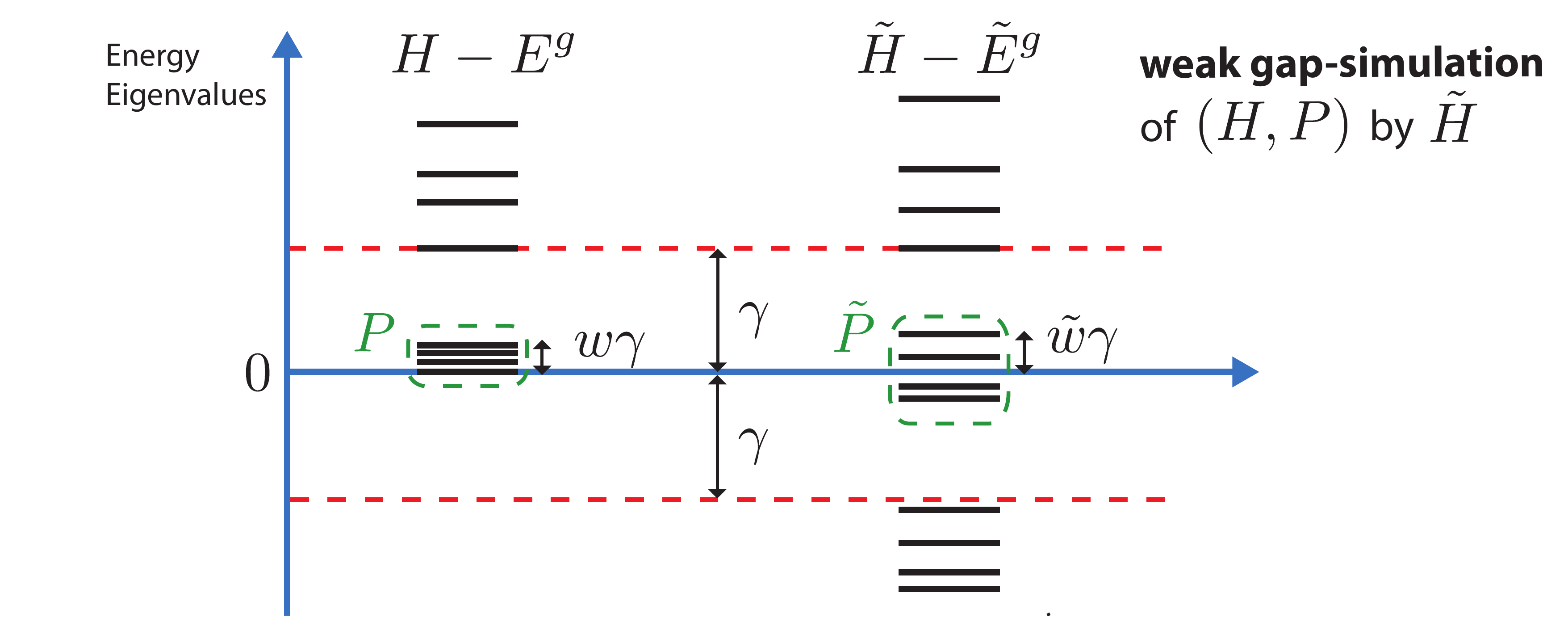}
\caption{\label{fig:weakgapsimul}Visualizing weak gap-simulation of $H$ by $\tilde{H}$.}
\end{figure}

\begin{prop}[weak star-graph diluter of $H_\oneone$]
\label{prop:star}
There is a 2-local, weak $[n, O(n), O(1/\epsilon^2)]$-diluter of $H_\oneone$ with $\epsilon$-incoherence, energy spread $\tilde{w}=\O(\epsilon^2)$ and trivial encoding $V=\Id$, using one additional ancilla qubit that interacts with all original qubits in a star-graph geometry.
\end{prop}

\begin{proof}
Let us denote the number operator $\hat{N}_e=\sum_{i=1}^n \ketbra{1}^{(i)}$. Note
\begin{equation}
H_\oneone = -\hat{N}_e+\hat{N}_e^2
\end{equation}
Consider the following Hamiltonian on $n$ system qubits and 1 ancilla qubit:
\begin{equation}
\tilde{H}_\oneone^\text{star} = -\Delta \ketbra{0}^\anc  - \hat{N}_e + \sqrt{\Delta}\hat{N}_e\otimes \sigma_x^\anc
\end{equation}
Note when expanded in terms of 2-local terms, $\tilde{H}_\oneone^\text{star}$ has $O(n)$ terms, and has maximum degree $n$ at the ancilla qubit.
Since $[\hat{N}_e, \tilde{H}_\oneone^\text{star}]=0$, we can replace $\hat{N}_e$ with its eigenvalue $m$ in $\tilde{H}_\oneone^\text{star}$, where $m=0,1,\ldots,n$. The Hamiltonian can now be rewritten as
\begin{equation}
\tilde{H}_\oneone^\text{star} = -\Delta/2-m + (m\sqrt{\Delta},0,-\Delta/2)\cdot \vec\sigma^\anc
\end{equation}
The energy eigenvalues are
\begin{eqnarray}
E_m^\pm = -\Delta/2-m \pm \sqrt{m^2\Delta + \Delta^2/4} =
\begin{cases}
-m+m^2-O(m^4/\Delta)\\
-\Delta-m-m^2+O(m^4/\Delta)
\end{cases}
\end{eqnarray}
The corresponding eigenstates are
\begin{eqnarray}
\ket{E_m^+} = \ket{m_e}\otimes(\cos\frac{\theta_m}{2}\ket{1}+\sin\frac{\theta_m}{2}\ket{0})_\anc \equiv \ket{m_e}\otimes \ket{\theta_m^+}\\
\ket{E_m^-} = \ket{m_e}\otimes(\sin\frac{\theta_m}{2}\ket{1}-\cos\frac{\theta_m}{2}\ket{0})_\anc \equiv \ket{m_e}\otimes \ket{\theta_m^-}
\end{eqnarray}
where $\ket{m_e}$ is any eigenvector of $\hat{N}_e$ with eigenvalue $m$, and $\tan\theta_m=m\sqrt{\Delta}/(\Delta/2)=2m/\sqrt{\Delta}$. In the small $m$ sector
\begin{eqnarray*}
E^+_0 &=& 0\\
E^+_1 &=& - 1/\Delta + O(1 /\Delta^2) \\
E^+_2 &=& 2 - 16/\Delta + O(1/\Delta^2)
\end{eqnarray*}
So $m=0,1$ states are quasi-degenerate up to $1/\Delta$, separated by a gap of $\gamma\approx2$ from $m=2$ states.
More explicitly, the states
\begin{eqnarray*}
\ket{E_0^+} &=& \ket{0_e}\otimes\ket{1}_\anc \\
\ket{E_1^+} &=& \ket{1_e}\otimes(\cos\frac{\theta_1}{2}\ket{1}+ \sin\frac{\theta_1}{2}\ket{0})_\anc
\end{eqnarray*}
spans a quasi-degenerate ``groundspace'' of zero energy.
From below, these states are gapped by $\Delta$ from the $\ket{E_m^-}$ states.
From above, these states are gapped by $E_2^+\simeq 2$ from states like
\begin{equation}
\ket{E_2^+} = \ket{2_e}\otimes(\cos\frac{\theta_2}{2}\ket{1}+\sin\frac{\theta_2}{2}\ket{0})_\anc
\end{equation}
The projector onto the ``groundspace'' of $\tilde{H}_\oneone^\text{star}$ is
\begin{equation}
\tilde{P} = P_{0e}\otimes\ketbra{1}_\anc + P_{1e} \otimes \ketbra{\theta_1^+}_\anc
\end{equation}
where $P_{me}$ projects onto $\{\ket{m_e}:\hat{N}_e\ket{m_e}=m\ket{m_e}\}$, the set of states containing $m$ excitations.
The original Hamiltonian's groundspace projector is $P=P_{0e} + P_{1e}$.
We can guess $P_\anc = \ketbra{1}$, and thus
\begin{eqnarray*}
\tilde{P}-P\otimes P_\anc &=& P_{1e} \otimes \left(\ketbra{\theta_1^+} - \ketbra{1}\right) \\
&=& P_{1e}\otimes \left(\sin^2\frac{\theta_1}{2}(\ketbra{0}-\ketbra{1}) + \frac{1}{2}\sin\theta_1(\ketbrat{0}{1}+\ketbrat{1}{0})\right)
\end{eqnarray*}
The incoherence of the $m=0,1$ ``groundspace'' is
\begin{equation}
\|\tilde{P}-P\otimes P_\anc\| \le \sin\theta_1 = \sqrt{\frac{2/\Delta}{1+4/\Delta+\sqrt{1+4/\Delta}}} \approx \sqrt{\frac{2}{\Delta}}
\end{equation}
Hence, by choosing $\Delta=O(1/\epsilon^2)$, $\tilde{H}_\oneone^\text{star}$ is a 2-local weak $[n,O(n),O(1/\epsilon^2)]$-sparsifier of $H_\oneone$ with $\epsilon$-incoherence and energy spread $\tilde{w}=1/\Delta=\O(\epsilon^2)$.
\end{proof}

\end{appendices}

\bibliographystyle{hieeetr}
\bibliography{sparsify}

\begin{thebibliography}{10}

\bibitem{Feynman1982}
R.~P. Feynman, ``{Simulating Physics with Computers},'' Tech. Rep.~6, 1982.

\bibitem{CiracZollerNatPhys2012}
J.~I. Cirac and P.~Zoller, ``Goals and opportunities in quantum simulation,''
  {\em Nature Physics}, vol.~8, pp.~264--266, Apr 2012.

\bibitem{AharonovFaultTolerance}
D.~Aharonov and M.~Ben-Or, ``Fault-tolerant quantum computation with constant
  error,'' in {\em Proceedings of the Twenty-ninth Annual ACM Symposium on
  Theory of Computing}, STOC '97, (New York, NY, USA), pp.~176--188, ACM, 1997.

\bibitem{KnillThreshold}
E.~Knill, R.~Laflamme, and W.~H. Zurek, ``{Resilient quantum computation: error
  models and thresholds},'' {\em Proceedings of the Royal Society A:
  Mathematical, Physical and Engineering Sciences}, vol.~454, pp.~365--384, Jan
  1998.

\bibitem{KitaevFaultTolerantAnyon}
A.~Kitaev, ``{Fault-tolerant quantum computation by anyons},'' {\em Annals of
  Physics}, vol.~303, pp.~2--30, Jan 2003.

\bibitem{LloydDigitalSimulation1996}
S.~Lloyd, ``{Universal Quantum Simulators},'' {\em Science (New York, N.Y.)},
  vol.~273, pp.~1073--8, Aug 1996.

\bibitem{SimonOpticalLatticeSim2011}
J.~Simon, W.~S. Bakr, R.~Ma, M.~E. Tai, P.~M. Preiss, and M.~Greiner,
  ``{Quantum simulation of antiferromagnetic spin chains in an optical
  lattice},'' {\em Nature}, vol.~472, pp.~307--312, Apr 2011.

\bibitem{Bloch2012}
I.~Bloch, J.~Dalibard, and S.~Nascimb{\`{e}}ne, ``{Quantum simulations with
  ultracold quantum gases},'' {\em Nature Physics}, vol.~8, pp.~267--276, Apr
  2012.

\bibitem{Blatt2012}
R.~Blatt and C.~F. Roos, ``{Quantum simulations with trapped ions},'' {\em
  Nature Physics}, vol.~8, pp.~277--284, Apr 2012.

\bibitem{Aspuru-Guzik2012}
A.~Aspuru-Guzik and P.~Walther, ``{Photonic quantum simulators},'' {\em Nature
  Physics}, vol.~8, pp.~285--291, Apr 2012.

\bibitem{Houck2012}
A.~A. Houck, H.~E. T{\"{u}}reci, and J.~Koch, ``{On-chip quantum simulation
  with superconducting circuits},'' {\em Nature Physics}, vol.~8, pp.~292--299,
  Apr 2012.

\bibitem{GeorgescuQuantumSimulationReview2014}
I.~M. Georgescu, S.~Ashhab, and F.~Nori, ``Quantum simulation,'' {\em Rev. Mod.
  Phys.}, vol.~86, pp.~153--185, Mar 2014.

\bibitem{Bernien2017}
H.~Bernien, S.~Schwartz, A.~Keesling, H.~Levine, A.~Omran, H.~Pichler, S.~Choi,
  A.~S. Zibrov, M.~Endres, M.~Greiner, V.~Vuleti{\'{c}}, and M.~D. Lukin,
  ``{Probing many-body dynamics on a 51-atom quantum simulator},'' {\em
  Nature}, vol.~551, pp.~579--584, Nov 2017.

\bibitem{Zhang2017}
J.~Zhang, G.~Pagano, P.~W. Hess, A.~Kyprianidis, P.~Becker, H.~Kaplan, A.~V.
  Gorshkov, Z.-X. Gong, and C.~Monroe, ``{Observation of a many-body dynamical
  phase transition with a 53-qubit quantum simulator},'' {\em Nature},
  vol.~551, pp.~601--604, Nov 2017.

\bibitem{Preskill2018}
J.~Preskill, ``{Quantum Computing in the NISQ era and beyond},'' {\em Quantum},
  vol.~2, p.~79, Jan 2018, 1801.00862.

\bibitem{KSV02}
A.~Y. Kitaev, A.~Shen, and M.~N. Vyalyi, {\em Classical and Quantum
  Computation}.
\newblock American Mathematical Society, 2002.

\bibitem{quantumNPsurvey}
D.~Aharonov and T.~Naveh, ``{Quantum NP - A Survey},'' 2002, quant-ph/0210077.

\bibitem{BravyiHastingsSim}
S.~Bravyi and M.~Hastings, ``{On complexity of the quantum Ising model},'' Oct
  2014, 1410.0703.

\bibitem{UniversalHamiltonian}
T.~Cubitt, A.~Montanaro, and S.~Piddock, ``{Universal Quantum Hamiltonians},''
  2017, 1701.05182.

\bibitem{KKR06}
J.~Kempe, A.~Kitaev, and O.~Regev, ``{The Complexity of the Local Hamiltonian
  Problem},'' {\em SIAM J. Comput.}, vol.~35, pp.~1070--1097, 2006.

\bibitem{OliveiraTerhal}
R.~Oliveira and B.~M. Terhal, ``{The complexity of quantum spin systems on a
  two-dimensional square lattice},'' {\em Quantum Inf. Comput.}, vol.~8,
  pp.~900--924, 2008.

\bibitem{BDLT08}
S.~Bravyi, D.~P. DiVincenzo, D.~Loss, and B.~M. Terhal, ``{Quantum Simulation
  of Many-Body Hamiltonians Using Perturbation Theory with Bounded-Strength
  Interactions},'' {\em Phys. Rev. Lett.}, vol.~101, p.~070503, Aug 2008.

\bibitem{JordanGadgets}
S.~P. Jordan and E.~Farhi, ``Perturbative gadgets at arbitrary orders,'' {\em
  Phys. Rev. A}, vol.~77, p.~062329, Jun 2008.

\bibitem{CaoImprovedOTGadget}
Y.~Cao, R.~Babbush, J.~Biamonte, and S.~Kais, ``{Hamiltonian gadgets with
  reduced resource requirements},'' {\em Phys. Rev. A}, vol.~91, p.~012315, Jan
  2015.

\bibitem{CaoNagajGadget}
Y.~Cao and D.~Nagaj, ``{Perturbative gadgets without strong interactions},''
  {\em Quantum Inf. Comput.}, vol.~15, pp.~1197--1222, 2015.

\bibitem{FarhiAdiabatic2000}
E.~Farhi, J.~Goldstone, S.~Gutmann, and M.~Sipser, ``{Quantum Computation by
  Adiabatic Evolution},'' Jan 2000, 0001106.

\bibitem{QAOA}
E.~Farhi and A.~W. Harrow, ``{Quantum Supremacy through the Quantum Approximate
  Optimization Algorithm},'' pp.~1--22, 2016, 1602.07674.

\bibitem{qPCPsurvey}
D.~Aharonov, I.~Arad, and T.~Vidick, ``{The Quantum PCP Conjecture},'' {\em ACM
  SIGACT News}, vol.~44, pp.~47--79, 2013.

\bibitem{LocalTestOfEntanglement}
D.~Aharonov, A.~W. Harrow, Z.~Landau, D.~Nagaj, M.~Szegedy, and U.~Vazirani,
  ``{Local Tests of Global Entanglement and a Counterexample to the Generalized
  Area Law},'' in {\em 2014 IEEE 55th Annual Symposium on Foundations of
  Computer Science}, pp.~246--255, IEEE, Oct 2014.

\bibitem{STNearlyLinearTimeAlg}
D.~A. Spielman and S.-H. Teng, ``{Nearly-Linear Time Algorithms for
  Preconditioning and Solving Symmetric, Diagonally Dominant Linear Systems},''
  Jul 2006, 0607105.

\bibitem{SpielmanSpectralSparsify}
D.~A. Spielman and S.-H. Teng, ``{Spectral Sparsification of Graphs},'' {\em
  SIAM J. Comput.}, vol.~40, pp.~981--1025, Jan 2011.

\bibitem{SpielmanEffectiveResistance}
D.~A. Spielman and N.~Srivastava, ``{Graph Sparsification by Effective
  Resistances},'' {\em SIAM J. Comput.}, vol.~40, pp.~1913--1926, Jan 2011.

\bibitem{BSST13}
J.~Batson, D.~A. Spielman, N.~Srivastava, and S.-H. Teng, ``{Spectral
  Sparsification of Graphs},'' {\em Commun. ACM}, vol.~56, pp.~87--94, 2013.

\bibitem{BSS12}
J.~Batson, D.~A. Spielman, and N.~Srivastava, ``{Twice-Ramanujan
  Sparsifiers},'' {\em SIAM J. Comput.}, vol.~41, pp.~1704--–1721, 2012.

\bibitem{dinur}
I.~Dinur, ``{The PCP Theorem by Gap Amplification},'' {\em J. ACM}, vol.~54,
  June 2007.

\bibitem{Wen1990}
X.~G. Wen, ``Topological order in rigid states,'' {\em International Journal of
  Modern Physics B}, vol.~04, pp.~239--271, Feb 1990.

\bibitem{LevinWenTopo}
M.~Levin and X.-G. Wen, ``Detecting topological order in a ground state wave
  function,'' {\em Phys. Rev. Lett.}, vol.~96, p.~110405, Mar 2006.

\bibitem{SachdevQPT}
S.~Sachdev, {\em {Quantum phase transitions}}.
\newblock Cambridge University Press, 2011.

\bibitem{UndecidabilityOfGap}
T.~S. Cubitt, D.~Perez-Garcia, and M.~M. Wolf, ``{Undecidability of the
  spectral gap},'' {\em Nature}, vol.~528, pp.~207--211, Dec 2015.

\bibitem{HastingsKoma}
M.~B. Hastings and T.~Koma, ``{Spectral Gap and Exponential Decay of
  Correlations},'' {\em Commun. Math. Phys.}, vol.~265, pp.~781--804, 2006.

\bibitem{DellvanMelkebeek}
H.~Dell and D.~van Melkebeek, ``Satisfiability allows no nontrivial
  sparsification unless the polynomial-time hierarchy collapses,'' in {\em
  Proceedings of the Forty-second ACM Symposium on Theory of Computing}, STOC
  '10, (New York, NY, USA), pp.~251--260, ACM, 2010.

\bibitem{Dicke}
R.~H. Dicke, ``Coherence in spontaneous radiation processes,'' {\em Phys.
  Rev.}, vol.~93, pp.~99--110, Jan 1954.

\bibitem{BravyiVyalyi}
S.~Bravyi and M.~Vyalyi, ``Commutative version of the local hamiltonian problem
  and common eigenspace problem,'' {\em Quantum Info. Comput.}, vol.~5,
  pp.~187--215, May 2005.

\bibitem{AharonovEldar2011}
D.~Aharonov and L.~Eldar, ``{On the Complexity of Commuting Local Hamiltonians,
  and Tight Conditions for Topological Order in Such Systems},'' in {\em 2011
  IEEE 52nd Annual Symposium on Foundations of Computer Science}, pp.~334--343,
  IEEE, Oct 2011.

\bibitem{qPCPArad}
I.~Arad, ``A note about a partial no-go theorem for quantum pcp,'' {\em Quantum
  Info. Comput.}, vol.~11, pp.~1019--1027, Nov. 2011.

\bibitem{BrandaoHarrow}
F.~G. Brand{\~a}o and A.~W. Harrow, ``{Product-state Approximations to Quantum
  Ground States},'' {\em Comm. Math. Phys.}, vol.~342, pp.~47--80, 2016.

\bibitem{qPCPHastings}
M.~B. Hastings, ``Trivial low energy states for commuting hamiltonians, and the
  quantum pcp conjecture,'' {\em Quantum Info. Comput.}, vol.~13, pp.~393--429,
  May 2013.

\bibitem{AharonovEldar2013}
D.~Aharonov and L.~Eldar, ``{The commuting local Hamiltonian on
  locally-expanding graphs is in NP},'' Nov 2013, 1311.7378.

\bibitem{BenSassonPCP}
E.~Ben-Sasson, O.~Goldreich, P.~Harsha, M.~Sudan, and S.~Vadhan, ``{Robust PCPs
  of Proximity, Shorter PCPs, and Applications to Coding},'' {\em SIAM J.
  Comput.}, vol.~36, pp.~889--974, Dec. 2006.

\bibitem{dinurgoldreich}
I.~Dinur, O.~Goldreich, and T.~Gur, ``Every set in $\mathcal{P}$ is strongly
  testable under a suitable encoding,'' {\em Electronic Colloquium on
  Computational Complexity (ECCC)}, vol.~25, p.~50, 2018.

\bibitem{NielsenChuang}
M.~A. Nielsen and I.~L. Chuang, {\em {Quantum Computation and Quantum
  Information}}.
\newblock New York, NY, USA: Cambridge University Press, 10th~ed., 2011.

\bibitem{LiebRobinson}
E.~H. Lieb and D.~W. Robinson, ``{The finite group velocity of quantum spin
  systems},'' {\em Communications in Mathematical Physics}, vol.~28,
  pp.~251--257, Sep 1972.

\bibitem{FarhiQAAFail}
E.~Farhi, J.~Goldstone, S.~Gutman, and D.~Nagaj, ``How to make the quantum
  adiabatic algorithm fail,'' {\em International Journal of Quantum
  Information}, vol.~6, pp.~503--516, 2008.

\bibitem{DicksonAmin}
N.~G. Dickson and M.~H. Amin, ``Algorithmic approach to adiabatic quantum
  optimization,'' {\em Phys. Rev. A}, vol.~85, p.~032303, Mar 2012.

\bibitem{AharonovAtia}
Y.~Atia and D.~Aharonov, 2018.
\newblock in preparation.

\bibitem{SilvaSparsification}
M.~K. D.~C. Silva, N.~J.~A. Harvey, and C.~M. Sato, ``Sparse sums of positive
  semidefinite matrices,'' {\em ACM Trans. Algorithms}, vol.~12, pp.~9:1--9:17,
  Dec. 2015.

\bibitem{AhlswedeWinter}
R.~Ahlswede and A.~Winter, ``{Strong converse for identification via quantum
  channels},'' {\em IEEE Transactions on Information Theory}, vol.~48,
  pp.~569--579, Mar 2002.

\bibitem{LechnerHaukeZoller}
W.~Lechner, P.~Hauke, and P.~Zoller, ``A quantum annealing architecture with
  all-to-all connectivity from local interactions,'' {\em Science Advances},
  vol.~1, no.~9, 2015.

\bibitem{AharonovAQCUniversal}
D.~Aharonov, W.~van Dam, J.~Kempe, Z.~Landau, S.~Lloyd, and O.~Regev,
  ``Adiabatic quantum computation is equivalent to standard quantum
  computation,'' {\em SIAM J. Comput.}, vol.~37, pp.~166--194, Apr. 2007.

\bibitem{RydbergBlockade}
M.~Saffman, T.~G. Walker, and K.~M\o{}lmer, ``{Quantum information with Rydberg
  atoms},'' {\em Rev. Mod. Phys.}, vol.~82, pp.~2313--2363, Aug 2010.

\bibitem{floquet}
M.~Grifoni and P.~H{\"{a}}nggi, ``{Driven quantum tunneling},'' {\em Physics
  Reports}, vol.~304, pp.~229--354, Oct 1998.

\bibitem{FloquetEngineering}
G.~Harel and V.~M. Akulin, ``Complete control of hamiltonian quantum systems:
  Engineering of floquet evolution,'' {\em Phys. Rev. Lett.}, vol.~82,
  pp.~1--5, Jan 1999.

\bibitem{LindnerFloquetTopo}
N.~H. Lindner, G.~Refael, and V.~Galitski, ``{Floquet topological insulator in
  semiconductor quantum wells},'' {\em Nature Physics}, vol.~7, pp.~490--495,
  Jun 2011.

\bibitem{ChoiDynamicalEngineering}
S.~Choi, N.~Y. Yao, and M.~D. Lukin, ``Dynamical engineering of interactions in
  qudit ensembles,'' {\em Phys. Rev. Lett.}, vol.~119, p.~183603, Nov 2017.

\end{thebibliography}

\end{document}